%% file: asymptotics-arxiv.tex
\documentclass[11pt]{article}
\usepackage{amsmath}
\usepackage{amsfonts}
\usepackage{amssymb}
\usepackage{natbib}
\usepackage{enumerate}
\usepackage[bottom]{footmisc}
\usepackage{graphicx}
\usepackage{float}
\usepackage{epsfig}
\usepackage{bm}
\usepackage{subfigure}
\input{tcilatex.tex}
\input{defs2.tex}

\makeatletter
\long\def\@footnotetext#1{\insert\footins{\def\baselinestretch{1.2}\footnotesize
\interlinepenalty\interfootnotelinepenalty
\splittopskip\footnotesep \splitmaxdepth \dp\strutbox
\floatingpenalty \@MM \hsize\columnwidth \@parboxrestore
\edef\@currentlabel{\csname
p@footnote\endcsname\@thefnmark}\@makefntext
{\rule{\z@}{\footnotesep}\ignorespaces #1\strut}}}

\newcommand \address[1]{\gdef \@address{#1}}

\def\maketitle{%
  \null
  \thispagestyle{empty}%
  \begin{center}\leavevmode
    \normalfont
    {\LARGE \bf \@title\par}%
    {\large \@author\par}%
    \vskip 0.05 cm
    {\large \rm \@address\par}%
    \vskip 0.05cm
    {\large \@date\par}%
  \end{center}%
}

\makeatother

\begin{document}

\title{Non-Regular Likelihood Inference for Seasonally Persistent Processes}
\author{Emma J. McCoy$^{\text{(1)}}$ \\ Sofia C. Olhede$^{\text{(1)}}$ \\ David A.
Stephens${^\text{(1,2)}}$}
\address{(1) Department of Mathematics, Imperial College London\\ 180 Queens Gate,\\ London SW7 2AZ,\\ UK
\\ and \\ (2) Department of Mathematics and Statistics, McGill University\\ 805
Sherbrooke St. W\\ Montreal, QC, H2K 2K2,\\ Canada}

\maketitle

\begin{abstract}
The estimation of parameters in the frequency spectrum of a seasonally
persistent stationary stochastic process is addressed. For seasonal persistence
associated with a pole in the spectrum located away from frequency zero, a new
Whittle-type likelihood is developed that explicitly acknowledges the location
of the pole. This Whittle likelihood is a large sample approximation to the
distribution of the periodogram over a chosen grid of frequencies, and
constitutes an approximation to the time-domain likelihood of the data, via the
linear transformation of an inverse discrete Fourier transform combined with a
demodulation. The new likelihood is straightforward to compute, and as will be
demonstrated has good, yet non-standard, properties. The asymptotic behaviour
of the proposed likelihood estimators is studied; in particular,
$N$-consistency of the estimator of the spectral pole location is established.
Large finite sample and asymptotic distributions of the score and observed
Fisher information are given, and the corresponding distributions of the
maximum likelihood estimators are deduced. Asymptotically, the estimator of the
pole after suitable standardization follows a Cauchy distribution, and for
moderate sample sizes, we can use the finite large sample approximation to the
distribution of the estimator of the pole corresponding to the ratio of two
Gaussian random variables, with sample size dependent means and variances. A
study of the small sample properties of the likelihood approximation is
provided, and its superior performance to previously suggested methods is
shown, as well as agreement with the developed distributional approximations.
Inspired by the developments for full likelihood based estimation procedures,
usage of profile likelihood and other likelihood based procedures are also
discussed.  Semi-parametric estimation methods, such as the Geweke-Porter-Hudak
estimator of the long memory parameter, inspired by the developed parametric
theory are introduced.

\vspace{0.05 in}


\end{abstract}

\section{Introduction}
In this paper, we develop likelihood estimation of the parameters of a
stationary stochastic process that exhibits \emph{seasonal persistence}, that
is, long memory behaviour associated with a stationary, quasi-seasonal
dependence structure. We introduce a new frequency-domain likelihood
approximation which is computed using demodulation and which, for the first
time, facilitates maximum likelihood estimation. We consider joint estimation
of the seasonality and persistence parameters, and establish the asymptotic and
large sample properties of the likelihood and its associated maximum likelihood
estimators. This is in direct contrast with previously suggested procedures,
where the distribution of the estimator of the seasonality parameter could not
be established \citep{Giraitis2001}.  The estimators are demonstrated to have
good small sample properties compared with estimators based on the classic
Whittle likelihood, and other non-likelihood derived estimators.  Our
non-standard asymptotic results rely on the appropriate renormalization of the
score and Fisher information, and utilize a parameter-dependent linear
transformation of the data. This transformation enables an efficient
approximation to the likelihood.  The transformation also introduces a number
of interesting and non-regular features into the likelihood surface: jumps,
local oscillations, and non-regular large sample theory. Despite these issues
the large sample theory can be determined, and appropriate finite large sample
approximations provided, as will be demonstrated. It transpires that the small
sample properties of the estimators are competitive with existing methods, as
well be discussed in later sections.

The contributions of this paper thus include new theory for non-regular maximum
likelihood problems. In similarly motivated work, \citet{cheng} discussed
problems associated with maximum likelihood estimation for unbounded
likelihoods: in contrast we discuss problems associated with distributions of
non-identically distributed, weakly dependent variables with highly compressed
and for increasing sample sizes unbounded variances. Given the importance of
compressed linear decompositions in modern statistical theory, our work has
implications for the distribution of sparseness-inducing transformations much
beyond the analysis of seasonal processes and Fourier theory, and forms a
contribution to developing methodology for inference of stochastically
compressible processes.

One of the concrete and substantive conclusions of our new estimation
procedures is illustrated in Figure \ref{simsplot};  this figure illustrates
that whereas a standard estimation procedure, based on the Whittle likelihood
(see Section \ref{PeriodogramSect}), produces estimates that are, on average,
biased even in large samples, our new procedure, based on a carefully
constructed likelihood (see Sections \ref{DDFTsec} and
\ref{Likelihoodsection}), produces estimators that exhibit no such bias.  Full
details of this Figure are given in Section \ref{SimExamp}.

\begin{figure}
\begin{center}
\includegraphics[height=4.7in, width=4.7in]{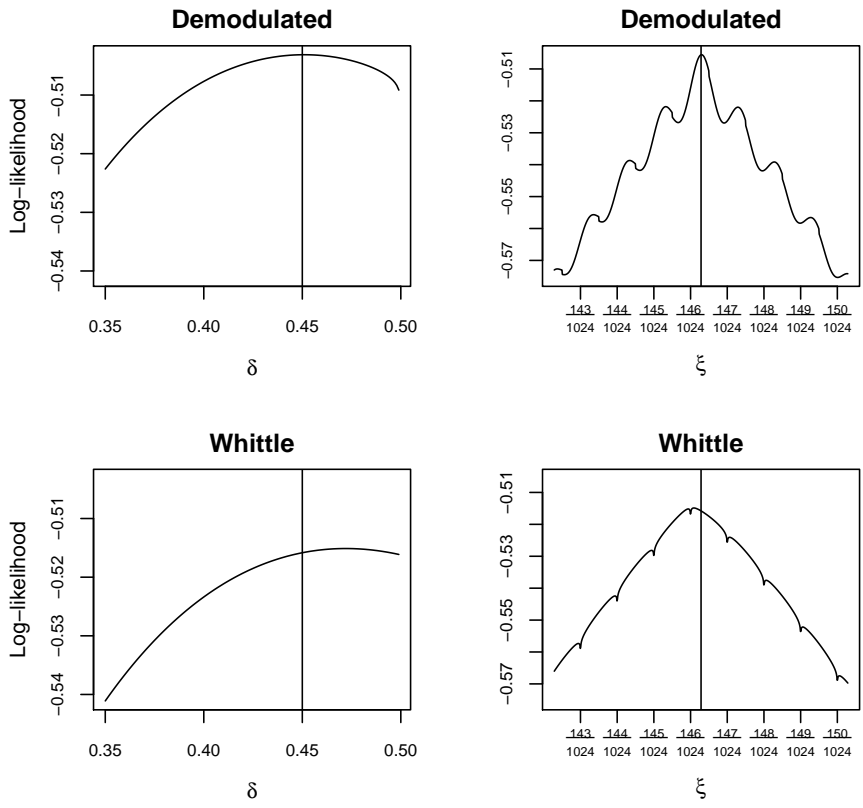}
\caption{\label{simsplot}Simulated Data:
Mean standardized likelihoods for the pole (right) and the long memory parameter (left) over 2000
simulations, with sample size of 1024, and the true values of the long memory
parameter and the pole taking the values 0.45 and 1/7, respectively.
The vertical solid lines indicate the true values of the parameters. The Demodulated
likelihood is noted in equation (\ref{approxlikeeqn}) whilst the discrete Whittle likelihood is
noted in equation (\ref{discfourlik}). On average, the demodulated likelihood has its mode
at the true values, whereas the Whittle likelihood does not.
See Section \ref{SimExamp} for full details.}
\end{center}
\end{figure}

\subsection{Seasonally Persistent Processes}

\label{timeseriesdef} Stationary time-series models with long range dependence
describe a wide range of physical phenomena; see for general discussion
\citet{Andel1986} and \citet{Gray1989}, and also applications in econometrics
\citep{Porter-Hudak1990,Gil-Alana2002}, biology \citep{Beran1994} and hydrology
\citep{Ooms2001}. Dependence in a stationary time series is parameterized via
the autocovariance sequence, $\{\gamma_{\tau}\}$. We are concerned with the
estimation of parameters that specify $\gamma_{\tau}$ under an assumption of
\textit{seasonal persistence}.  Specifically, of particular importance is the
{\em seasonality} of the data characterized by a frequency, $\pole$, termed the
\textit{pole}, and an associated {\em degree of dependence}, characterized by a
\textit{persistence} (or \textit{long memory}) parameter $\delta$. Whereas
inference for the persistence parameter in the context of poles at frequency
zero has been much studied \citep{Beran1994}, the theoretical behaviour of
estimators of the persistence parameter remains largely uninvestigated when the
underlying seasonality of the process is unknown.

Let $\Xt$ be a zero-mean, second-order stationary time series with
autocovariance (acv) sequence
$\gamma_{\tau}=\textrm{cov}\left\{X_t,X_{t+\tau}\right\}={\mathrm{E}}\left\{X_t
X_{t+\tau}\right\}$, and spectral density function (sdf), $f(\cdot)$,
\begin{equation}
\label{SDF} f\left(\lambda\right)=\sum_{\tau=-\infty}^{\infty}
\gamma_{\tau} e^{-2i\pi \lambda \tau}.
\end{equation}
The process $\Xt$ exhibits \emph{seasonal} or \emph{periodic} persistence if
there exist real numbers $H\in(1/2,1)$ and $\pole\in\left(0,1/2\right)$, and a
bounded function $c(\gamma)$ such that
\[\lim_{\tau \rightarrow \infty} \frac{\gamma_{\tau}}{c(\gamma)\left|\tau\right|^{2H-2}}
=\cos\left(2\pi \pole \tau\right),\] or equivalently if there exist
$\beta\in(0,1)$ and $\pole\in\left(0,\frac{1}{2}\right)$ and a bounded function
$c(\lambda)$ such that
\[
\lim_{\lambda \rightarrow \pole}
\frac{f(\lambda)\left|\lambda-\pole\right|^{\beta}} {c(\lambda)} =1.
\]
Following convention, we parameterize the persistence parameter via $\delta =
\beta /2$. In line with this definition, a process is considered to be a
\textit{seasonally persistent process} (SPP) if, in a neighbourhood of $\pole$,
\begin{equation}
\label{sdfekvation} f\left(\lambda\right) = f^{\dagger}
\left(\lambda\right)\left|\lambda-\pole\right|^{-\beta} +
\o(\left|\lambda-\pole\right|^{-\beta}),
\end{equation}
where $f^{\dagger}(\lambda) \equiv c(\lambda) > 0$, $0< \lambda <
\frac{1}{2}$ is bounded above.

Parameters $(\pole,\delta)$ determine the dominant long term behaviour of the
process; typically, $\pole$ corresponds to the location of an unbounded but integrable
singularity in the sdf.  In this paper we consider a parametric family of sdfs
consistent with (\ref{sdfekvation}), that is, the parametric model of
\citet{Giraitis2001}, where
\begin{equation}
\label{giraitis} f(\lambda)
=f_G(\lambda;\pole,\delta,\bm{\theta},\sigma_{\epsilon}^2)=\sigma^2_{\epsilon}
|h(\lambda\,;\,\bm{\theta})|^2 (1-2e^{-2i\pi\lambda}\cos(2\pi
\pole)+e^{-4i\pi\lambda} )^{-2\delta},
\end{equation}
where $h(\lambda\,;\,\bm{\theta})$ is bounded above and below at
$\lambda=\pole$, with some linear process assumptions, given for instance in
\citet{Hannan1973}; for example, $h(\cdot)$ could be the sdf for a stationary
and invertible ARMA process, such is the case for GARMA processes, see
\citet{Gray1989}. We consider behaviour near the pole in such models by
defining $f^{\dagger}(\lambda)$, where
\begin{eqnarray}
\label{sdfwithpole}
f\left(\lambda\right) = f^{\dagger}(\lambda)
\left|\lambda-\pole\right|^{-2\delta} = f^{\dagger}(\lambda;
\pole,\delta, \bm{\theta}) \left|\lambda-\pole\right|^{-2\delta}.
\end{eqnarray}
The results in this paper will also be applicable to nearly non-stationary unit
root AR processes, when the roots of the AR process approach unity at a
suitable rate in the sample size, this quantifying issues with near unit root
processes.

\subsection{Estimation for Seasonally Persistent Processes}

We consider maximum likelihood estimation of $\pole$ and $\delta$, and denote
the true values of these parameters by $(\polestar,\deltastar)$.  Joint
estimation of the seasonality and persistence parameters is of importance, as
inaccurate estimation of $\pole$ will affect the estimation of $\delta$, and
any other parameters of the sdf -- $\delta$ quantifies the rate of decay of the
dependence, and thus determines the long-term behaviour of the series. Note
also that, even in cases where $\pole$ is believed to be \textbf{known} (for
calendar data, equal to 1/12, or 1/7, or 1/4 say), there may on occasion be
finite sample advantage in estimating $\pole$ rather than using its known
value, in terms of estimation of the other parameters of the system.  For
example, if $\pole$ is regarded as a nuisance parameter, then $\delta$ may be
more efficiently estimated after conditioning on $\polehat$ rather than
$\polestar$; see, for example, \citet{Robins1994} and \citet{Rathouz2002} for
supporting theory.  This issue goes beyond the scope of this paper, but gives
further indication that estimation of $\pole$ is intrinsically important.

We will examine inference for the parameters of an SPP based on a realization
of the process of length $N$. Throughout this paper, for convenience and with
minimal loss of generality, we will assume $N$ is even, $N=2M$ say. We
establish asymptotic results for these estimators $(\polehat,\deltahat)$, and
provide practically useful large sample approximations to the distribution of
the estimators. In particular, we define a large sample approximation to the
log-likelihood of the periodogram evaluated at a full set of frequencies spaced
$\O(N^{-1})$ apart. At a local scale the variational structure of the
log-likelihood in $\pole$ remains appreciable over $\O(N^{-1})$ distances;
however the magnitude of these variations becomes negligible compared to the
total accumulated magnitude of the log-likelihood for increasing sample sizes.
We demonstrate that this variation prevents standard likelihood results being
valid for the estimator of $\pole$, although standard asymptotic results can be
established for the estimator of the $\delta$, which is in agreement with
previous results, see \citep{HandS2004,Giraitis2001}. We discuss in detail the
large sample behaviour of $N(\polehat-\polestar)$ and establish its approximate
large sample distribution, as well as a moderate sample size approximation.
Finally we demonstrate that our likelihood-based estimators have good small
sample properties on simulated series compared with other, non-likelihood
estimators, and consider estimation of the system parameters in a econometric
example, using a data set with weekly gasoline sales in the United States, and
two meteorological examples, monthly temperature data from a Californian
shore-station, and the Southern Oscillation Index data set.

\subsection{The Periodogram, Likelihoods and Approximations}
\label{PeriodogramSect} We consider a sample from a stationary Gaussian time
series, $\bm{X}=\left(X_0,X_1,\dots,X_{N-1}\right)^\transpose$, as defined in
section \ref{timeseriesdef}, with covariance matrix ${\mathcal{G}}_N =
{\mathcal{G}}_N(\pole,\delta,\bm{\theta},\sigma^2_{\epsilon})$ with
$(i,j)^\textrm{th}$ element $\gamma_{\left|i-j\right|}$. The exact log
likelihood, $\loglike_N$, of the finite time-domain sample is given by
\begin{equation}
2 \loglike_N \left(\pole,\delta,\bm{\theta},\sigma^2_{\epsilon}\right)= 2 \log L_N
\left(\pole,\delta,\bm{\theta},\sigma^2_{\epsilon}\right) = - N\log
\left(2\pi\right)-\log\left|{\mathcal{G}}_N\right|
-\bm{X}^\transpose{\mathcal{G}}_N^{-1}\bm{X}. \label{ExactTimeLike}
\end{equation}
This likelihood is often approximated due to the computational complexity
associated with the calculation of ${\mathcal{G}}_N^{-1}$. The standard
approximation approach was introduced by \citet{Whittle1951}, and the
resulting, much studied, discretized approximate likelihood is commonly known
as the discrete Whittle likelihood.  The Whittle likelihood gives an
approximation to the likelihood of the time domain data in the frequency domain
via the Fourier coefficients, under assumptions as specified by \citet[p.
109--113, and 116--7]{Beran1994}. Problems associated with the usage of
Whittle's approximation for non-Gaussian and small sample size Gaussian time
series has been discussed by \cite{Contreras}.

The final two terms in equation (\ref{ExactTimeLike}) are approximated using
results of \citet{Whittle1951} and \citet{Grenander1984}.  It follows that the
Whittle likelihood for $(\pole,\delta)$ and $\bm{\theta}$ is given by:
\begin{equation}
\label{Whittle}
\loglike_N^{(W)}\left(\pole,\delta,\bm{\theta},\sigma^2_{\epsilon}\right)=-\int_{-\frac{1}{2}}^
{\frac{1}{2}}\frac{I_0\left(\lambda\right)}{f\left(\lambda\right)}\;d\lambda,
\end{equation}
where $I_0\left(\lambda\right)$ is the \textit{periodogram}, defined as the
modulus square of the discrete Fourier transform (DFT),
$Z_0\left(\lambda\right)$, of the realized time series.

At the Fourier frequencies $\freqj = j/N$, $j=0,\ldots, M$,
the periodogram, $I_0$, is given by,
$I_0(\freqj) = |Z_0(\freqj)|^2 $ where
\begin{equation}
\label{Zj} Z_0(\freqj) = \frac{1}{\sqrt{N}} \sum_{t=0}^{N-1}  X_t
e^{-i2\pi t \freqj} = A_0(\freqj) - i B_0(\freqj),\quad j=0,\ldots, M,
\end{equation}
so that
\[
I_0(\freqj) = A_0^2(\freqj)+B_0^2(\freqj)=
\frac{1}{N} \left[ \sum_{t=0}^{N-1} X_t^2 + 2
\sum_{t=1}^{N-1} \sum_{s=0}^{t-1} X_t X_s \cos\left\{ 2 \pi j (t-s)/N
\right\} \right] .
\]
For short memory data, the periodogram is an asymptotically unbiased but
inconsistent estimator of $f(\cdot)$ that is commonly used as the basis of more
sophisticated estimation procedures. The use of (\ref{Whittle}) for parameter
estimation has been discussed in detail by \cite{Walker,Walker2} and
\citet{Hannan1973} under the assumption that the log spectrum integrates to
zero.  \citet{Hosoya1974} added a second term of
$\log\left\{f(\lambda)\right\}$ to the integral to deal with more general
processes.

For the likelihood in equation (\ref{Whittle}) to have desirable asymptotic
properties, it is assumed that the process is linear, and satisfies certain
regularity conditions, thus ensuring good large sample properties of the
likelihood based estimators. Note that (\ref{Whittle}) is an {\em
approximation} to the log-likelihood of $\bm{X}$ based on the periodogram, but
that (\ref{Whittle}) is not a likelihood {\em for} the periodogram. The
approximation of the likelihood in equation (\ref{ExactTimeLike}) by equation
(\ref{Whittle}), performs well when the process is Gaussian and the covariance
of the time series is either rapidly decaying or exactly periodic.

A Riemann approximation to the integral in equation (\ref{Whittle}) yields the
discrete analogue
\begin{equation}
\label{discfourlik}
\loglike_N^{(DW)}(\pole,\delta,\bm{\theta},,\sigma^2_{\epsilon})=-\frac{2}{N}\sum_{j=0}^{M}
\frac{I_0(\varphi_j)}{f(\varphi_j)},
\end{equation}
and we could also adjust this to allow for more general processes:
\begin{equation}
\label{discfourlik2}
\loglike_N^{(DW)}(\pole,\delta,\bm{\theta},,\sigma^2_{\epsilon})=-\frac{2}{N}\sum_{j=0}^{M}
\frac{I_0(\varphi_j)}{f(\varphi_j)}-\frac{2}{N}\sum_{j=0}^{M}\log[f(\varphi_j)],
\end{equation}
following Hosoya's proposal. By defining the vector
$\bm{C}_{2j,2j+1}(A_j,\;B_j)^\transpose,$ where $A_j=A_0(\freqj)$ and
$B_j=B_0(\freqj)$, and $\bm{\Sigma}_{\bm{C}}$ as the exact covariance of
$\bm{C}$, we may consider the exact log-likelihood, $\loglike_N^{(f)}$ of the
DFT of observed and Gaussian data via:
\begin{equation}
2 \loglike_N^{(f)} \left(\pole,\delta,\bm{\theta},\sigma^2_{\epsilon}
\right) = -N \log (2 \pi)
- \log \left|{\bm{\Sigma}}_{\bm{C}}\right|
-\bm{C}^\transpose{\bm{\Sigma}}_{\bm{C}}^{-1}\bm{C}, \label{DFTLike}
\end{equation}
in direct analogue with equation (\ref{ExactTimeLike}), acknowledging finite
sample effects of the DFT. The difference between this equation and the
discrete Whittle likelihood is that it involves the {\bf exact} covariance
matrix, ${\bm{\Sigma}}_{\bm{C}}$, of the FFT coefficients. Analysis based on
the likelihood of the Fourier coefficients (in general) involves the inversion
of the large, non-sparse covariance matrix, and is thus equally inefficient as
the basis of likelihood procedures as equation (\ref{ExactTimeLike}).

Having specified these various likelihood functions that could be used for
inference, some justification must be used to motivate their usage. Equation
(\ref{DFTLike}) is a natural choice for analysis of seasonal time series, given
the compression of the variables of the seasonal effects. We shall use the
compression to approximate the likelihood more carefully, acknowledging large
finite sample effects related to the compression explicitly.

\subsection{Contributions of the Paper}
We introduce an approximation to equation (\ref{DFTLike}), and use this as the
basis of a maximum likelihood procedure. We focus on the distribution and other
properties of the periodogram, given an underlying SPP with sdf $f(\cdot)$. We
focus on Gaussian processes, and do not consider here the non-Gaussian case.
However, for other processes, such as those in \cite{brillinger}, where
asymptotic normality of the DFT holds, our distributional results are still
valid.

Specifically, we consider estimation of parameters of spectra with spectral
poles away from frequency zero. We consider an adjustment to the standard DFT
that simplifies the technical developments of this paper. A simple (but
parameter dependent) modification of the choice of grid, conditional on a known
spectral pole location, leads to simple approximations to the likelihood of the
periodogram at a new set of frequencies spaced at a distance $\O(N^{-1})$
apart. In particular
\begin{enumerate}

\item We propose a new demodulated Whittle discrete likelihood for seasonal
    processes (sections 2 \& 3). We show that the proposed likelihood
    approximates the distribution of the discrete Fourier transform for any
    posited value of the true parameters (see Theorem 1). The key idea is
    to use a different orthogonal transformation of the data conditional on
    each fixed value of the location of the pole (specification of a
    compressed representation). This is a non-standard situation.

\item To establish the properties of the likelihood we
    calculate the large finite sample distribution of the
    periodogram at the pole itself (Section 2.3).

\item We bound the covariance of the demodulated periodogram at different
    frequencies spaced $1/N$ apart (noted in Section 3), and note its
    asymptotically negligible contribution to the normalized
    log-likelihood. Furthermore, the choice of approximation to the
    likelihood is not everywhere continuous.  However, we demonstrate
    (Section 3) that the discontinuities in the likelihood surface
    represent a negligible contribution for finite large samples.

\item We prove consistency of the MLEs (see Theorem 2), and determine the
    large sample first order properties of the score and observed Fisher
    information (Theorem 3).

\item We determine the asymptotic distribution of the score
    and observed Fisher information (see Theorem 4) and the asymptotic distribution of the MLEs (see Theorem 5).

\item We give a large finite sample approximation to the
    distribution of the pole estimator (see Proposition 6).

\end{enumerate}

To derive the appropriate large sample theory, some care is required.  It
transpires that the score and Fisher information do not exhibit the usual large
sample behaviour.  Our results are based on a Taylor expansion of the
log-likelihood; we adopt the normalization of the observed Fisher information
adopted by \citet{Sweeting,Sweeting2}. We thus renormalize the observed Fisher
information appropriately with a suitable power of $N$. The renormalized score
and observed Fisher information converge in law to Gaussian random variables
that are asymptotically uncorrelated. The distribution of $\polehat$ converges
slowly to the asymptotic distribution, and so alternate finite large sample
approximations are also given.

These results establish a new large sample theory for seasonally persistent
processes, and utilize the data-dependent transformation of the time-domain
data that facilitate the computation of the distribution of different random
variables for each posited value of the pole, and appropriate normalisation
techniques for the score and Fisher information when the data is modelled as
highly compressed in the Fourier domain.

\subsection{Connections with Recent Work}

In connections with other related work, we distinguish between likelihood-based
methods and semi-parametric methods for processes exhibiting seasonal
persistence. \citet{Giraitis2001} consider fully parametric models, and  constrain the
maximization over the location to a grid of frequencies spaced $\O(N^{-1})$
apart. \citet{HandS2004} consider semi-parametric models, and the theoretical
properties of the extended Geweke-Porter-Hudak estimator, basing their analysis
on estimating the location of the singularity as the Fourier coefficient of the
maximum periodogram value in a given frequency interval; in their simulation
study, the true location of the singularity is aligned with the Fourier
frequency grid. \citet{HandS2004} evaluate the Fourier coefficients at the
Fourier frequency grid, and restrict the estimate of the location of the pole
to a grid of frequencies spaced $\O(N^{-1})$ apart. \citet{Hidalgo2005} used
semi-parametric methods to estimate the location of the pole, as well as the
long memory parameter. By using a two-step procedure he is able to develop
large sample theory for the estimator of the singularity, whereas in contrast
we focus on full likelihood methods. More recently, \citet{Whitcher2004} used a
wavelet packet analysis approach for estimation of seasonally persistent
processes.

In terms of asymptotic properties, our rate of convergence matches that of
\citet{Giraitis2001}. However, in addition, we obtain the large sample
distributional results for the estimator of the pole, which they fail to do,
having produced a different estimator. Similarly to Giraitis' {\em et al.},
\citet{Beran} estimate the location of the pole using the coefficient which
maximises the periodogram. Our estimator is again different although
asymptotically equivalent with the same rate of convergence, and it has a
determinable asymptotic, as well as large finite sample approximate,
distribution.

Our work also has a connection with, but is different in spirit from,
\textit{hidden frequency estimation}, in which the seasonal structure is
modelled as deterministic, corresponding to a single sinusoid. In this case,
the Fourier coefficient which maximizes the periodogram converges to the true
coefficient with a faster rate than the convergence of the MLE of the pole.
Such rates were improved by secondary analysis, and the corresponding analysis
using data tapers, see for example
\citet{ChenWuDahlhaus,Hannan1973,Hannan1986,VonSachs1993}. Secondary analysis
corresponds to partitioning the time series into several groups of data, and
using regression to estimate the so-called hidden frequency.
\citet{Thomson1990} used multitaper methods to improve the detection of a set
of hidden frequencies, and use least squares methods over a given bandwidth.
Neither the model we use, nor our proposed inferential method, is equivalent to
the above mentioned procedures. Secondary analysis can be considered to `zoom
in' on local structure near the pole, and may be philosophically related to our
procedure, but we implement full likelihood for a full set of Fourier
coefficients. Conditionally for each fixed value for the pole, we calculate the
distribution of a different set of random variables, but as each set is a
linear and orthogonal transformation of the original data, and with a constant
and equal Jacobian, this is appropriate.

Finally, we note that the inferential issues are of importance beyond
seasonally persistent processes. The inherent non-regularity arises due to a
parameter dependent transformation of the time-domain data. Whenever the
process is modelled using a suitable parametric linear transformation of the
data that will give decomposition coefficients that are non-negligible only for
a few sets of indices, our methods will be applicable with some minor
modifications.  In a more general setting we would write the variances of a set
of basis coefficients as satisfying a power-law decay, and we refer to such
processes as {\emph{second order compressive}} processes. Power-law decay in a
suitable basis is an relatively common phenomenon - see for example the
discussion in \citet{Donoho2006,Abramovich+2006,CandesTao} - and our
developments will carry across to this setting if the compression is stochastic rather than deterministic, once the location and decay
parameters have been incorporated in the arbitrary basis. Issues of alignment,
and/or shift-variance, akin to results that arise for misspecified location of
the pole, are very well-documented in other basis expansions \citep{coif}. Note
that the equivalent to the decay parameter discussed by the aforementioned
authors will be $p=1/(2\delta)$. Only for $\delta>0.25$ are we in their mode of
decay, corresponding to extreme regimes of long memory behaviour.

\section{Distributional results for the Periodogram}
\label{demodsec}
\subsection{Large Sample Properties} The large
sample properties of the periodogram of seasonally persistent processes were
determined in \citet{OMS2004}.  We summarize and extend these results below; in
particular we compute the statistical properties of the periodogram itself at
the pole $\pole$, as this specific Fourier coefficient will contribute
substantively to the subsequent likelihood calculation.

Theorem 1 in \citet{OMS2004} gives the following result concerning the relative
bias at frequency $\lambda$, $B_{\lambda,N}(\pole, \delta )$, of the
periodogram for all $\lambda \in (-1/2,1/2), \pole \in (0,1/2)$,
\renewcommand{\arraystretch}{1.8}
\[
B_{\lambda,N}(\pole, \delta ) =
\left\{
\begin{array}{ll}
{\mathrm{E}}\left\{ \dfrac{\Pgram(\lambda) }{f( \lambda ) }\right\} & \qquad \lambda \neq \pole \\[12pt]
{\mathrm{E}}\left\{ \dfrac{\Pgram( \pole ) }{N^{2\delta}f^{\dagger}( \pole)
}\right\} & \qquad \lambda = \pole
\end{array} \right.
\]
This notation makes explicit the dependence of the relative bias on $\pole$ and
$N$.  For frequencies $\freqk = k/N, k \in \left\{0,\ldots, M \right\}$, we
have, for large $N$ and a fixed value of $\pole$, with $\freqk\neq \pole$,
\begin{equation}
B_{\freqk,N}(\pole, \delta ) =\frac{2}{\pi }\int_{-\infty
}^{\infty } \left[\frac{\sin \{u /2-\pi c_{N}(\pole,\freqk)\} }{ u - 2\pi
c_{N}(\pole,\freqk)}\right]^2 \left| \frac{2\pi c_{N}(\pole,\freqk)}{u
}\right| ^{2\delta }\;du +\o(1), \label{deffy}
\end{equation}
where $c_{N}(\pole,\freqk) = N(\freqk - \pole)$ denotes $N$ times the distance
between the $k^{th}$ Fourier frequency and the pole at $\pole$. For the case
$\freqk = \pole$, the large sample value of $B_{\pole,N}\left(\pole, \delta
\right)$ is given in Lemma \ref{PeriodogramPole} in Section \ref{secPP}.

For the second order moment properties, let
\begin{eqnarray*}
C_{\varphi_k,\varphi_l,N}( u, \pole ) &=& \frac{\sin \{ u /2-\pi
c_{N}(\pole,\freqk )\} \sin \{ u /2-\pi c_{N}(\pole,\freql )\}
}{\{u -2\pi c_{N}(\pole,\freqk)\} \{ u -2\pi
c_{N}(\pole,\freql)\} } \\[12pt]
{V}_{\varphi_k,\varphi_l,N}\left( \pole,\delta  \right) &=& \left( -1\right) ^{k+l}
\frac{2}{\pi }\int_{-\infty }^{\infty }C_{\varphi_k,\varphi_l,N}\left(u,
\pole\right) \left| \frac{2\pi }{u }\right| ^{2\delta }|
c_{N}(\pole,\freqk )c_{N}(\pole,\freql )| ^{\delta }\;du +\o(1).
\end{eqnarray*}
Then, for $A_0(\freqj),B_0(\freqj)$ from (\ref{Zj}), \citet{OMS2004} gives
\begin{eqnarray*}
{\mathrm{E}}\{A_0(\freqk) A_0(\freql)\} = {\mathrm{E}}\{B_0(\freqk) B_0(\freql)\} &=& \{V_{\varphi_k,\varphi_l,N} (\pole,\delta)/2 +
\o(1)\} \sqrt{f (\freqk )f (\freql)} \\ {\mathrm{E}}\{A_0(\freqk) B_0(\freql)\} = {\mathrm{E}}\{B_0(\freqk) A_0(\freql)\}
&=& \o(1) \sqrt{f (\freqk )f (\freql)}\\
&=&\o\left(N^{2\delta}\right)\quad {\mathrm{if}}
\quad c_N(\pole,\freqk),\;c_N(\pole,\freql)=O(1).
\end{eqnarray*}
These results specify the large sample first and second order structure of the
periodogram.  We now extend these results to the demodulated periodogram
described in section \ref{DDFTsec}.  Note that a direct implication of these
results is that the distribution of the periodogram is highly dependent on the
distances between the pole $\pole$ and the Fourier frequencies $\{\freqk\}$.

\subsection{The Demodulated Discrete Fourier Transformation}
\label{DDFTsec}  The Discrete Fourier Transform of $\{X_t\}$ is not constrained
to be evaluated at $\{\freqk\}$, but in fact any $\O(N^{-1})$ grid could be
considered. This fact leads us to consider \textit{demodulation}, a grid
realignment technique, which for any fixed value of $\pole$ produces a new grid
aligned with the pole. Demodulation ensures that the large sample behaviour of
the demodulated periodogram is similar to that of the periodogram of a standard
long memory process (where $\pole = 0$).  Specifically, the large sample bias
is the same but the distribution of the periodogram is $\chi^2_2$ rather than a
sum of unequally weighted $\chi^2_1$ random variables
\citep[see][p.~621]{HB1993,OMS2004}.

The Demodulated Discrete Fourier Transform (DDFT) or offset DFT
\citep{PeiDing2004} of a sample of size $N$ from time series $\{X_t\}$ with
demodulation via a fixed frequency $\lambda$ is denoted $Z_{\lambda}$, and is
defined for Fourier frequency $\freqj$ by
\begin{equation} \label{ddft}
Z_{\lambda} (\freqj)=\frac{1}{\sqrt{N}}\sum_{t=0}^{N-1} X_t
e^{-2i\pi (\freqj+\lambda)t} = A_{\lambda} (\freqj ) - i
B_{\lambda} (\freqj), \quad j=0,\ldots, M.
\end{equation}
The demodulated periodogram at frequency $\freqj$ with demodulation via
$\lambda$ is denoted $I_{\lambda}(\freqj)$, and is defined via the ordinary
periodogram $I_0$ by
\[
I_{\lambda} (\freqj) = I_0 (\freqj+\lambda ) = |Z_{\lambda} (\freqj)|^2
=A_{\lambda}^2 (\freqj)+B_{\lambda}^2(\freqj).
\]
Hence $I_{\lambda} (\freqj)$ is simply the periodogram evaluated at frequency
$\freqj+\lambda,$ or $I_0(\freqj+\lambda).$ We will consider evaluating this
expression at arbitrary frequency $\varphi$. We define
$C_{\lambda;2j,2j+1}=(A_{\lambda,j}, B_{\lambda,j})^\transpose =
\{A_{\lambda}(\freqj),\;B_{\lambda}(\freqj)\}^\transpose,$ in analogue to
$\bm{C}$ in equation (\ref{DFTLike}). For Gaussian data $\bm{X}$ we then find:
\begin{equation}
\bm{C}_{\lambda} = \left(
A_{\lambda,0},
B_{\lambda,0},
\dots
A_{\lambda,M},
B_{\lambda,M}
\right)^T 
\overset{\mathfrak{L}}{=}\mathcal{N}\left(\bm{0},\bm{\Sigma}_{\bm{C}_{\lambda}}\right) \label{MVNSigmaC}
\end{equation}
Note that due to the demodulation, $B_{\lambda,0} \neq 0$ in general, unlike
the imaginary component of the DFT at frequency zero. To efficiently formulate
the likelihood, we need to explicitly consider the computation of
$\bm{\Sigma}_{\bm{C}_{\lambda}}$, the covariance of the DDFT coefficients.

\subsection{Extending the \citet{OMS2004} result}
The results in \citet{OMS2004} do not cover the case of demodulation, and to
enable calculation of the new likelihood, further results are required. For
example $B_{\pole,N}(\pole,\delta)$ needs to be explicitly determined. To
minimize the bias in the demodulated periodogram, and simplify the covariance
structure, we shift the Fourier grid so that the closest Fourier frequency to
the pole in the original grid coincides exactly with the pole in the
demodulated version. For a pole at $\pole$, we denote by $j_{0,N}(\pole)=[
N\pole],$ where $[x]$ indicates the nearest integer to $x.$ We furthermore let
$c_{N}(\pole,\varphi_{j_{0,N}(\pole)})=j_{0,N}(\pole)-N\pole$ and specify
$\lambda=\lambda_{D,N}(\pole)$ in (\ref{ddft}) as $\lambda_{D,N}(\pole) =
-c_{N}(\pole,\varphi_{j_{0,N}(\pole)})/N$. The approach introduces a new grid
of frequencies, namely
\begin{equation}
\label{lambdagrid}
\lambda_k \equiv \lambda_{k(j),N}(\pole)=
\varphi_j+\frac{c_{N}(\pole,\varphi_{j_{0,N}(\pole)})}{N}=
\pole+\frac{j-j_{0,N}(\pole)}{N}=\pole+
\frac{k(j)}{N}.
\end{equation}
We exclude Fourier frequencies $0$ and $1/2$, and taking $j=1,\dots,M-1$ we
have $k=k(j)=j-j_{0,N}(\pole)\equiv J_1,\dots,J_2\equiv
-j_{0,N}(\pole),\dots,M-j_{0,N}(\pole)$. For example, if $N=16$ and $\pole =
0.15$, then $\lambda_{D,16}(0.15) = 0.025$, $[N\pole] = 2$, $J_1=-1$ and $J_2 =
5$.  Note that for $k(j_2) > k(j_1)\neq 0$, then
\[
V_{\lambda_{k(j_1)}, \lambda_{k(j_2)},N}(\pole,\delta)=
V_{\varphi_{j_1},\varphi_{j_2},N}\left\{\frac{j_{0,N}(\pole)}{N},\delta\right\},
\]
so that the covariance properties of the DDFT can be easily determined.

Under this demodulation, the DDFT yields the original periodogram $I_0$
evaluated at frequencies $\lambda_k \equiv
\lambda_{k(j)}=\pole+(j-j_{0,N}(\pole))/N,$ and takes the form
\begin{equation}
\label{ddftfinal} Z_{\lambda_D} \left(\freqj\right) = Z_0
\left(\lambda_{k(j)}\right)=\frac{1}{\sqrt{N}}\sum_{t=0}^{N-1}
X_t e^{-2i\pi  \lambda_{k} t}, \quad k=J_1,\ldots, J_2,
\end{equation}
so that, for $k=J_1,\ldots, J_2$, $I_{\lambda_D}(\freqj) = I_0(\lambda_{k}) =
I_0(\pole + k/N).$ The DDFT can be computed efficiently by applying the DFT to
the new series $\{Y_t\},$ defined for $t=0,\ldots,N-1,$ by $Y_t = X_t
\exp\left\{ -2 \pi i \lambda_D t\right\}$. Demodulation both simplifies the
mathematical calculations considerably, and improves estimation of the
persistence parameter $\delta.$ Naturally the operation is very straightforward
to implement. The parameter dependent choice of $\{\lambda_k\}$ will need
careful analysis when deriving properties of the parameter estimators.

\subsection{Expectation of the Periodogram at the Pole} \label{secPP}
The result in (\ref{deffy}) gives the relative bias of periodogram. The
expectation of the periodogram is given in the following Lemma.
\begin{lemma}
\label{PeriodogramPole} The expected value of the periodogram
evaluated at the pole $\pole$, after demodulation by $\pole$, is
\begin{eqnarray*}
E\left\{I_0(\pole)\right\} & = &
(2\pi N)^{2\delta} \{-2f^{\dagger}(\pole)
\Gamma(-1-2\delta)\}\cos\{\pi( 1/2+\delta)\}\pi^{-1}+\o(1)\\[6pt]
& \overset{def}{=} & N^{2\delta} f^{\dagger}(\pole) \Polebias + \o(1) = \O(N^{2\delta}).
\end{eqnarray*}
\end{lemma}
\begin{proof}
See Appendix \ref{PeriodogramAppendix}.
\end{proof}

\section{Asymptotic Properties of the Likelihood and Estimators}
\label{Likelihoodsection}

In this section we utilize demodulation, and the large sample approximations
described above, to present three theorems that characterize the asymptotic
behaviour of the likelihood, the corresponding MLEs for $(\pole,\delta)$ and
the associated Fisher information to obtain their large sample properties.
Specifically, we establish $N$-consistency for the estimator of the location of
the pole, thus matching the result of \citet{Giraitis2001}.

\subsection{Large-sample Likelihood Approximation}
For a periodogram demodulated to align the Fourier grid with pole $\pole$, we
have the following asymptotic result.
\begin{theorem}
\label{ApproxLike}{\text{Approximating the Likelihood Function.}}\\
For a Gaussian series from a periodic long memory model as described by
(\ref{giraitis}), where $f^{\dagger}(\cdot)$ is twice partially differentiable
with respect to $(\pole,\delta)$, the log-likelihood of the discrete Fourier
transform can be approximated by
\begin{eqnarray}
\ell\left(\pole,\delta,\bm{\theta},\sigma^2_{\epsilon}\right) &=& \sum_{j=J_1}^{J_2} \left\{
\log \eta_j - \eta_j I_0(\pole + j/N) \right\} 
\label{approxlikeeqn}
\end{eqnarray}
accurate to $o(N)$, where
\begin{equation}
\label{Theorem1Etaj} \eta_j =\frac{
|j|^{ 2 \delta\Upsilon{\{j\neq0\}}} }
{\Polebias^{\Upsilon{\{j=0\}}}
N^{2\delta} f_j^{\dagger}}
\end{equation}
for $0<\delta<0.5$, where $\Upsilon\{A\}$ is the indicator function for event $A$,
\[
f^{\dagger}_j \equiv f^{\dagger}(\lambda_j) = f^{\dagger}(\pole + j/N)
\]
and $\Polebias$ is the asymptotic relative bias given by Lemma
\ref{PeriodogramPole}.
\end{theorem}

\begin{proof}
See the Appendices A.2-A.4.
\end{proof}

\vspace{0.1in} \noindent\textit{\Note[Note I]} The approximation to the
likelihood is equivalent to that of independent exponential random variables
with rate parameters $\eta_j$ that depend on $j$ and $\delta$ but not on
$\pole$.  In equation (\ref{Theorem1Etaj}), the function $\Polebias$
appropriately scales the periodogram contribution from the Fourier frequency
aligned with $\pole$. $\Polebias$ is monotonically increasing in $\delta$, with
$\lim_{x\rightarrow 0} {\mathbf{B}}_\xi(\xi,x)=1$, and $\Polebias$ is bounded away from zero.
As the function is monotonic the derivatives of $\Polebias$ are also bounded
away from zero. If $f^{\dagger}(\cdot)$ is also bounded away from zero, then
the log likelihood is bounded in $\pole$ and $\delta$. Thus it is possible to
find efficiently the MLEs of $\pole$ and $\delta$ numerically.

\noindent\textit{\Note[Note II]} This likelihood is not differentiable with
respect to $\pole$ at all values of $\pole$; although $\dot{I}_0(\pole + j/N)$
is available in simple form, the dependence of $J_1=-j_{0,N}(\pole)$ and
$J_2=M-j_{0,N}(\pole)$ on $\pole$ renders the overall function discontinuous.
However, the discontinuities are $\O(1)$ in magnitude, and the log likelihood
is uniformly at least $\O(N)$, so in fact the discontinuities are negligible,
but motivate us to look, in standard fashion, at the $N$-standardized
likelihood function $\ell(\pole,\delta,\bm{\theta},\sigma^2_{\epsilon})/N$. See
Appendix A.4 for further details.

\noindent\textit{\Note[Note III]} The formulation in Theorem \ref{ApproxLike}
summarizes the data in the frequency domain via the demodulated periodogram for
a given $\pole$. We avoid the introduction of the substantial bias and
covariance terms found in \citet{OMS2004}, as the demodulated periodogram is
perfectly aligned with this singularity.  When other demodulations are chosen
the likelihood cannot be approximated in such a fashion. Even for frequencies
of sufficient distance from any irregular behaviour, the results of
\citet{OMS2004} cannot be applied directly, and to find the approximate Whittle
likelihood we additionally need to make assumptions about the spectral density
function, and its smoothness (see \citet{Dzham1983} and \citet{Taniguchi2000}).

\noindent\textit{\Note[Note IV]} The result differs with that of \citet{HB1993}
in a number of ways. The ordinates subscripted $j$ and $-j$ in the DFT are no
longer complex conjugates, and the likelihood at evaluated at $\lambda_j$ is
now approximately $\chi_2^2$ (rather than a mixture of two different $\chi^2$
terms) even for those coefficients closest to the pole. Strictly, the
definition for $\eta_j$ in equation (\ref{Theorem1Etaj}) has an additional term
$V_{\lambda_j,\lambda_k,N}(\pole,\delta)$ for $j,k\in \mathbb{Z}$, but these
terms can be bounded appropriately, and thus contribute in a negligible
fashion. The bias at the pole reported in \citet{HB1993} is (identically)
present in our formulation, but is $\o(N)$, and is thus subsumed into the final
term -- see \citet{Hurvich1998} for relevant supporting arguments.

\vspace{0.1 in}


\subsection{Existence and Consistency of the ML estimators}
We now use the results of the previous section to construct likelihood-based
estimators of $\pole$ and $\delta$ and establish their properties. The
following theorem establishes the existence and consistency of the ML
estimators derived from the likelihood in Theorem \ref{ApproxLike}.

\begin{theorem}
\label{GiraitisLike} \textbf{Existence and Consistency}\newline For the
likelihood of Theorem \ref{ApproxLike}, the ML estimators of $\pole$ and
$\delta$, $\polehat$ and $\deltahat$, exist and are consistent, with
convergence rates $N$ and $N^{1/2}$ respectively.
\end{theorem}

\begin{proof}
See Appendix \ref{ConsistencyAppendix}.
\end{proof}

\vspace{0.1 in}

\noindent The $N$-consistency of $\polehat$ matches the convergence rate of
\citet{Giraitis2001}. It is unusual to find superconsistent estimators in
likelihood based procedures. An intuitive understanding of the rate can be
found in the time domain. As we collect $N\polestar$ full periods of the data
the periodicity of the data is determined to an accuracy of $\O(N^{-1})$. The
reason why this rate is achieved is that the log-likelihood is varying
$\O(N^{3/2})$ (see proposition 14) over distances in $\pole$ of $\O(N^{-1})$
near the value $\pole = \polestar.$ However, the convergence rate is different
to that of \citet{ChenWuDahlhaus}. The latter model the seasonality as a
deterministic seasonal component embedded in stationary noise. In Chen {\em et
al.} a regression model is employed to estimate the amplitude and locations of
the seasonality, and a rate of $N^{-3/2}$ rather than $N^{-1}$ is achieved. We
employ a different model and hence do not expect the same convergence rates as
is achieved by Chen {\em et al.}.

\subsection{Properties of The Fisher Information Matrix}
\begin{theorem}
\label{Fisher}{\textbf{The Fisher Information.}}\\
For a series from a periodic long memory model as described by
(\ref{giraitis}), for large $N$, the components of the Fisher information
\[
\fisherN=
\left(
\begin{array}{cc}
\fisherNe_{\pole,\pole} & \fisherNe_{\pole,\delta}\\
\fisherNe_{\pole,\delta} & \fisherNe_{\delta,\delta}
\end{array}
\right)
\]
are given by
\begin{eqnarray*}
\fisherNe_{\delta,\delta} = \fisherNlme_{\delta,\delta}N+\o(N), ~~~~~
\fisherNe_{\pole,\delta} = \fisherNlme_{\pole,\delta}^{(1)}N+o\{\log(N)\},~~~~~
\fisherNe_{\pole,\pole} = \fisherNlme_{\pole,\pole} N^2+\o(N^2),
\end{eqnarray*}
where $\fisherNlme_{\delta,\delta}$,
$\fisherNlme_{\pole,\delta}^{(1)}$ and
$\fisherNlme_{\pole,\pole} $ are constants independent of
$N$ but are functions of the true values of $\pole$ and
$\delta$.
\end{theorem}
\begin{proof} See Appendices \ref{FirstOrder} and \ref{SecondOrder}.
\end{proof}

\vspace{0.1 in} \noindent For a full analysis in a regular ML setting, the
second order properties of the MLEs can in a general setting be deduced from
the above quantities. Large sample properties, specifically consistency and
asymptotic variance, may be considered via a Taylor expansion of the
log-likelihood, see for example \cite{cheng}. However, we note that we are
\textbf{not} in a standard setting; even if we may expand the log-likelihood
near the true value of the parameter, because of the non-standard behaviour of
the derivatives of the log-likelihood, the observed Fisher information does
\textbf{not} converge to a diagonal matrix with constant entries, but rather
the (appropriately standardized) observed Fisher information for $\pole$
converges to a random variable with order one variance. We will discuss the
interpretation of the Fisher information in this context, in the appendix. For
the derivatives involving the location of the pole, extra terms of magnitude
$N$ are introduced and thus the variance of the observed Fisher information in
$\pole$ is $\O(N^5).$ The magnitude of the variance of the observed Fisher
information in $\pole$ implies that a standardization of the random variable
must be employed that results in a negligible expectation of the restandardized
random variable. We also therefore discuss the large sample theory of the {\em
observed} Fisher information.

\subsection{The Asymptotic Properties of the MLEs}
We now consider the use of the Fisher information to determine the asymptotic
variance. Consider a Taylor expansion of the score near the true value
$\chistar$ of the parameters $\chival = (\pole,\delta)$ evaluated at the MLE
$\chihat$.  We denote the observed Fisher information by $\obsfisherN
\left(\chival\right)$, and let $\chiprime$ lie between $\chival$ and
$\chistar$. We denote by $\ldot(\chival)$ the score in $\chival$, noting that
the score is well defined if the log-likelihood is evaluated ignoring the
$\pole$ dependence of $J_1$ and $J_2$, see section A-4:
\begin{equation}
\label{scoreinlambda}
\ldot(\chival)=
\left(\begin{array}{c}
\ell_{\pole}\left(\chival\right)\\
\ell_{\delta}\left(\chival\right)
\end{array}
\right). \end{equation} Then using a first-order expansion of the log
likelihood in the usual way for $N$ sufficiently large, we have the (vector)
score function:
\[
\ldot(\chihat) = \ldot(\chistar)-
\obsfisherN(\chiprime)( \chihat-\chistar) \qquad
 \Longleftrightarrow \qquad
\obsfisherN(\chiprime)(\chihat-\chistar ) =
\ldot(\chistar)-\ldot(\chihat).
\]
Thus the difference between $\chihat$ and $\chistar$, when appropriately scaled
by the random matrix $\obsfisherN(\chiprime)$ corresponds to the value of the
score at $\chistar$ in the usual fashion.  The statistical properties of
$\obsfisherN(\chiprime)$ are not straightforward in this non-regular problem,
and require further investigation. Following \citet{Sweeting2}, we define a
suitable standardization matrix $\bm{B}_N$ and the standardized observed Fisher
information by
\begin{equation}
\label{defofBN}
\bm{B}_N=\left(\begin{array}{cc} N^{5/2} & 0 \\
0 & N \fisherNlme_{\delta,\delta}
\end{array}\right) \qquad \obsfisherstdN=\bm{B}_N^{-1/2}\obsfisherN \bm{B}_N^{-1/2},
\end{equation}
as the large sample properties of $\obsfisherstdN$ are tractable, and their
determination is an important step to finding the large sample properties of
the MLE.  Specifically, we let
\[
\bm{B}_N^{-1/2} \obsfisherN(\chiprime){\bm{B}}_N^{-1/2}
\bm{B}_N^{1/2}(\chihat-\chistar) =
\bm{B}_N^{-1/2}\{\ldot(\chistar) -\ldot(\chihat)\},
\]
so that
\[
\obsfisherstdN(\chiprime)
{\bm{B}}_N^{1/2}(\chihat-\chistar ) =
{\bm{B}}_N^{-1/2}\ldot(\chistar)=\stscore(\chistar).
\]
The latter expression defines the standardized score $\stscore(\cdot).$ See the
Appendix for a full discussion of these quantities. Note that $\bm{B}_N$ is the
large $N$ approximation to the Fisher information matrix for $\delta$ and
corresponds to an appropriate order normalisation for $\pole,$ thus
$\obsfisherstdN\left(\chival\right)$ is the observed Fisher information
renormalized by ${\bm{B}}_N$.

We note that for $N$ large enough, the expected value of
$\obsfisherstdN(\chistar)$ is the identity matrix for the $\delta$ entry, but
the expectation of the first entry of $\obsfisherstdN(\chistar)$ is $\o(1)$
whilst the variance of the first entry is $\O(1).$ In a standard setting the
expectation is $\O(1)$ and the variance $\o(1)$.

\begin{theorem}
\label{Score}{\textbf{Distribution of the Score and Observed Fisher Information.}}\\
For the likelihood of Theorem \ref{ApproxLike} the standardized score
$\stscore(\chistar)$ and the standardized Observed Fisher information matrix
$\obsfisherstdN(\chistar)$ asymptotically have the following properties:
\begin{equation}
\stscore(\chistar)\overset{\cal{L}}{\longrightarrow}\limstscore,\quad
\obsfisherstdN(\chistar)\overset{\cal{L}}{\longrightarrow}\bm{W},
\end{equation}
where the entries of $\limstscore = (k_1,k_2)^\transpose$ and $\bm{W}$ are uncorrelated and
\begin{equation}
\begin{array}{lcccccccccr}
W_{11}&\sim& \mathcal{N}\left(0,8 \pi^4/15 \right),& & W_{12}&=& 0, & & W_{22}&=&1\\
k_1&\sim& \mathcal{N} \left(0,\pi^2/3 \right),& & k_{2}&\sim& \mathcal{N}\left(0,1\right) .& & &&
\end{array}
\end{equation}
\end{theorem}
\begin{proof}
An outline of the proof given in the Appendix, see Proposition \ref{dist8},
Section \ref{FisherLim} and Proposition \ref{mledistt}.
\end{proof}

\begin{theorem}
\label{mle}{\textbf{Distribution of the MLE.}}\\
For the likelihood of Theorem \ref{ApproxLike}, the ML estimators of $\pole$
and $\delta$, $\polehat$ and $\deltahat$ have distributions that for large
sample approximately take the form:
\begin{equation*}
N(\polehat-\polestar) =
\left\{\frac{N^{5/2}}{-\loglike_{\pole,\pole}(\chiprime)} \right\}
\left\{\dfrac{\loglike_{\pole}(\chistar)}{N^{3/2}} \right\}
\: \overset{\cal{L}}{\longrightarrow} \; \frac{\sqrt{5}}{2\pi\sqrt{2}} C,
\end{equation*}
where $C$ is distributed according to the standard Cauchy distribution, and
\begin{equation}
\sqrt{N}(\widehat{\delta}-\delta^{\star})
\sim \mathcal{AN}\left(0,N \left\{\fisherNe_{\delta,\delta}\right\}^{-1} \right).
\end{equation}
An estimator of the asymptotic variance is formed via
${\hatfisherNe}_{\delta,\delta}=\fisherNe_{\delta,\delta}(\widehat{\xi},\widehat{\delta})$.
The forms of $\fisherNe_{\delta,\delta}(\xi,\delta)$ and
$\fisherNlme_{\delta,\delta}(\xi,\delta)$ are given in the Appendix.
\end{theorem}
\begin{proof}
An outline proof is given in the Appendix, see Proposition \ref{mledistt} and section \ref{distscore}.
\end{proof}

\noindent\textit{\Note} \citet{Giraitis2001} do not find the limiting
distribution of their estimator, $\hat{\xi}_G$, of $\pole$. They note that this
is an artefact of the maximization over the specified grid.  This constraint is
not enforced in our approach.  Note that $\mbox{E}(|\widehat{\xi} -
\widehat{\xi}_G|)= O\left(N^{-1}\right)$ but that $N|\widehat{\xi} -
\widehat{\xi}_G|$ is {\em not} constrained to be zero or even to have a
tractable distribution, this result is demonstrated empirically in the
simulations. The convergence to the Cauchy for extreme values of $\delta$ is
quite slow, we provide, in the Appendix, a second approximation to the
distribution of the renormalized estimator of the pole, via more carefully
approximating the dominant contributions to the mean and variance of the
numerator and denominator that define the random variable the estimator
follows.

To compare the two large sample and asymptotic forms of the distributions, we
refer to Figure \ref{poly} (a) and (b). As $\delta$ increases in magnitude it
takes longer for the large sample approximation to be close the asymptotic
distribution, as is obvious from these plots. For a list of critical values of
the distribution see Table \ref{table1}.

The likelihood of the data changes in magnitude dramatically depending on the
value of $\xi$ and its alignment with the grid of frequencies at which the
periodogram is evaluated, determination of the best value of $\xi$ is pivotal
for characterizing the system, and must be the first stage of any analysis. For
completeness we now discuss the estimation of the additional parameters, i.e.
$\bm{\theta}$ and $\sigma^2_{\epsilon}$.

\subsection{White Noise Variance and Nuisance Parameters}
We now consider the estimation of the white noise component and regular
spectral component.  We model the sdf parametrically by
\begin{alignat}{1}
f(\lambda)&=\frac{\sigma^2_{\epsilon}\left|h(\lambda;\bm{\theta})\right|^2}{\left|2\cos\left(2\pi\lambda\right)
-2\cos\left(2\pi\xi\right) \right|^{2\delta}},\quad
f^{\dagger}(\lambda)=\frac{\left|h(\lambda;\bm{\theta})\right|^2\sigma^2_{\epsilon}\left|\lambda-\xi
\right|^{2\delta}}{\left|2\cos\left(2\pi\lambda\right)
-2\cos\left(2\pi\xi\right) \right|^{2\delta}}.
\end{alignat}
Differentiating $\ell(\pole,\delta,\bm{\theta},\sigma^2_{\epsilon})$ from
equation (\ref{approxlikeeqn}) with respect to $\sigma^2_{\epsilon}$ we obtain
that:
\begin{alignat}{1}
\nonumber
\frac{\partial \loglike(\xi,\delta,\bm{\theta},\sigma^2_{\epsilon})}{\partial
\sigma^2_{\epsilon}}
\nonumber
&=\sum_{j=J_1}^{J_2}\left\{-\frac{1}{\sigma^2_{\epsilon}}+\frac{\eta_jI_0
\left(\xi+j/N\right)}{\sigma^2_{\epsilon}} \right\}\\
\nonumber
\widehat{\sigma}^2_{\epsilon}/\sigma^2_{\epsilon}&=\frac{1}{-J_1+J_2+1}\sum_{j=J_1}^{J_2}
\widehat{\eta}_j I_0
\left(\xi+j/N\right)\\
\widehat{\eta}_j&=\frac{
\left|j\right|^{2\widehat{\delta}}}{N^{2\widehat{\delta}}}\frac{\left|2\cos\left(2\pi\lambda\right)
-2\cos\left(2\pi\xi\right) \right|^{2\widehat{\delta}}}
{\sigma^2_{\epsilon}|h(\lambda;\widehat{\bm{\theta}})|^2|\lambda-\xi
|^{2\widehat{\delta}}}.
\end{alignat}
Thus it follows that (Taylor expanding the other MLEs and using
their rates of convergence):
\begin{alignat}{1}
\widehat{\sigma}^2_{\epsilon}/\sigma^2_{\epsilon}&\overset{\mathcal{L}}{=}\frac{1}{2(-J_1+J_2+1)}
\chi^2_{2(-J_1+J_2+1)}+\o(1),
\end{alignat}
for $J_1,$ $J_2$ sufficiently large. This follows as the estimators of the
other parameters of the sdf are nearly unbiased for sufficiently large values
of $N$. We note that the covariance of $\widehat{\sigma}^2_{\epsilon}$ with
$\widehat{\delta}$ and $\widehat{\xi}$ can be treated analogously to the
results deriving the covariance of $\widehat{\xi}$ and $\delta$ or using
standard results for $\delta$ and/or $\bm{\theta}$ and
$\widehat{\sigma}^2_{\epsilon}$. Denoting by
\[
\hhat(\lambda) = h(\lambda;\widehat{\bm{\theta}}) \qquad
\dot{h}_i(\lambda;\bm{\theta})=\frac{\partial
\left|h(\lambda;\bm{\theta})\right|^2} {\partial
\theta_i}\qquad {\mathrm{and}} \qquad
\ddot{h}_{ik}(\lambda;\bm{\theta})=\frac{\partial^2
\left|h(\lambda;\bm{\theta})\right|^2} {\partial
\theta_i\partial \theta_k},
\]
we determine that
\begin{alignat}{1}
\nonumber \frac{\partial
\loglike(\xi,\delta,\bm{\theta},\sigma^2_{\epsilon})}{\partial \theta_i}
&=-\sum_{j=J_1}^{J_2}\frac{\dot{h}_i(\lambda_j;\bm{\theta})}
{|\hhat(\lambda_j)|^2}\left\{1-\eta_j I_0 \left(\lambda_j \right) \right\}
=\O(N),
\end{alignat}
and
\begin{alignat}{1}
\nonumber \frac{\partial^2
\loglike(\xi,\delta,\bm{\theta},\sigma^2_{\epsilon})}{\partial \theta_i\partial \theta_k}
&=\sum_{j=J_1}^{J_2}\left[-\frac{\ddot{h}_{ik}(\lambda_j;\bm{\theta})}
{|\hhat(\lambda_j)|^2}\left\{1-\eta_j I_0 \left(\lambda_j\right)
\right\}\right.\left.+\frac{\dot{h}_i(\lambda_j;\bm{\theta})\dot{h}_k(\lambda_j;\bm{\theta})}
{|\hhat(\lambda_j)|^4}\left(1-2\eta_j I_0 \left(\lambda_j\right) \right)
\right],
\end{alignat}
which is $\O(N).$ Thus we find that
\begin{alignat}{1}
\nonumber
{\mathrm{E}}\left\{\frac{\partial^2
\loglike(\xi,\delta,\bm{\theta},\sigma^2_{\epsilon})}{\partial \theta_i\partial
\theta_k} \right\}&=-\sum_{j=J_1}^{J_2}
\frac{\dot{h}_i(\lambda_j;\bm{\theta})\dot{h}_k(\lambda_j;\bm{\theta})}
{|\hhat(\lambda_j)|^4}=\breve{\mathbf{V}}^{-1}_{N,ik}\\
{\mathrm{E}}\left\{\frac{\partial^2
\loglike(\xi,\delta,\bm{\theta},\sigma^2_{\epsilon})}{\partial^2 \theta_i}
\right\}&=-\sum_{j=J_1}^{J_2} \frac{\dot{h}_i^2(\lambda_j;\bm{\theta})}
{|\hhat(\lambda_j)|^4}=\breve{\mathbf{V}}^{-1}_{N,ii}=\O(N).
\end{alignat}
This then provides the required score equations. Furthermore using regular ML
theory, we have that:
\begin{alignat}{1}
\sqrt{N}\left(\widehat{\bm{\theta}}-\bm{\theta}\right)\overset{\mathcal{
L}}{\longrightarrow} {\cal N}\left(\bm{0},\breve{\mathbf{V}}\right),
\end{alignat}
where $\breve{\mathbf{V}}$ contains the Fisher information, and
$N^{-1}\breve{\mathbf{V}}_{N}\rightarrow \breve{\mathbf{V}}$. This allows us to
fit the more general class of GARMA rather than Gegenbauer models, see
\citet{Gray1989}.

\section{Examples \label{examples}}

\vspace{0.1 in}

\subsection{Analysis of Simulated Data}
\label{SimExamp} For our simulation studies we examine the performance of our
adjusted Whittle likelihood-based estimators in comparison with those derived
from the classic Whittle likelihood.  Data were simulated in the time domain
using the covariance recursion formulae given in \cite{Lapsa1997} for a
seasonally persistent Gegenbauer process with $\pole = 1/7$ corresponding to an
weekly cycle in daily data, and $\delta = 0.3, 0.4$ and $0.45$. We generate
2000 replicate series of lengths $N=1024, 2048, 4096$ and 8192. Tables
\ref{simresults1} to \ref{simresults4} demonstrate the performance of the ML
estimators of $\pole$ and $\delta$, in terms of bias, variance, and the
relative efficiency ($\sigma^2_d/\sigma^2_W$) of our estimators compared with
those derived using the classic Whittle likelihood.  For $N=1024$ the
demodulated estimator significantly improves the bias present in the Whittle
estimators for both $\xi$ and $\delta$.  As $N$ increases and the spacing in
the Fourier grid decreases, both estimators for $\pole$ perform well, however,
the bias in the Whittle estimator for $\delta$ is still present even for
$N=8192$, and becomes more severe as $\delta$ increases.

To illustrate the problems with the Whittle likelihood for smaller $N$ and
large $\delta$, Figure \ref{simsplot} shows the mean conditional likelihoods
evaluated at the ML estimates. The improvement gained by demodulation is
evident, the scale of the improvement will be dependent on the distance of the
pole from the Fourier grid.  The plots also demonstrate the discontinuities in
the likelihood for $\pole$, as discussed in Section \ref{Likelihoodsection}.

\subsection{U.S. Weekly Crude Oil Imports}

\label{CrudeOil}

The first real data set comprises 756 observations of U.S. Weekly Crude Oil
Imports (in millions of barrels per day) from 6th December 1991 to 26th May
2006, downloaded from
\[
{\verb"http://tonto.eia.doe.gov/dnav/pet/hist/wcrimus2w.htm"}.
\]
The data were detrended using a linear trend, and are displayed in Figure
\ref{USPetrolFigure}.  Periodic behaviour is evident in the raw detrended data.

For these data, we fitted a low order Gegenbauer-ARMA (GARMA) model; the
process $\{X_t\}$ is represented as the unique stationary solution of
\[
\phi(B) X_t = \theta(B) G_t
\]
where $B$ is the backshift operator, and polynomial operators $\phi$ and
$\theta$ define an ARMA  process in the usual way, and where $\{G_t\}$ is a
pure Gegenbauer process as defined by the sdf in equation (\ref{giraitis}) with
$h$ the identity function. We consider at most ARMA(1,1) models, so that
$\phi(z) = 1 - \phi z$ and $\theta(z) = 1+\theta z$ where, under the
assumptions of stationarity and invertibility, $|\phi|, |\theta| < 1$.  Using
standard results, the parametric sdf that we consider takes the form
\[
f(\lambda) = \frac{\sigma^2_{\epsilon}}{\left|1-2e^{-2i\pi\lambda}\cos(2\pi
\pole)+e^{-4i\pi\lambda} \right|^{2\delta}} \frac{(1+\theta \cos(2\pi
\lambda)+\theta^2)}{(1-\phi \cos(2\pi \lambda)+\phi^2)}
\]
In our notation, a GARMA(1,1) model has both $\phi$ and $\theta$ non-zero; for
GARMA(1,0), $\theta \equiv 0$, whereas for GARMA(0,1), $\phi\equiv 0$.
GARMA(0,0) corresponds to the Gegenbauer model with no ARMA component.

\vspace{0.05 in}

\noindent \textbf{Results :}  Using numerical methods (the \texttt{optim}
function in \texttt{R}), each model was fitted using our demodulation approach
and also using the standard Whittle likelihood, and the results compared using
BIC.  The results are presented in Table \ref{USPetrolTable}.  The best model
is the overall is the GARMA(0,1) model fitted under demodulation, indicating
that the use of a non-standard Fourier grid can improve the quality of fit,
that is, the fit of the model under the standard derivation (we term this the
standard Whittle model) is inferior.

For the selected model, the parameter estimates and approximate standard errors
are displayed in Table \ref{USPetrolEsts}.

\subsection{Farallon temperature data}

\label{Farallon}

The second real data set is a surface temperature series for the shore station
at the Farallon Islands, California, United States.  Daily temperature data
were obtained from the ftp site
\[
{\verb"ftp://ccsweb1.ucsd.edu/shore/CURRENT_DATA/Temperature/"}
\]
and formed monthly averages for the period 1960-1996; missing daily
quantities were omitted from the monthly averages, whole missing
months (there were six in the period of study) were imputed by
taking averages for that calendar month across the 37 years of
study.  In total there were 444 monthly average observations.

We analyze these data in two ways to compare the Whittle maximum likelihood
estimates with our demodulation approach.  First, we take the 444 data in their
entirety, then we perform a second analysis using only the last 440
observations. As the expected annual periodicity would induce a pole in the
spectrum at frequency 1/12, and 12 divides 444, the pole will lie at a Fourier
frequency when the whole data set is analyzed.  However 12 does not divide 440,
so for the second analysis, the pole will not lie at a Fourier frequency.

\vspace{0.1 in}

\noindent \textbf{Results :}  Each of the low order GARMA models were fitted
and compared using BIC.  The two cases, $N=444$ and $N=440$ were analyzed.  The
results are presented in Table \ref{FarallonTable}, and the raw time series as
well as fitted models are plotted in Figure \ref{FarallonFigure}. The model
with the highest BIC is, in both cases, the GARMA(1,0), but for the two values
of $N$, the different approaches are favoured in the two cases.  For $N=444$,
the classic Whittle approach yields a higher log-likelihood, but for $N=440$
the demodulated model performs better, yielding a higher log-likelihood.
Parameter estimates from the model are presented in Table \ref{USPetrolEsts}
for the two values of $N$.

This data set and analysis illustrates perfectly another of the advantages of
using the demodulated likelihood with the bias-adjustment procedured outlined
in Section \ref{Likelihoodsection} and Theorem \ref{ApproxLike}. In the classic
Whittle likelihood, when a Fourier frequency exactly coincides with the pole,
the on-the-pole likelihood contribution erases the contribution of that
periodogram element.  Note first that the omission of a data point from the
likelihood causes the likelihood to increase (that is, become less negative)
and this explains the higher likelihood value for the classic Whittle
likelihood.  For the Farallon data set, this omission also leads the remaining
periodogram appearing as if it corresponded to a short memory process, hence
the low estimated value of $\delta$ that is essentially no different from zero.
The conclusion of such an analysis would be that the underlying process has a
pure seasonality at the estimated $\pole$, in this case $\pole = 1/12$, and the
seasonally differenced series was essentially a white noise process. However,
seasonal first differencing of the original series leads to a new series that
is not a white noise process; in fact the differenced series appears
over-differenced. Hence, such a model does not provide an adequate explanation
of the data.  When $N$ is changed to 440, inferences using the classic Whittle
method change dramatically.  Note, however, that for the new demodulated
likelihood, parameters estimates are closely comparable across different values
of $N$.

\subsection{Southern Oscillation Index}

\label{SOI} We consider the Southern Oscillation Index (SOI) data analyzed by,
for example, \citet{HuertaWest1999}.  The version of the data we consider has
$N=1668$, the data and fitted spectrum are presented in Figure \ref{SOIFigure}.
For this large sample size, the difference between the two approaches is
minimal; the BIC values are negligibly different, and the estimates and
estimated 95 \% intervals are presented in Table \ref{SOIEsts}. In this case,
the estimates of the pole position obtained from the likelihood approaches are
markedly different from the naive estimate obtained by taking the ordinate
corresponding to the maximum of the periodogram (shown as a dotted line in
Figure \ref{SOIb}).

\section{Implications for Non-Likelihood Approaches}

The results derived in previous sections focus explicitly on likelihood based
procedures.  However, they motivate the use of adjusted versions of currently
existing estimation procedures that improve the performance of those procedures
when applied to seasonally persistent series. Given the special role of the
location of the singularity when formulating the likelihood, we propose a
series of procedures that profit on the simplified distribution that arises by
using the demodulation by the (estimated) pole.

\subsection{Profile Likelihood} \label{PL}

The profile likelihood of $\pole$ is a pseudo-likelihood function given, for
each possible $\pole$, by
\begin{equation}
\label{ProfileLikelihood} \loglike_{N}^{P}\left( \pole\right) =\max_{\delta
,\bm{\theta}|\pole}\loglike_{N}\left( \pole,\delta ,\bm{\theta}\right) .
\end{equation}
$\loglike_{N}^{P}\left( \pole\right)$ may be maximized, yielding a maximum
pseudo-likelihood (MPL) estimate of $\pole$, denoted $\polehatPr$. Then the
values of $\delta$ and $\bm{\theta}$ which maximize the conditional likelihood
given $\pole = \polehatPr$, are computed. Specifically, the ML estimate of
$\delta$ based on the demodulated likelihood for \textbf{all} values of $\pole
\in (0,1/2)$ is computed. Finally, the estimate
$\widehat{\delta}=\widehat{\delta}(\polehatPr)$ based on demodulation at
$\polehatPr$ is obtained.

In many cases the MPL and ML estimators agree closely; in given applications,
the MPL approach may potentially be more readily implemented. Note that some
care must in generality be used when applying profile likelihood estimation,
(see, for example, \cite{Berger}), but given the rate of convergence of the MLE
of $\pole$ such problems are unlikely to arise.

\vspace{0.1 in}
\subsection{A Semi-Parametric Analysis: The Geweke-Porter-Hudak Estimator}

\label{SPA}

The Geweke-Porter-Hudak (GPH) procedure \citep{GPH1983} implements
semiparametric estimation of $\delta$ for the case $\pole=0$ which can be
adapted to incorporate a demodulation procedure and profile marginalization.
The GPH procedure examines the behaviour of the periodogram on the log scale
near frequency zero, and estimates the long memory parameter $\delta$ using
ordinary least squares and a linear regression.  We omit full details for
brevity, but outline a possible adjustment based on a recent formulation given
by \citet{HandS2004}.  It is sufficient to say that the GPH procedure relies on
distributional properties of the periodogram near the presumed pole.

In light of the results of earlier sections of this paper, to obtain an
improved estimate of the $\delta$ using GPH we could take two alternative
approaches. First, we could adjust the GPH to the demodulated setting, taking
the distribution of the periodogram at the singularity fully into account.
Alternatively, we could utilize large sample arguments and consider the score
function. We consider frequencies indexed $j$ whose likelihood contributions
are influenced by the singularity. The log-likelihood then has three
parameters; $\pole$, $\delta$, and
$C=f^{\dagger}\left(\pole;\pole,\delta,\bm{\theta},\sigma^2_{\epsilon}\right)$
which can be treated as a constant, if the $j$ included are chosen judiciously.

%
\citet[p. 58-59]{HandS2004} consider the modified GPH by
introducing the following notation:
\[
g\left(\omega\right) = -\log\left(\left|1-e^{i\omega}\right|\right)
\qquad \bar{g}_m =\frac{1}{m}\sum_{k=1}^m g\left(2\pi \freqk \right)
\qquad s_m^2  =  2\sum_{k=1}^m \left\{g\left(2\pi \freqk \right)-
\bar{g}_m\right\}^2
\]
where $m$ periodogram ordinates on either side of the pole are
included in the regression. They also define (with slightly
different notation) $a_k=s_m^{-2}\{g(2\pi \freqk)-\overline{g}_m\}$,
and define the estimator to be:
\begin{equation}
\label{GPHHS} {\widehat{\delta}}_{GPH}=\sum_{1\le \left|k\right|\le
m} a_k\log\{\Pgram (\freqk+\polehat_{S}) \}.
\end{equation}
Hidalgo and Soulier note that the asymptotic distribution of $\polehat_{S}$ is
not known, and that estimation of the pole is an open problem. In their
simulation studies, 5000 replications of series length 256, 512 and 1024 are
used, with $\pole=1/4$, and thus there is grid alignment with the pole. They
implement the GPH procedure, assuming $\polehat_{S}$ is correct on the
demodulated periodogram, excluding the contribution from the pole itself.
Notice also that they chose $m=N/4$, $m=N/8$ and $m=N/16$, i.e $m=\O(N)$.


Having found the distribution of the periodogram at the pole in Lemma
\ref{PeriodogramPole}, we can adjust the GPH estimator using a similar profile
likelihood approach.  Assume that $\pole$ is known, and consider $\lambda_k$
such that
\[
\frac{\Pgram(\lambda_k) \left| \lambda_k \right |^{2\delta}}{f^{\dagger
}(\lambda_k)} \sim \chi _{2}^{2}.
\]
As $\xi$ is known we are on the grid, and $B_{\lambda_k,N}
(\pole,\delta)=1+\O(k^{-1})$. Thus we may ignore the contributions of the
$\O(k^{-1})$ term, and omit this from the procedure as the terms sum to a
negligible contribution. Based on these values of $k$, least squares is then
used to estimate $\delta $. This requires knowledge of $\pole$; note that
Hidalgo and Soulier estimate $\pole$ as the Fourier frequency at which the
periodogram is maximized, and therefore are restricted to an $\O(N^{-1})$ grid.
In contrast, for any $\pole$, we demodulate the periodogram by $\pole$ giving
\[
\log\{
\Pgram(\lambda_k(\pole))\}-2\delta \log(| \lambda_k(\pole)
|)-\log(C)\sim\log(\chi_{2}^{2}),
\]
where $\lambda_k$ is calculated for the specified $\pole$, not necessarily on
the Fourier grid; in practice it is straightforward to use a finer grid over
which to do a systematic search.  More generally we may allow $\pole$
continuously across the interval $(0,1/2)$, and to use numerical routines, and
choose $\polehat$ to minimize the residual sum of squares after a least squares
fit. Not that we can approximate the distribution of the log periodogram
accordingly \textbf{only} if we demodulate, as otherwise the distribution is
shifted in location by a constant depending on $\pole$. In this case the
correlation between the periodogram at frequencies spaced $N^{-1}$ apart is
non-negligible, thus necessitating usage of weighted least squares.

\section{Discussion}

This paper has illustrated the inherent problems with seasonally persistent
processes and approximation based on the periodogram. We have demonstrated that
realigning the grid of frequencies at which the periodogram is evaluated will
simplify the distributional properties and enables us to specify a useful
approximation to a likelihood function.  Analysis of seasonal persistence will
usually be based on frequency domain descriptions. For the usual Fourier grid,
the distributional properties of the periodogram are generally not useful for
SPPs, if given by previously derived theory \citet{OMS2004}. This paper shows
how a small technical adjustment to the DFT to the DDFT alters the
distributional properties substantially, making analytic investigation of the
properties of the MLEs possible. The theoretical and practical utility of this
adjustment is apparently under-appreciated in the literature. Potentially, even
for short memory models (with bounded but highly peaked spectra) for moderate
values of $N$, there will be an advantage in demodulation.

In this paper, attempts have been made to fill the gaps of current theory. To
avoid the problems associated with the location of the singularity,
\citet{Giraitis2001} constrained the maximization of the Whittle likelihood to
a set of frequencies spaced $\O(N^{-1})$ apart, where the likelihood performs
well under the assumption that the true value of location of the singularity is
constrained to this set. Their important result states that
$\polehat-\polestar=\O_{p}\left( N^{-1} \right)$.  In fact, this is ensured
(informally) by picking a Fourier frequency a distance $C/N $ from the
singularity, hence not even necessarily the closest Fourier frequency. In
contrast, we have studied the sensitivity of the likelihood of the periodogram
to $\O(N^{-1})$ perturbations in $\pole$, and found that the estimate of
$\delta$ for large finite sample sizes is very sensitive to such variation,
thus clarifying that despite the very rapid convergence of $\polehat$ to
$\polestar$ the potential misalignment of the Fourier grid with the unknown
$\polestar$ must be acknowledged. We also derive the large sample form of the
distribution of $\widehat{\delta}$ and $\widehat{\pole},$ where the latter when
re-normalised appropriately has a scaled Cauchy distribution.

Our results relate to frequency domain based analysis at some grid of
frequencies. For processes with absolutely convergent autocovariance sequences
for large samples, no gain is made by a particular choice of Fourier domain
gridding, however for processes with seasonal persistence it is of fundamental
importance to chose the correct grid alignment, even in large samples, as this
simplifies the distributional results substantively.

Simulated examples show the superiority of our approach in finite sample
situations.  Furthermore, the methodology has the philosophical advantage of
acknowledging the estimation of $\pole$. While other methods do well
asymptotically for estimation of the long memory parameter, it is worth noting
that for any fixed (maybe large) sample-size, improvements can usually be found
by explicitly considering the estimation of $\pole$ separately. The profile
likelihood methods can be simply employed in the extended GPH estimator
discussed by \citet{HandS2004}, extending the ideas to semi-parametric models,
and facilitating a tractable analysis.

\bibliography{newbib2}
\bibliographystyle{chicagorssb}

\newpage

\renewcommand{\arraystretch}{1.3}

\begin{table}
\centering \caption{\label{simresults1}Demodulated ML estimates of $\pole$}
\begin{tabular}{|c|c|c|c|c|c|} \hline
N & $\delta$ & bias ($\times 10^{-5}$)
& sd($\times 10^{-3}$) & $\slfrac{\sigma_d^2}{\sigma_W^2}$ & 95\% interval \\ \hline
1024 & 0.30 & 0.8125 & 1.0218 & 0.7793 &(0.1406, 0.1453) \\
1024 & 0.40 & 1.2402 & 0.5829 & 0.8262 &(0.1416, 0.1442) \\
1024 & 0.45 & 0.6699 & 0.3574 & 0.6424 &(0.1421, 0.1435) \\
\hline
2048 & 0.30 &-4.4170 & 0.5355 & 0.7461 &(0.1414, 0.1439)\\
2048 & 0.40 & 0.1699 & 0.3011 & 0.6511 &(0.1422, 0.1435) \\
2048 & 0.45 & 0.7178 & 0.1958 & 0.5286 &(0.1425, 0.1433) \\
\hline
4096 & 0.30 & 1.1035 & 0.2770 & 0.9237 &(0.1422, 0.1436) \\
4096 & 0.40 & 0.4150 & 0.1459 & 0.9518 &(0.1426, 0.1432) \\
4096 & 0.45 &-0.5254 & 0.0941 & 1.0443 &(0.1426, 0.1431)\\
\hline
8192 & 0.30 & 0.9830 & 0.2031 & 0.9516 &(0.1424, 0.1433) \\
8192 & 0.40 & 0.2851 & 0.1450 & 0.9986 &(0.1426, 0.1432) \\
8192 & 0.45 & 0.1254 & 0.0722 & 1.0283 &(0.1426, 0.1431)\\
\hline
\end{tabular}
\end{table}

\begin{table}
\centering \caption{\label{simresults2}Whittle estimates of $\pole$}
\begin{tabular}{|c|c|r|r|c|c|} \hline
N & $\delta$ & bias($\times 10^{-5}$)
& sd($\times 10^{-3}$) & 95\% interval \\ \hline
1024 & 0.30 &  -2.3159 & 1.1575 &  (0.1406,0.1455) \\
1024 & 0.40 &   -9.8354 & 0.6413 & (0.1416,0.1445) \\
1024 & 0.45 &  -17.3549 & 0.4460 & (0.1416,0.1436) \\
\hline
2048 & 0.30 &   -1.7787 & 0.6196 & (0.1411,0.1440) \\
2048 & 0.40 &   3.5435 & 0.3731 & (0.1421,0.1436) \\
2048 & 0.45 &   7.6451 & 0.2694 & (0.1426,0.1436) \\
\hline
4096 & 0.30 &    0.2720 & 0.2882 & (0.1423,0.1436) \\
4096 & 0.40 &   -1.7299 & 0.1495 & (0.1426,0.1433) \\
4096 & 0.45 &   -3.4389 & 0.0920 & (0.1426,0.1431) \\
\hline
8192 & 0.30 &   -0.1842 & 0.2082 & (0.1424,0.1433) \\
8192 & 0.40 &   -0.2633 & 0.1451 & (0.1426,0.1432) \\
8192 & 0.45 &   -1.1890 &  0.0712 & (0.1426,0.1431) \\
\hline
\end{tabular}
\end{table}
\begin{table}
\caption{\label{simresults3}Demodulated ML estimates for $\delta$} \centering
\begin{tabular}{|c|c|c|r|r|c|c|} \hline
N & $\delta$ & mean & bias ($\times 10^{-4}$)
& sd ($\times 10^{-2}$) & $\slfrac{\sigma_d^2}{\sigma_W^2}$ & 95\% interval \\ \hline
1024 & 0.30 &  0.2990 & -9.6364 & 1.9200 & 0.8662 &(0.2642,0.3370)\\
1024 & 0.40 &  0.3995 & -4.9333 & 1.8672 & 0.7150&(0.3600,0.4334) \\
1024 & 0.45 &  0.4492 & -8.4182 & 1.6481 &0.5222 &(0.4136,0.4773)\\
\hline
2048 & 0.30 &  0.2998 & -1.7333 & 1.3862 &  0.8898&(0.2709,0.3267)\\
2048 & 0.40 &  0.3999 & -0.9515 & 1.3914 & 0.6902&(0.3715,0.4249)\\
2048 & 0.45 &  0.4499 & -0.8000 & 1.2573 & 0.4348&(0.4246,0.4736) \\
\hline
4096 & 0.30 &  0.3001 & 1.3333 & 0.9528 & 0.9523&(0.2818,0.3200) \\
4096 & 0.40 &  0.4003 & 2.9515 & 0.9475 &  0.8964&(0.3812,0.4188) \\
4096 & 0.45 &  0.4505 & 5.0612 & 0.9173 & 0.7909 &(0.4316,0.4663) \\
\hline
8192 & 0.30 &  0.3001 & 1.0872 & 0.5028 & 1.0068&(0.2912,0.3124) \\
8192 & 0.40 &  0.3997 & -2.5643 & 0.4655 & 0.9099 &(0.3906,0.4088) \\
8192 & 0.45 &  0.4497 & -3.0083 & 0.4056 &0.9313  &(0.4414,0.4576) \\
\hline
\end{tabular}
\end{table}
\begin{table}
\caption{ \label{simresults4}Whittle estimates for $\delta$} \centering
\begin{tabular}{|c|c|c|r|r|c|} \hline
N & $\delta$ & mean & bias ($\times 10^{-4}$)  & sd($\times 10^{-2}$) & 95\% interval \\ \hline
1024 & 0.30 &  0.3004 & 4.0909  & 2.0630 & (0.2606,0.3436) \\
1024 & 0.40 &  0.4076 & 76.1111 & 2.2083 & (0.3636,0.4505) \\
1024 & 0.45 &  0.4677 & 176.915 & 2.2807 & (0.4209,0.4997)\\
\hline
2048 & 0.30 &  0.3015 & 14.9899 & 1.4696 & (0.2717,0.3293) \\
2048 & 0.40 &  0.4077 & 76.8081 & 1.6748 & (0.3737,0.4414) \\
2048 & 0.45 &  0.4680 & 180.012 & 1.9068 & (0.4324,0.4997)\\
\hline
4096 & 0.30 &  0.3004 &  4.0204 & 0.9764 & (0.2816,0.3204) \\
4096 & 0.40 &  0.4018 & 17.5306 & 1.0008 & (0.3816,0.4205)\\
4096 & 0.45 &  0.4537 & 36.8571 & 1.0285 & (0.4337,0.4745)\\
\hline
8192 & 0.30 &  0.3003 & 3.3474 & 0.5011 & (0.2874,0.3133) \\
8192 & 0.40 &  0.4009 & 9.3895 & 0.4880 & (0.3900,0.4131)\\
8192 & 0.45 &  0.4522 &21.9531 & 0.4203 & (0.4401,0.4651)\\  \hline
\end{tabular}
\end{table}
\renewcommand{\arraystretch}{1}

\clearpage

\begin{table}
\centering \caption{\label{USPetrolTable}BIC values for U.S. Petroleum
Data}\vspace{0.05 in}
\begin{tabular}{|l|l|c|}
\hline
Method & Model & BIC\\
\hline
Demodulated & GARMA(0,0) & 71.246\\
& GARMA(1,0) & 32.211\\
& GARMA(0,1) & 27.153\\
& GARMA(1,1) & 30.921\\
\hline
Standard Whittle & GARMA(0,0) & 73.330\\
& GARMA(1,0) & 42.794\\
& GARMA(0,1) & 38.905\\
& GARMA(1,1) & 42.021\\
\hline
\end{tabular}
\end{table}

\begin{table}
\centering \caption{\label{USPetrolEsts}U.S. Petroleum data - GARMA(0,1)
parameter estimates}\vspace{0.05 in}
\begin{tabular}{|l|cccc|}
\hline
& & & & \\
& $\polehat \times 10^{-2} $ & $\deltahat$ & $\thetahat$ & $\widehat{\sigma}^2_{\epsilon}$ \\
\hline
Estimate          & 1.918         & 0.295 & -0.517 & 0.372\\
Approx 95 \% CI   & (1.762,1.956) & (0.221,0.384) & (-0.675,-0.360) & (0.337,0.412)\\
\hline
\end{tabular}
\end{table}

\begin{table}
\centering \caption{\label{FarallonTable}BIC values for Farallon
data.}\vspace{0.05 in}
\begin{tabular}{|l|l|c|c|}
\hline
&   & $N=444$ & $N=440$\\
Method & Model & BIC & BIC\\
\hline
Demodulated & GARMA(0,0) & 274.918 & 272.710\\
& GARMA(1,0) & 257.080& 254.612\\
& GARMA(0,1) & 266.528& 262.589\\
& GARMA(1,1) & 262.474& 259.754\\
\hline
Standard & GARMA(0,0) & 278.567&281.290\\
& GARMA(1,0) & 243.009&260.035\\
& GARMA(0,1) & 265.769&274.484\\
& GARMA(1,1) & 246.742&265.129\\
\hline
\end{tabular}
\end{table}

\begin{table}
\centering \caption{\label{FarallonEsts}Farallon data - GARMA(1,0) parameter
estimates.}\vspace{0.05 in}
\begin{tabular}{|lll|cccc|}
\hline
& & & & & &\\
& & & $\polehat \times 10^{-2} $ & $\deltahat$ & $\phihat$ & $\widehat{\sigma}^2_{\epsilon}$ \\
\hline
Demodulated & $N=444$ & Estimate &  8.358 & 0.221 & 0.628 & 0.431\\
            &         & 95 \% CI & (8.206,8.438) & (0.157,0.314) & (0.558,0.726) & (0.266,0.562)\\[6pt]
            & $N=440$ & Estimate &  8.295 & 0.234 & 0.644 & 0.401\\
            &         & 95 \% CI & (8.290,8.391) & (0.156,0.311) & (0.558,0.728) & (0.252,0.556)\\
\hline
Standard    & $N=440$ & Estimate &  8.409 & 0.156 & 0.629 & 0.520 \\
            &         & 95 \% CI & (8.222,8.497) & (0.133,0.305) & (0.562,0.736) & (0.286,0.594)\\
\hline
\end{tabular}
\end{table}

\begin{table}
\centering \caption{\label{SOIEsts}SOI data - GARMA(0,0) parameter
estimates}\vspace{0.05 in}
\begin{tabular}{|ll|ccc|}
\hline
& & & & \\
& & $\polehat \times 10^{-2} $ & $\deltahat$ & $\widehat{\sigma}^2_{\epsilon}$ \\
\hline
Demodulated & Estimate          & 2.366         & 0.237 & 0.782\\
            & Approx 95 \% CI   & (1.452,2.399) & (0.215,0.254) & (0.728,0.833)\\
\hline
Standard    & Estimate          & 2.247         & 0.235 & 0.778\\
            & Approx 95 \% CI   & (1.402,2.381) & (0.215,0.255) & (0.730,0.833)\\
\hline
\end{tabular}
\end{table}

\newpage

\begin{figure}
  \centering
    \subfigure[$\delta=0.40$]{\label{fig:poly-a}\includegraphics[scale=1.5]{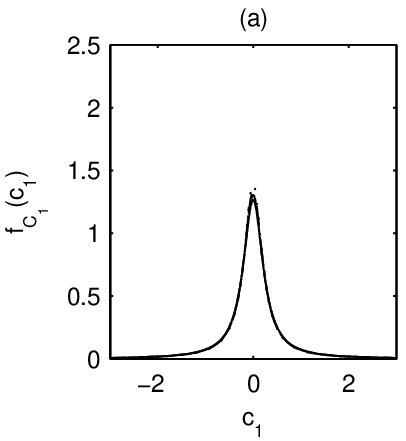}}
    \subfigure[$\delta=0.45$]{\label{fig:poly-b}\includegraphics[scale=1.5]{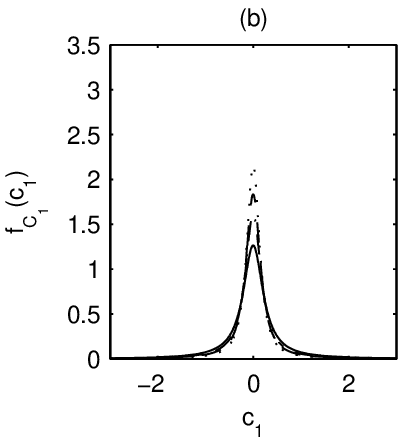}}
  \caption{\label{poly}Simulated Data:
The finite $N$ approximation to the distribution of $N(\widehat{\pole}-\polestar)$ for
(a) $\delta=0.4$ and (b) $\delta=0.45$. The dotted and
dash-dotted curves give the proposed finite large sample approximation for
different values of $N$ whilst the solid line gives the Cauchy asymptotic form. It is clear from the plot that for
large values of $\delta$ the distribution is quite slow converge to the Cauchy.}
\end{figure}

\begin{figure}
\centering
\subfigure[Raw Data]{\label{USPeta} {\rotatebox{-90}{\includegraphics[scale=0.4]{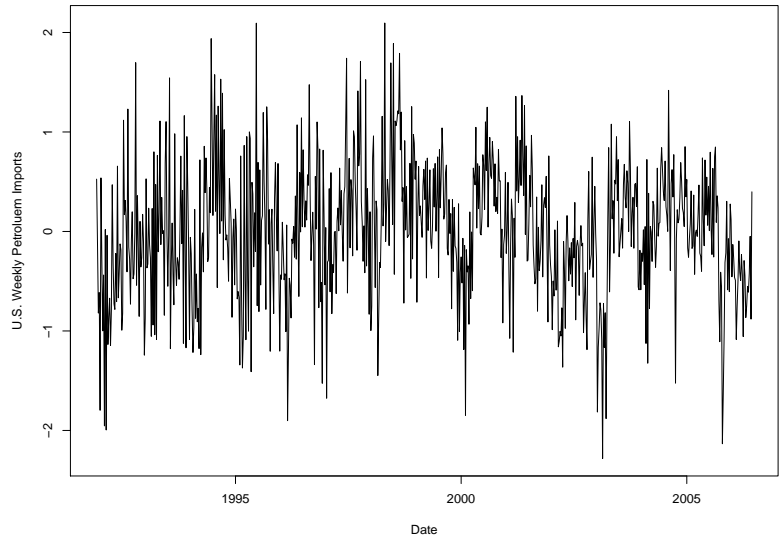}}}}\\
\subfigure[Fit of two GARMA models]{\label{USPetb} {\rotatebox{-90}{\includegraphics[scale=0.4]{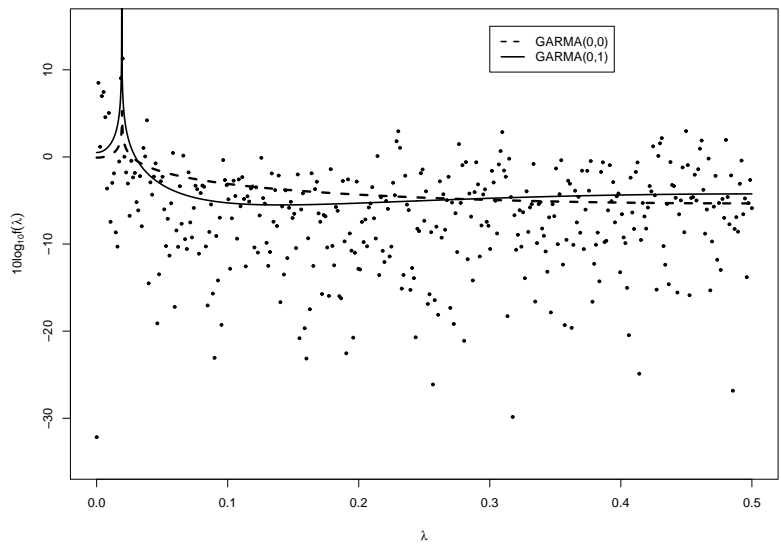}}}}\\
\caption{U.S. Petroleum Data: Raw data and spectral fits of GARMA(0,0) and GARMA(0,1) models.  The GARMA(0,1) model
yields a lower BIC value.}
\label{USPetrolFigure}
\end{figure}

\begin{figure}
\centering
\subfigure[Raw Data]{\label{Farallona} {\rotatebox{-90}{\includegraphics[scale=0.4]{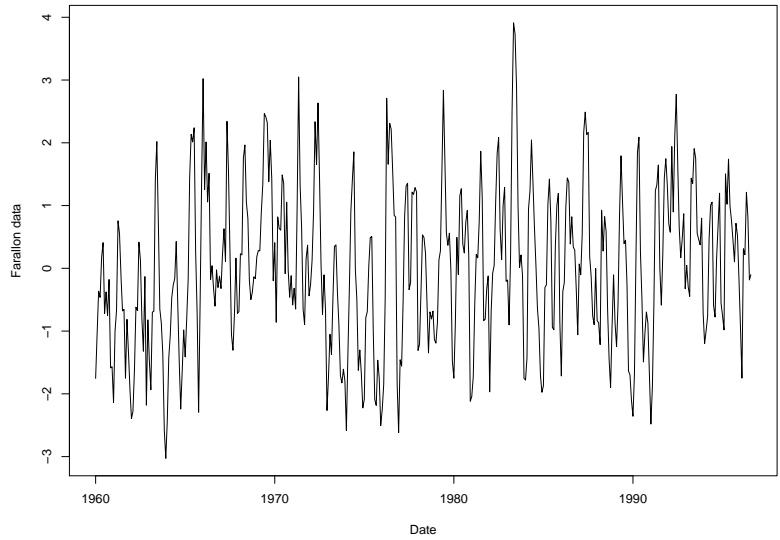}}}}\\
\subfigure[Fit of two GARMA models]{\label{Farallonb} {\rotatebox{-90}{\includegraphics[scale=0.4]{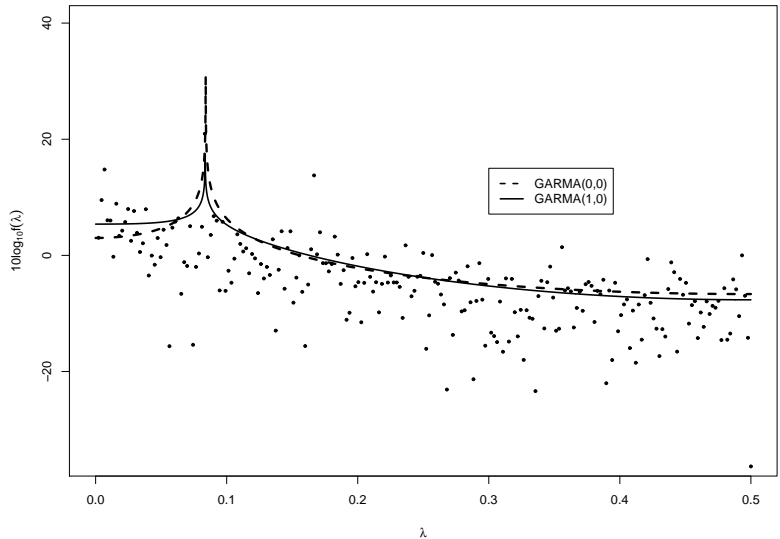}}}}\\
\caption{Farallon data: Raw data and spectral fits of of GARMA(0,0) and GARMA(1,0) models to the data
set with $N=440$ observations.  The fit of the models using the Whittle likelihood are similar, but
inferior in BIC terms.}
\label{FarallonFigure}
\end{figure}

\begin{figure}
\centering
\subfigure[Raw Data]{\label{SOIa} {\rotatebox{-90}{\includegraphics[scale=0.4]{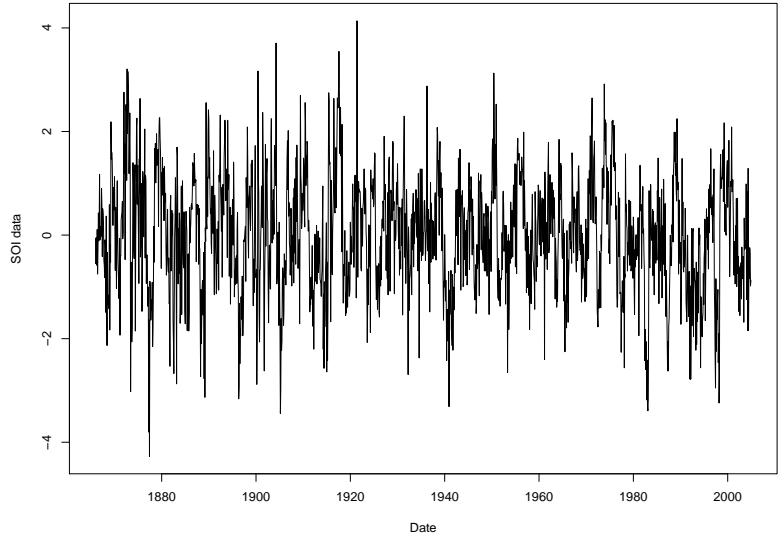}}}}\\
\subfigure[Fit of GARMA models under standard and demodulated approaches]{\label{SOIb} {\rotatebox{-90}{\includegraphics[scale=0.4]{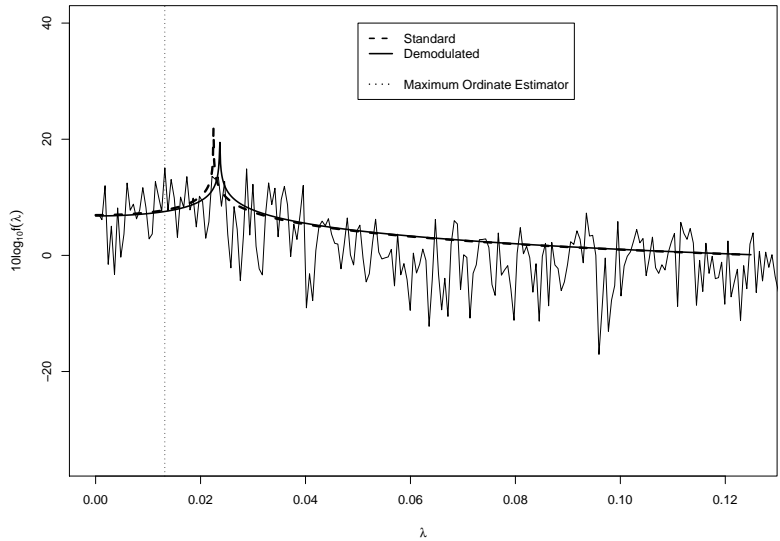}}}}\\
\caption{Southern Oscillation Index data: Raw data and spectral fits of GARMA(0,0) model under standard Whittle and demodulation.  For comparison, the
estimator that takes the maximum periodogram ordinate as the estimate is also displayed.}
\label{SOIFigure}
\end{figure}

\clearpage

\newpage

\setcounter{section}{0} \setcounter{equation}{0}
\renewcommand{\thesection}{{\Alph{section}}}
\renewcommand{\thesubsection}{\Alph{section}.\arabic{subsection}}
\renewcommand{\thesubsubsection}{\Alph{section}.\arabic{subsection}.\arabic{subsubsection}}
\renewcommand{\theequation}{\Alph{section}-\arabic{equation}}
\section{Appendix: {\em Proofs}}

{ \small
\input{appendixa.tex}

\input{appendixb.tex}

\input{appendixc.tex}

\end{document}

%% file: defs2.tex
\def\O{\text{O}}
\def\o{\text{o}}

\def\pole{\xi}
\def\polehat{\widehat{\xi}}
\def\polehatPr{\widehat{\xi}_{\text{Pr}}}

\def\polestar{{\xi^\star}}

\def\Polebias{\B_\xi(\xi,\delta)}
\def\Polestarbias{\B_{\polestar}(\polestar,\deltastar)}
\def\freqj{{\varphi_j}}
\def\freqk{{\varphi_k}}
\def\freql{{\varphi_l}}

\def\freqdj{{\lambda_j}}

\def\lamj{\lambda_j}
\def\lamk{\lambda_k}
\def\lam0{\lambda_0}

\def\deltahat{\widehat{\delta}}
\def\deltastar{\delta^\star}

\def\chival{\bm{\psi}}
\def\chihat{\widehat{\bm{\psi}}}
\def\chistar{\bm{\psi}^{\star}}
\def\tildechistar{\widetilde{\bm{\psi}}^{\textrm{\raisebox{-1.0ex}{$\star$}}}}
\def\chiprime{{\bm{\psi}^\prime}}

\def\ldot{\dot{\ell}}
\def\lddot{\ddot{\ell}}
\def\stscore{\bm{k}_N}
\def\limstscore{\bm{k}}

\def\hhat{\widehat{h}}

\def\Pgram{I_0}

\def\Pgrm{I}
\def\Idot{\dot{I}}

\def\Istd{\cal{I}^{(f,N)}}
\def\Istddot{\dot{\cal{I}}^{(f,N)}}
\def\Istdddot{\ddot{\cal{I}}^{(f,N)}}

\def\standA{A^{(f,N)}}
\def\standAs{A^{(f,N)2}}
\def\standB{B^{(f,N)}}
\def\standBs{B^{(f,N)2}}
\def\standC{C^{(f,N)}}
\def\standCs{C^{(f,N)2}}
\def\standD{D^{(f,N)}}
\def\standDs{D^{(f,N)2}}
\def\standE{E^{(f,N)}}
\def\standF{F^{(f,N)}}
\def\standG{G^{(f,N)}}
\def\standH{H^{(f,N)}}
\def\standU{U^{(f,N)}
\def\standUconj{U^{(f,N)*}}}

\def\kap{\kappa_{0,N}(\xi)}
\def\kapj{\kappa_{j_p,N}(\xi,\xi^{\ast})}

\def\fisherN{\bm{{\cal{F}}}_N}
\def\fisherNe{{{\cal{F}}}^{(N)}}

\def\fisherNlme{{{\cal{F}}}}

\def\obsfisherN{\bm{{{F}}}_{N}}
\def\obsfisherstdN{\bm{{{W}}}_{N}}

\def\pole{\xi}
\def\Xt{\{X_t\}}
\def\polehat{\widehat{\xi}}
\def\polestar{\xi^\star}

\def\thetahat{\widehat{\theta}}
\def\phihat{\widehat{\phi}}

\def\freqj{{\varphi_j}}
\def\freqk{{\varphi_k}}
\def\freql{{\varphi_l}}

\def\freqdj{{\lambda_j}}

\def\lamj{\lambda_j}
\def\lamk{\lambda_k}
\def\lam0{\lambda_0}

\def\deltahat{\widehat{\delta}}
\def\deltastar{\delta^\star}
\def\chival{\bm{\psi}}
\def\chihat{\widehat{\bm{\psi}}}
\def\chistar{\bm{\psi}^{\star}}
\def\chiprime{{\bm{\psi}^\prime}}

\def\stscore{\bm{k}_N}
\def\limstscore{\bm{k}}

\def\Istd{{\mathcal{I}}^{(f,N)}}

\def\Idot{\dot{I}}
\def\Istddot{\dot{\mathcal{I}}^{(f,N)}}

\def\Pgram{I_0}

\def\ddotPgramj{{\ddot{I}_{0j}}}
\def\ddotPgramzero{{\ddot{I}_{00}}}
\def\loglike{\ell}
\def\transpose{\top}

\def\kap{\kappa_{0}}
\def\kapj{\kappa_{j_0}}

\def\fisherN{\bm{{\cal{F}}}_N}
\def\fisherNe{{{\cal{F}}}^{(N)}}
\def\hatfisherNe{{\widehat{\cal{F}}}^{(N)}}

\def\fisherNlme{{{\cal{F}}}}

\def\obsfisherN{\bm{{{F}}}_{N}}
\def\obsfisherstdN{\bm{{{W}}}_{N}}

\newtheorem{case}{Case}

\newtheorem{Note}[theorem]{Note}

\newtheorem{defn0}{Definition}[section]
\newtheorem{case0}{Case}
\newtheorem{lemma0}{Lemma}[section]
\newtheorem{theorem0}[theorem]{Theorem}
\newtheorem{prop0}{Proposition}

\renewenvironment{lemma}{\medskip \flushleft \begin{lemma0}\rm}{\end{lemma0}}
\renewenvironment{case}{\medskip \flushleft \begin{case0}\rm}{\end{case0}}
\renewenvironment{theorem}{\medskip \flushleft \begin{theorem0}  }{\end{theorem0}}
\renewenvironment{proposition}{\medskip \flushleft \begin{prop0}}{\end{prop0}}
\renewenvironment{Note}[1][Note]{\flushleft \textbf{#1 :} }{ }

\newcommand{\B}{\mbox {\boldmath$B$} }

\newcommand {\be} {\begin {equation} }
\newcommand {\ee} {\end {equation} }
\newcommand {\bmath} {\begin {displaymath} }
\newcommand {\emath} {\end {displaymath} }

\newcounter{romnum}

\newcommand{\slfrac}[2]{
\hspace{3pt}\!\!\!^{#1}\!\!\hspace{1pt}/\hspace{2pt}\!\!_{#2}\!\!\!\hspace{3pt}
}

%% file: appendixa.tex
\subsection{Expectation of the Periodogram at the Pole}
\label{PeriodogramAppendix} Starting with the same method of calculation as in
\citet[p. 623]{OMS2004} we find a large $N$ approximation to the expected value
of the periodogram at $\pole$, after demodulation via $\pole$.  We have
\begin{eqnarray}\nonumber
{\mathrm{E}}\left\{\frac{\Pgram(\pole)}{N^{2\delta}f^{\dagger}(\pole)}\right\}&=& \frac{1}{
N^{2\delta}f^{\dagger}(\pole)}\int_{\pole-\frac{1}{\surd{N}}}^{\pole+\frac{1}{\surd{N}}
} \frac{f^{\dagger}(\lambda)}{|\lambda-\pole|^{2\delta}}
\frac{\sin^2\{\pi N(\lambda-\pole)\}}{N
\sin^2\{\pi (\lambda-\pole)\} }\;d\lambda+o(1) \\\nonumber
\nonumber\\
&=& \frac{1}{f^{\dagger}(\pole)}\int_{-\surd{N}}^{\surd{N}} \frac{
f^{\dagger}(\pole+u/N)}{|u|^{2\delta}} \frac{%
\sin^2(\pi u)}{ N \sin^2(\pi u/N)
}\;\frac{du}{N} \nonumber+o(1)\\\nonumber
\\
&=& \nonumber\frac{1}{\pi^2}\int_{-\infty}^{\infty}
 \frac{\sin^2(\pi u)}{%
|u|^{2\delta+2}}\;du+\o(1) =
\frac{2}{\pi^2}\int_{0}^{\infty} \frac{
\sin^2(\pi u)}{u^{2\delta+2}}\;du +\o(1)\\
\nonumber\\
&=& -\Gamma(-1-2\delta)\cos\{\pi(
1/2+\delta)\}2^{2\delta+1}\pi^{2\delta-1} + \o(1)
 = \Polebias + \o(1) \label{refsing}
\end{eqnarray}
This result implicitly defines $\Polebias$.  The final line follows from
\citet[\S 3.823]{Gradshteyn}. From equation (\ref{giraitis}),
\[f^{\dagger}(\pole)= \frac{\sigma_{\epsilon}^2
\left|h\left(\pole;\bm{\theta}\right)\right|^2}{\left\{4\pi\left|
\sin (2\pi \pole )\right|\right\}^{2\delta}}=f_0^{\dagger},
\]
as $\lambda \rightarrow \pole$. Thus after demodulation, the expectation at the
singularity is given by equation (\ref{refsing}):
\[
{\mathrm{E}}\left\{I_0 (\pole )\right\}\bumpeq-N^{2\delta}\frac{2\Gamma\left(-1-2%
\delta\right) \cos\left\{\pi\left(\frac{1}{2}+\delta\right)\right\}
\sigma_{\epsilon}^2\left|h\left(\pole;\bm{\theta}\right)\right|^2}{\pi\left\{2\left|
\sin (2\pi \pole )\right|\right\}^{2\delta}}+\o\left(N^{2\delta} \right).
\]

%% file: appendixb.tex
\subsection{Bounding the Covariance Contributions}
Under Gaussianity of the original time series, the DDFT will also be jointly
proper complex Gaussian, thus we only need only to approximate for large $N$
the first and second order joint properties of these variables; the zeroth
order properties are given in Appendix \ref{PeriodogramAppendix}, in
conjunction with the results in \cite{OMS2004}.

We consider the discontinuities of the likelihood of the DDFT coefficients
explicitly, and also the effects of ignoring the weak correlation between the
Fourier coefficients near the pole (see Robinson (1995)). It is easier to deal
with the demodulated sequence only, and so we shall only evaluate the frequency
domain quantities at frequencies $\freqdj$ from (\ref{lambdagrid}). Let
$\Pgrm_j=\Pgram(\freqdj),$ and take $A_j=A_{0}(\freqdj)$ and
$B_j=B_{0}(\freqdj).$  As we only consider demodulation by $\lambda_D$ we in
this section suppress the subscript $D$. We note that with $i'=i/2$ for $i$
even and $i'=(i-1)/2$ for $i$ odd, and similarly for $j$, then
\begin{eqnarray*}
(\bm{\Sigma}_{\bm{C}_{\lambda}})_{i,j} &=&\left\{\begin{array}{ll}
\frac{1}{2}B_{\lambda_{i'},N}(\xi,\delta)f(\lambda_{i'}) &i=j\\[6pt]
0 & (i-j)\mod 2=1\\[6pt]
V_{i',j',N}\left(\frac{j_{0,N}(\xi)}{N},\delta\right)\sqrt{f(\lambda_{i'})f(\lambda_{j'})} +
\o(1)\sqrt{f(\lambda_{i'})f(\lambda_{j'})}
 &(i-j)\mod 2=0,\;i\neq j\\[6pt]
\end{array} \right.
\end{eqnarray*}
Let $\bm{D}_{\lambda}={\mathrm{diag}}\left(
\sqrt{\frac{1}{2}B_{\lambda_{J_1},N}(\xi,\delta)f(\lambda_{J_1})} \;\dots \;
\sqrt{\frac{1}{2}B_{\lambda_{J_2},N}(\xi,\delta)f(\lambda_{J_2})} \right)$, and
let $ \tilde{\bm{\Sigma}}_{\bm{C}_{\lambda}}=\bm{D}_{\lambda}^{-1}
\bm{\Sigma}_{\bm{C}_{\lambda}}\bm{D}_{{\lambda}}^{-1}. \nonumber $ Twice the
log-likelihood based on the sample $\bm{C}_{\lambda}$ takes the form:
\begin{eqnarray}
\nonumber
2\ell_N^{(f)}(\xi,\delta,\bm{\theta},\sigma^2_{\epsilon})
&=&-N\log(2\pi)-\log|\bm{\Sigma}_{\bm{C}_{\lambda}}|-\bm{C}_{\lambda}^\transpose
\bm{\Sigma}_{\bm{C}_{\lambda}}^{-1}\bm{C}_{\lambda}+o(N)\\
&=&-N\log(2\pi)-\log|\bm{D}_{\lambda}^2|-\bm{C}_{\lambda}^\transpose
\bm{D}_{\lambda}^{-2}\bm{C}_{\lambda}
+R(\bm{\theta},\pole,\delta,\sigma^2_{\epsilon})\nonumber+o(N)\\
&=&2\ell(\xi,\delta,\bm{\theta},\sigma^2_{\epsilon})+R(\bm{\theta},\pole,\delta,\sigma^2_{\epsilon})
+\o(N),
\label{newthing1}
\end{eqnarray}
where
\begin{equation*}
R(\bm{\theta},\pole,\delta,\sigma^2_{\epsilon})=-\log|\tilde{\bm{\Sigma}}_{\bm{C}_{\lambda}}|-
\bm{C}_{\lambda}^\transpose\bm{D}_{\lambda}^{-1}(
\tilde{\bm{\Sigma}}_{\bm{C}_{\lambda}}^{-1}-\bm{I}_2)\bm{D}_{\lambda}^{-1}
\bm{C}_{\lambda}.
\end{equation*}
Note that $2\ell(\xi,\delta,\bm{\theta},\sigma^2_{\epsilon})=\O(N).$ Also, $
\log|\tilde{\bm{\Sigma}}_{\bm{C}_{\lambda}}|=\log\{\O(1)\}. $ The latter
statement holds as the magnitude of this object can be bounded by considering
the trace of the matrix $\tilde{\bm{\Sigma}}_{\bm{C}_{\lambda}},$ and the fact
that for $\log(N)<k<j$ the covariance terms can be bounded by $k^{-1} \log(j)$
({\em cf} \cite{Robinson1995a}). If, for $\log(N)<k<j$, we consider the terms
in the log-likelihood involving $A_j A_k$ and $B_j B_k$, then these are
$\O\{k^{-2} \log^2(j)\}$. The higher order terms are obtained by inverting the
covariance matrix, and the second term coming directly from the order of the
contributions. We write ${\mathrm{E}}\left\{A_j
A_k\right\}={\mathrm{E}}\left\{B_j B_k\right\}= T_{jk} k^{-2} \log^2(j)$, for
$T_{jk}=\O(1)$ and let $\overline{T}=\max_j \max_k |T_{jk}|$.

When summing the covariance terms we need to split up the terms indexed by
negative and positive $j$ into two sum. Consider one of the two sums, and sum
the contributions over indices $\log(N)<k<j<J=\O(N)$, denoting the sum $R_2$.
To formally derive this for contributions to the left and right of the pole, we
can use twice this term, and the order of the contributions are the most
important result. Then we note that using Minkowski inequality arguments:
\begin{eqnarray*}
\frac{1}{N}|R_2|&\le&\frac{1}{N}\sum_{j=\log(N)}^{J}\sum_{k=\log(N)}^{j}\overline{T} \frac{\log^2(j)}{k^2}
=\frac{\overline{T}}{N^2}\sum_{j=\log(N)}^{J}  \frac{\log^2(j)}{N}\sum_{k=\log(N)}^{j} \frac{N^2}{k^2}\\
&=&\frac{1}{N}\int_{\log(N)/N}^{J/N} \left\{\log^2(x)+2\log(x)\log(N)+\log^2(N)\right\}
\left\{\frac{N}{\log(N)}-\frac{1}{x} \right\}\;dx+\o(1)\\
&=&\frac{1}{\log(N)}\left[ x\{\log^2x-2\log(x)+1\}+2x\{\log (x)-1\}\log(N)+
\log^2(N)x
\right]_{\log(N)/N}^{J/N}\\
&&-\left[\frac{\log^3(x)}{3}+\log(N)\log^2(x)+\log(x)\log^2(N) \right]_{\log(N)/N}^{J/N} = \o(1).
\end{eqnarray*}
Note that $A_j B_k$ is for any choice of $j$ and $k$, $\o(1)$. Thus as
$(2\ell_N^{(f)}(\xi,\delta,\bm{\theta},\sigma^2_{\epsilon}))/N$ is $\O(1)$ we
can ignore the covariance contributions. Asymptotically, using the likelihood
from equation (\ref{approxlikeeqn}) yields equivalent results to using the
likelihood constructed from independent exponential random variables with
non-equal variances, due to the weak correlation between the Fourier
coefficients.

\subsection{Additional Notation}
Define $\chival=\left(\pole,\delta\right)^\transpose,$ and denote the true
values of the parameters by $\chistar.$ We suppress the dependence on other
parameters, i.e. the dependence on $\bm{\theta}$ and $\sigma^2_{\epsilon}$.
Consider first expansions of the log-likelihood defined by equation
(\ref{approxlikeeqn}), $\ell\left(\chival \right)$. Let
\[ \lddot= \left(\begin{array}{cc}
\ell_{\pole,\pole}\left(\chival\right) &
\ell_{\delta,\pole}\left(\chival\right)
\\
\ell_{\delta,\pole}\left(\chival\right) &
\ell_{\delta,\delta}\left(\chival\right)
\end{array}
\right)=-\obsfisherN(\chival),\quad
{\mathrm{E}}\left\{\obsfisherN(\chival)\right\}= \fisherN(\chival).
\]
denote the matrix of second partial derivatives. Furthermore, it is convenient
to introduce additional random variables, required to study the properties of
the score and the observed Fisher information. We denote by $\Idot_j$ and
$\ddotPgramj$ the quantities $\dot{I}_0(\freqdj)$ and $\ddot{I}_0(\freqdj)$
respectively, and by $\Istd_j,$ $\Istddot_j$ and $\Istdddot_j$ the standardized
periodogram, derivative of the periodogram wrt the $\pole$ and the second
derivative of the periodogram wrt to the $\pole,$ all evaluated on the shifted
grid.  Then
{\small
\[
\Istd_j= \left\{ \begin{array}{c} \dfrac{I_0 (
\freqdj)}{ f\left(\lamj\right)},  \\[9pt]
\dfrac{I_0(\lambda_{0})}{\Polebias N^{2\delta }f^{\dagger}_0},
\end{array} \right.\;\; \Istddot_j= \left\{ \begin{array}{c} \dfrac{\dot{I}_0
(\lamj)}{N f\left(\freqdj\right)}, \\[9pt]
\dfrac{\dot{I}_0(\lambda_{0})}{\Polebias N^{2\delta+1 }f^{\dagger}_0},
\end{array} \right.\;\; \Istdddot_j= \left\{ \begin{array}{cl} \dfrac{\ddot{I}_0
(\freqdj)}{N^2 f\left(\lamj\right)} & j \neq 0, \\[10pt]
\dfrac{\ddot{I}_0(\lambda_{0})}{\Polebias N^{2\delta+2 }f^{\dagger}_0} & j = 0.
\end{array} \right.
\]}
These quantities can be written in terms of the real and imaginary part of the
DDFT and its derivatives, and so we define for $j=J_1,\dots,J_2:$
\begin{equation}
\label{trig}
\begin{array}{cc}
A_j=\dfrac{1}{\sqrt{N}}\sum_t X_t \cos\left\{2\pi (\pole+j/N)t\right\}
&
B_j=\dfrac{1}{\sqrt{N}}\sum_t X_t \sin\left\{2\pi (\pole+j/N)t\right\}\\
C_j=\dfrac{1}{\sqrt{N}}\sum_t t X_t \cos\left\{2\pi
(\pole+j/N)t\right\} &
D_j=\dfrac{1}{\sqrt{N}}\sum_t t X_t \sin\left\{2\pi (\pole+j/N)t\right\} \\
E_j=\dfrac{1}{\sqrt{N}}\sum_t t^2 X_t \cos\left\{2\pi
(\pole+j/N)t\right\}&
F_j=\dfrac{1}{\sqrt{N}}\sum_t t^2 X_t \sin\left\{2\pi (\pole+j/N)t\right\}\\
G_j=\dfrac{1}{\sqrt{N}}\sum_t t^3 X_t \cos\left\{2\pi
(\pole+j/N)t\right\}& H_j=\dfrac{1}{\sqrt{N}}\sum_t t^3 X_t
\sin\left\{2\pi (\pole+j/N)t\right\},
\end{array}
\end{equation}
for $j = J_1, \ldots, J_2$, where the sum over $t$ ranges over $t=0,\dots,N-1$.
Also, let
\[\begin{array}{cccc}
\standA_j=\dfrac{A_j}{\sqrt{f\left(\lamj\right)}} &
\standB_j=\dfrac{B_j}{\sqrt{f\left(\lamj\right)}} &
\standC_j=\dfrac{C_j}{N\sqrt{f\left(\lamj\right)}} &
\standD_j=\dfrac{D_j}{N\sqrt{f\left(\lamj\right)}}\\
\\
\standE_j=\dfrac{E_j}{N^{2}\sqrt{f\left(\lamj\right)}} &
\standF_j=\dfrac{F_j}{N^{2}\sqrt{f\left(\lamj\right)}} &
\standG_j=\dfrac{G_j}{N^3\sqrt{f\left(\lamj\right)}} &
\standH_j=\dfrac{H_j}{N^3\sqrt{f\left(\lamj\right)}}
\end{array}.
\]
be the corresponding suitably standardized quantities. We shall also derive
expressions for the expectation of $\Istd_j,$ $\Istddot_j$ and $\Istdddot_j$
and these will be denoted $B_{\lambda_j,N},$ $\dot{B}_{\lambda_{j},N},$ and
$\ddot{B}_{\lambda_j,N},$ respectively. Their variances take quite complicated
forms, and we denote the theoretical constants that give their forms for
$\Istddot_j$ and $\Istdddot_j$ via $\dot{C}_{\lambda_j,N},$
$\ddot{C}_{\lambda_j,N}^{(1)}$ and $\ddot{C}_{\lambda_j,N}^{(2)},$ where the
first of these terms is a rough approximation to the variance of $\Istddot_j.$
More details follow later in the text when appropriate. Furthermore, the
covariances of the $j$th and $k$th DDFT coefficients and their derivatives, are
denoted by $V_{\lambda_j,\lambda_k,N}$, $\dot{V}_{\lambda_j,\lambda_k,N}$ and
$W_{\lambda_j,\lambda_k,N}$ respectively.

\subsection{Zeroth Order Properties}
To acknowledge the dependence of the likelihood on the indices
$J_1=-j_{0,N}(\pole)+1,$ and $J_2=M-1-j_{0,N}(\pole)$, and the fact that these
indices depend on $\pole$, we thus in this section write explicitly
$\ell\left(\chival,J_1,J_2\right)$. Note that $J_1<0.$ For any finite value of
$N$ this dependence introduces a discontinuity in the log-likelihood in the
form of a jump when the demodulation makes the range to the left decrease by
one, and the range on the right increase by one, or vice-versa. This fact is
inconvenient for our calculations, as it makes the log-likelihood discontinuous
and hence not differentiable. However, it transpires that the magnitude of the
discontinuities are of an order that can be ignored for large sample sizes, as
will be shown by the first proposition, so that subsequent calculations will be
in terms of $\ell(\chival),$ where $J_1$ and $J_2$ are treated as {\em fixed}
with respect to $\pole$ and of order $\O(N).$

\begin{proposition} \label{zerothord} Consider the log-likelihood at
$\pole=\pole^{\prime}+\Delta,$ and assume that $\pole^{\prime} \neq 0,1/2.$ Without
loss of generality, assume that $J_2(\pole^{\prime}+\Delta )=J_2(\pole^{\prime})+1,$
so that $J_1(\pole^{\prime}+\Delta)=J_1(\pole^{\prime})+1$.
Let
\[
\Lambda_N = \ell\left(\pole^{\prime}+\Delta,\delta,J_1+1,J_2+1\right)-
\ell\left(\pole^{\prime}+\Delta,\delta,J_1,J_2\right)
\]
be the magnitude of the discontinuity introduced by perturbing
$\pole^{\prime}.$ Then
\begin{eqnarray*}
{\mathrm{E}}\left[ \Lambda_N \right]=\O(1) \qquad {\mathrm{var}}\left[\Lambda_N
\right]=\O(1)
\end{eqnarray*}
and for every $\epsilon > 0$
\[
P\left(N^{-1}\left| \Lambda_N \right|\ge
\epsilon\right)\rightarrow 0 \quad \textrm{ as }N \rightarrow \infty.
\]

\end{proposition}

\begin{proof}
(Sketch) It is straightforward to show that the discontinuities, $\Lambda_N$,
in the likelihood are random quantities with mean and variance that are
$\O(1)$, so after standardization it follows from the weak law of large numbers
that $\Lambda_N \stackrel{P}{\longrightarrow} 0$ and the result follows.  Full
details are omitted.
\end{proof}

This difference between the log-likelihoods at different values of $\xi$ that
induce a change of the grid is $\O(1).$ We can therefore apply arguments such
as those developed by \cite{coursol}, to justify the usage of a form of the
likelihood which ignores the the jump in the indices, when deriving the
properties of using a form of the likelihood that does experience
discontinuities as the value of $\pole$ alters. We may from the above
calculations note that for large samples it is equivalent to use
$\ell\left(\pole+\Delta,\delta,J_1,J_2\right)$ or
$\ell\left(\pole+\Delta,\delta\right)$ in the analysis of the data; see also
detailed discussion by \cite{Dzham1983}. For our weak convergence result, we
standardize the log-likelihood by a factor of $N^{-1}$ as the log-likelihood
terms are both $\O(N).$ The log-likelihoods are constructed from a data-sample
of size $N,$ and so we can ignore any contribution of order $\O(1).$  Then
$\ell\left(\pole,\delta\right)/N$ will be $\O(1)$, and we shall discuss limits
of properties expressed in terms of this standardized quantity. Finally,
informally, whilst any individual term is contributing $\O(1)$ to the
likelihood, on differentiating the log-likelihood, this is no longer true - the
individual contributions to the score in $\pole$ will be $\O(N)$ near the pole,
and $\O(1)$ away from the pole. This effect renders the discontinuities even of
lesser importance. Note that we can establish a large sample approximation to
the distribution of the standardized log likelihood. We approximate the sum by
an integral and as the correlation between the Fourier coefficients is
sufficiently weak we have
\begin{eqnarray}
\nonumber
\frac{1}{N}\ell(\xi,\delta)&\asymp &
-\int_{-\xi^{\star}}^{1/2-\xi^{\star}} \log\left(f(\lambda) \right)\;d\lambda
-\frac{1}{2N}\chi^2_{N}+\o(1).
\label{forstaeq}
\end{eqnarray}
These results for the entire log likelihood at any fixed value of the
parameters agree with standard likelihood theory. We shall see that the
behaviour of the pole is such that subsequently no result for the estimation of
the pole follows as standard likelihood theory would make us anticipate.
However, with a suitable standardization, the properties of the MLEs and the
likelihood are still tractable.

\subsection{Existence and Consistency Proof}
\label{ConsistencyAppendix}

The existence of the ML estimators is guaranteed as it is easy to show that the
log-likelihood is everywhere bounded on the parameter space.  The proof of
consistency proceeds very similarly to \cite{Giraitis2001}, who assume that the
maximisation over $\xi$ is over a grid of frequencies, where is each grid-point
is spaced $\O(N^{-1})$ apart. This is a sensible choice as the estimation is
most often carried out over the Fourier frequency grid via the DFT. Define
\[
\widetilde{f}_{G}\left(\lambda;\delta,\pole\right)=\sigma^{-2}_{\epsilon}
f_{G}(\lambda;\pole,\delta),
\]
and note that this constrains $\log(\widetilde{f}_{G}(\lambda;\pole,\delta))$
to integrate to zero. \cite{Giraitis2001} show strong convergence of the
estimated location of the singularity to the point on the grid closest to the
true value of the pole, $\polestar,$ using the likelihood defined by equation
(\ref{discfourlik}).

The likelihood approximation defined in Theorem \ref{ApproxLike} cannot be
treated identically to the function of $\pole$ defined in (\ref{discfourlik}),
as the Fourier transform in the former likelihood is calculated at a different
set of frequencies whenever a different value of $\pole$ is picked. However to
compare the magnitude of the log-likelihood at $\pole$ and at $\polestar$ we
need to compare likelihood based on {\em different} Fourier grids. This may
seem problematic, but recall that the DDFT is a linear orthogonal transform,
and so both likelihoods may be directly related to the likelihood of the time
domain sample whatever grid is used. It is hence suitable to compare the
magnitude of the likelihood of the DDFT at different grids. To be able to do
this, we introduce some extra notation. Recall the demodulated grid
$\lambda_j(\pole)=\xi+j/N,\;j=J_1\dots ,J_2.$ First, define $j_p=
j_{p,N}(\pole,\polestar) =
\arg\min_{j\in{\mathbb{Z}}}\left|\polestar-\lamj\left(\pole\right)\right|.$
Thus at any value of $N,$ when the true value of the pole is $\polestar,$ but
the likelihood is evaluated at a grid evenly spaced around $\pole:$ $j_p$ is
then the index of the frequency on the grid demodulated by $\pole$ that is
closest to $\polestar$. Thus $\left|j_p-N\pole\right|\le 1/2$, and we define
$j_p$ uniquely by taking the least of possible values is the pole is evenly
spaced between two demodulated Fourier frequencies. Similarly define $\kap =
\lambda_0(\pole)=\pole$ and $\kapj = \lambda_{j_p} (\pole)= \pole+j_p/N$. Thus
$\kap$ is the demodulated Fourier frequency corresponding to $\pole$ whilst
$\kapj$ is the demodulated Fourier frequency closest to $\polestar.$ Note that
using the triangle inequality
\begin{equation}
N\left|\pole-\polestar\right| \le
N\left|\pole-\kapj\right|+N\left|\kapj-\polestar\right| \leq N\left|\pole-\kapj\right|+1/2.
\end{equation}
This allows us to consider the properties of the log-likelihood at the same
grid explicitly, as $P\left(N\left|\pole-\kappa_{j_p,N}(\pole) \right|\ge
K\right)\ge P\left(N\left|\pole-\polestar \right|\ge K+1/2\right)$. If we
establish the result for $N\left|\pole-\kapj\right|,$ we can redefine $K$ to
derive the same result for $N\left|\pole-\polestar\right|.$ In the vein of
\cite{Giraitis2001}, to show consistency, we fix $\epsilon$ and consider
choosing $K$ such that
\begin{equation}
\label{eq1} P\left(N |\widehat{ \delta  }- \deltastar |^2 \ge
K\right)+P\left(N |\polehat-\polestar |\ge (K+1)\right)\le \epsilon
\end{equation}
Let
\[
u_{N}\left(\chival \right)
=N\left| \delta - \deltastar \right|^2
+\left|N\pole-j_{p,N}(\pole,\polestar)\right|.
\]
We may obtain a bound
for (\ref{eq1}), $2P\left(u_N\left(\chival \right)\ge K\right)$, by considering
\begin{equation}
\label{eq1b} P\left(N  | \deltahat-\deltastar |^2 \ge
K\right)+P\left(N |\polehat- \kappa_{j_p,N}(\widehat{\pole},\polestar) |\ge K\right)\le \epsilon.
\end{equation}
Define $\Omega\left(K\right),$ a subset of the parameter space $(\pole,\delta),$
defined for each fixed constant $K,$ by
\[\Omega\left(K\right)=
\left\{\chival:\;\pole\in\left(0,1/2\right),\;\delta\in\left(0,1/2\right),\;
u_N\left(\chival \right)\ge K \right\}.\] Let
$\tildechistar=(\kapj,\deltastar)$. Analogous to Giraitis {\em et al.}, we
bound (\ref{eq1b}) by
\begin{equation}
\label{bd1}
P\left(\inf_{\chival\in\Omega(K)}\left[\frac{1}{N}
\left\{\ell(\chistar)-\ell(\chival)\right\}
\right]\le 0\right)=
P\left(\inf_{\chival\in\Omega(K)}\left[\frac{1}{N}
\left\{\ell(\chistar)-\ell(\chival)\right\}
\right]/u_N(\chival)\le 0\right)
\end{equation}
Note that the constant $\Polebias$ (see equation (\ref{refsing})) does not explicitly depend on
$N$ or $\pole$ (although the bias is computed at a fixed $\pole$). Also denote the Kronecker-delta
by $\delta_{ij}$ as usual. Consider first
\begin{eqnarray*}
\ell(\tildechistar)-\ell(\chival) &\overset{(1)}{=}&U_N+T_N-1+\frac{\widetilde{f}_{G}\left(\kap
;\deltastar,\kapj\right)}{\widetilde{f}_{G}\left(\kap;\delta,\kap\right)}-\frac{\Pgram(
\kap )}{\sigma_{\epsilon}^2 \widetilde{f}_{G}\left(\kap;\delta,\kap\right)}+
\frac{\Pgram(\kapj )}{\sigma_{\epsilon}^2
\widetilde{f}_{G}\left(\kapj;\delta,\kapj\right)}\\
&&-\frac{\Pgram( \kapj ) }{ B_{\polestar,N}(\polestar,\deltastar )N^{2\deltastar
}f^{\dagger }(\kapj,\deltastar,\kapj)} +\frac{\Pgram( \kap
)}{ B_{\pole,N}(\pole,\delta)N^{2\delta }f^{\dagger }(\kap,\delta,\kap)}\\
\\
&=&U_N+T_N-1+\delta_{j_p,0} -\left\{ \frac{\Pgram( \kap)}{\sigma_{\epsilon}^2
\widetilde{f}_{G}\left(\kap;\delta,\kap\right)}- \frac{\Pgram( \kap)}{
B_{\pole,N}(\pole,\delta)N^{2\delta }f^{\dagger }
(\kap,\delta,\kap)}\right\}\\
\\
&&-\left\{ \frac{\Pgram( \kapj)}{ B_{\polestar,N}(\polestar,\deltastar
)N^{2\deltastar }f^{\dagger }(\kapj,\deltastar ,\kapj)}-
\frac{\Pgram( \kapj)}{\sigma_{\epsilon}^2 \widetilde{f}_{G}\left(\kapj
;\delta,\kap\right)}\right\}\\
\\
&=&U_N+T_N+\delta_{j_p,0}-1-\Pgram( \kap)W_1-\Pgram(
\kapj)(W_2-W_3),
\end{eqnarray*}
where
\[
W_1 = -\frac{1}{ \Polebias N^{2\delta }f^{\dagger
}(\kap,\delta,\kap)} \quad W_2 = \frac{1}{ \Polestarbias
N^{2\deltastar }f^{\dagger }(\kapj,\deltastar,\kapj)} \quad
W_3 = - \frac{1}{\sigma_{\epsilon}^2 \widetilde{f}_{G}\left(\kapj ;\delta,\kap\right)}
\]
where in (1) we have defined $U_N$ and $T_N$ as in \cite{Giraitis2001}. We can
bound the probability in (\ref{bd1}), in a similar fashion:
\begin{eqnarray*}
&&P\left( \sup_{\chival\in\Omega(K)}\left|u_N^{-1}U_N\right|+
\sup_{\chival\in\Omega(K)}\left|u_N^{-1}\left\{1+\Pgram(
\kapj)W_3\right\}
\right|\right.\\
&&\qquad \qquad \qquad\left.+ \frac{1}{K}+
\sup_{\chival\in\Omega(K)}\left|u_N^{-1} \Pgram( \kap)W_1\right|+
\sup_{\chival\in\Omega(K)}\left|u_N^{-1} \Pgram( \kapj)W_2\right|\ge
\inf_{\chival\in\Omega(K)}\left|u_N^{-1}T_N\right| \right)
\end{eqnarray*}
Most terms are the same as in \cite{Giraitis2001}, and
bound in an identical fashion, apart from $\left|u_N^{-1}\Pgram(
\kap)W_1\right|$ and $\left|u_N^{-1}\Pgram( \kapj)W_2\right|.$
Clearly
\[{\mathrm{E}}\left\{\sup_{\chival\in\Omega(K)}\left|u_N^{-1}
\Pgram( \kap)W_1\right|\right\}=C_2 K^{-1} \qquad \textrm{and} \qquad
{\mathrm{E}}\left\{\sup_{\chival\in\Omega(K)}\left|u_N^{-1}
\Pgram( \kapj)W_2\right|\right\}=C_3 K^{-1}.
\]
Hence the result follows, see Theorem 3.1 in \cite{Giraitis2001}. The proof
follows Giraitis {\em et al.'s} and thus for a fixed grid with even spacing
from $\xi,$ shows that the maximiser in terms of $\xi$ of $\ell(\chival)$
becomes close to the point on the grid closest to $\polestar,$ which by the
properties of the grid has to be at most $1/(2N)$ from $\polestar.$ This
obviously is not the distance between the maximiser and $\polestar$ but can be
used to show convergence in probability. This strategy lets us avoid dealing
with problems in the singularity of the likelihood, as well as the local
periodic ripples.

\subsection{First Order Properties of Derivatives of the Likelihood}
\label{FirstOrder}

\begin{proposition}
\label{firstord} For an SPP with parameters $\chistar$, the expectation of the
score evaluated at the $\chival = \chistar$ is zero, that is
${\mathrm{E}}\left\{\ldot(\chistar) \right\}=\o(N)$.
\end{proposition}
\begin{proof}
To deal with the statistical properties, we first note the expectation of the
standardized periodogram as given in \cite{OMS2004}, and Lemma 1 in this proof,
so that
\[
{\mathrm{E}}\left\{\Istd_j\right\}=B_{\lambda_j,N}(\pole,\delta)+o(1)=1+
\O \left( \log(j)/j \right)+o(1),
\]
by results derived from \cite{Robinson1995a}. For large $N$, the second order
properties of the score is dominated by the $\Istddot_j$ terms, that are
distributed like quadratic forms of correlated normal random variables. We
start by deriving the expectation of $\Idot_j$ in terms of the trigonometrical
forms defined in equation (\ref{trig}). We find that:
\begin{eqnarray}
\nonumber {\mathrm{E}}\{  \Idot_j \} = 4\pi
{\mathrm{E}}\left\{B_jC_j-A_jD_j \right\} &=& -2\delta
Nf\left(\lamj\right)\int_{-\infty}^{\infty} u^{-1}
\left|\frac{u}{j}\right|^{-2\delta} \frac{\sin^2\left\{\pi
(j-u)\right\}}{
\left\{\pi (j-u)\right\}^2}\;du+\o(N^{2\delta})\\
\nonumber \\
\label{valofone} &=& Nf\left(\lamj\right)
\dot{B}_{\lambda_j,N}(\pole,\delta) +\o(N^{2\delta}),
\end{eqnarray}
this defining $\dot{B}_{\lambda_j,N}(\pole,\delta)$. From this expression it is
obvious that $\dot{B}_{\lambda_j,N}(\pole,\delta) =
-\dot{B}_{\lambda_{-j},N}(\pole,\delta)$. For large $j$ we have that
$\dot{B}_{\lambda_j,N}(\pole,\delta)=\O\{j^{-1}\},$ where to derive this
result, consider the decomposition
\[\int_{-\infty}^{\infty}=
\int_{-\infty}^{-\epsilon}+\int_{-\epsilon}^{\epsilon}+
\int_{\epsilon}^{j-\epsilon}+\int_{j-\epsilon}^{j+\epsilon}
+\int_{j+\epsilon}^{\infty}.
\]
As in \cite{Robinson1995a} we bound the individual contributions of these
integrals. Using identical arguments, we find for $j=0$,
\[
\nonumber {\mathrm{E}}\{ \Idot_0\}=4\pi {\mathrm{E}}\left\{B_0
C_0-A_0 D_0 \right\} = -\left(2\delta\right)N^{1+2\delta} f^{\dagger}
\left(\pole\right)\int_{-\infty}^{\infty} \frac{u^{-1}}{
|u|^{2\delta}} \frac{\sin^2(\pi u)}{
(\pi
u)^2}\;du+\o(N^{2\delta})=\o(N^{2\delta})
\]
as the integral is zero. Thus when $\pole = \pole^{\star}$,
\[
{\mathrm{E}} \left\{ \Istddot_0\right\}=0 \qquad {\mathrm{E}}\left\{
\Istddot_j\right\}=\dot{B}_{\lambda_j,N}(\pole,\delta)+o(1),
\]
where $\dot{B}_{\lambda_j,N}(\pole,\delta)=O\left(j^{-1}\right)$, and
furthermore note that $\dot{B}_{\lambda_j,N}(\pole,\delta) =
-\dot{B}_{-\lambda_j,N}(\pole,\delta)$. Recall the definition of $\eta_j$ from
equation (\ref{Theorem1Etaj}) in Theorem 1; locally $\eta_j=\eta_{-j}$. We also
define for $j=J_1,\dots,J_2$
\[
R_j^{(1)}=\frac{\partial }{\partial \delta}\log(\eta_j)=-2\log\left|\frac{N}{j}\right|-
\left\{\frac{B_{\xi,N;\delta}(\xi,\delta)}
{B_{\xi,N}(\xi,\delta)}\right\}^{\Upsilon{\{j=0\}}}-
\frac{f^{\dagger}_{j,\delta}}{f^{\dagger}_{j}},\quad
S_j^{(1)}=\frac{\partial }{\partial \pole}\log(\eta_j)=-
\frac{f^{\dagger}_{j,\pole}}{f^{\dagger}_{j}}.
\]
We may then write
\begin{eqnarray*} \ell_{\pole}\left(\chival\right) &=&
\sum_{j=J_1}^{J_2} \left[S_j^{(1)}\left\{1
 -\Istd_j\right\} - N\Istddot_j\right]+\o(N),\\
\ell_{\delta}\left(\chival\right) &=&
\sum_{j=J_1}^{J_2} R_j^{(1)}\left\{1
 -\Istd_j\right\} +\o(N)\\
{\mathrm{E}}\left\{\ell_{\pole}\left(\chistar\right\} \right\}&=&
\sum_{j=J_1}^{J_2} \left(S_j^{(1)}\O\left(\log(j)/j\right)
- N\dot{B}_{\lambda_j,N}(\pole,\delta)\right)+\o(N)=\o(N),\\
{\mathrm{E}}\left\{\ell_{\delta}\left(\chistar\right\}\right\}& =&
\sum_{j=J_1}^{J_2} R_j^{(1)}\O\left(\log(j)/j\right) +\o(N)=\o(N).
\end{eqnarray*}
Thus ${\mathrm{E}}\left\{\ell_{\pole}\left( {\chistar}\right\}\right\}$ is
$\o(N)$. This  characterizes the first order properties of the score functions.
\end{proof}

\subsection{Second Order Properties of Derivatives of the Likelihood}
\label{SecondOrder} The following result enables us to determine the properties
of the Fisher information. We shall discover that the observed Fisher
information does not converge to a point mass, and so far from standard theory
ensues.
\begin{proposition}
\label{secondord} For an SPP with parameters $\chival = \chistar$,
the Fisher information evaluated at $\chistar$ is given by
\begin{eqnarray*}
{\mathrm{E}}\left\{\obsfisherN\left(\chistar\right\}\right\}&=&
\begin{pmatrix}
\fisherNe_{\pole,\pole} & \fisherNe_{\pole,\delta}\\
\fisherNe_{\pole,\delta}& \fisherNe_{\delta,\delta}
\end{pmatrix} = \begin{pmatrix}\fisherNlme_{\pole,\pole} N^2+\o(N^2) & \fisherNlme^{(1)}_{\pole,\delta} N +\o(N)\\
\fisherNlme^{(1)}_{\pole,\delta} N +\o(N)  &
\fisherNlme_{\delta,\delta} N+\o(N)
\end{pmatrix}
\end{eqnarray*}
where $\fisherNe_{\pole,\pole}$ is the expected value of the
negative of the second derivative of the log likelihood taken with
respect to $\pole$, and
\[
\fisherNlme_{\pole,\pole} = \lim_{N \rightarrow \infty}
\frac{\fisherNe_{\pole,\pole}}{N^2},
\]
with $\fisherNlme_{\delta,\delta}$ and $\fisherNlme^{(1)}_{\pole,\delta}$
similarly defined.
\end{proposition}

\begin{proof}
Consider the expectation of the second derivative of the
periodogram; we must calculate
\[{\mathrm{E}}\{\ddotPgramj \}=8\pi^2
{\mathrm{E}}\left\{C_j^2+D_j^2-\left(A_j E_j+ B_j F_j\right\}\right\}.
\]
First, consider $U_j=C_j-i D_j,$ and the standardized version
$\standU_j=\standC_j-i\standD_j.$ After some algebra, suitable
standardization, and integrating by parts on $(\pole-1/\sqrt{N},
\pole+1/\sqrt{N})$, with change of variable to $u$ where
$\pole_1=\pole+u/N$, we have
\begin{eqnarray*}
{\mathrm{E}}\left\{\standU_j \standUconj_k\right\}
&=&\frac{1}{4\pi^2}\left(-1\right)^{k-j}
\int_{-\infty}^{\infty}\left|\frac{j k}{u^2}\right|^{\delta}
\psi(j,k,u) du+o(1)
\end{eqnarray*}
where
{\small
\begin{eqnarray*}
&&\psi(j,k,u)  =  \pi^2 \frac{\sin\left\{\pi(u-j)\right\}}
{\left\{\pi(u-j)\right\}} \frac{\sin\left\{\pi(u-k)\right\}}
{\left\{\pi(u-k)\right\}}-2i\pi \frac{\sin\left\{\pi(u-k)\right\}}
{\pi(u-k)}\left[ \frac{\pi\cos\left\{\pi(u-j)\right\}}
{\left\{\pi(u-j)\right\}}\right. \\
&&\left.-\frac{\sin\left\{\pi(u-j)\right\}}
{\pi^2\{\pi(u-j)\}^2}\right]+ \left[ \frac{\pi\cos\left\{\pi(u-j)\right\}}
{\left\{\pi(u-j)\right\}}-\frac{\sin\left\{\pi(u-j\right\}}
{\pi^2\{\pi(u-j)\}^2}\right] \left[
\frac{\pi\cos\left\{\pi(u-k)\right\}}
{\left\{\pi(u-k)\right\}}-\frac{\sin\left\{\pi(u-k\right\}}
{\pi^2\{\pi(u-k)\}^2}\right].
\end{eqnarray*}}
The calculations are very much in the spirit of \cite{OMS2004}. After some
algebra, we have
\begin{eqnarray*}
{\mathrm{E}}\left\{\standU_j \standUconj_k\right\}
 = \frac{1}{8\pi^2}K_{jk}+o(1)
 = \frac{1}{8\pi^2}\left\{2\pi^2 V_{\lambda_j,\lambda_k,N}(\xi,\delta)
 +W_{\lambda_j,\lambda_k,N}(\xi,\delta) \right\}+o(1),
\end{eqnarray*}
where $V_{\lambda_j,\lambda_k,N}(\xi,\delta)$ is defined in Section
\ref{demodsec} and we define
\begin{equation}
\label{defofW}
W_{\lambda_j,\lambda_k,N}(\xi,\delta)=(-1)^{k-j} 2  \int_{-\infty}^{\infty}
\left|\frac{u^2}{jk}\right|^{-\delta}
\tilde{C}_{\lambda_j,\lambda_k}(u)\;du,
\end{equation}
where
\[
\tilde{C}_{\lambda_j,\lambda_k}(u)=\left[\frac{\partial}{\partial u}
\frac{\sin\{\pi(u-j)\}}{\pi(u-j)}\right]\left[\frac{\partial}{\partial u}
\frac{\sin\{\pi(u-k)\}}{\pi(u-k)}\right].
\]
Similarly, after standardization, and some algebra, we obtain
${\mathrm{E}}\left[\standU_j \standU_k\right]=\o(1)$. Hence
\begin{eqnarray*}
{\mathrm{E}}\left\{U_j U_k*\right\}&=&{\mathrm{E}}\left\{C_j C_k\right\}+{\mathrm{E}}\left\{D_j
D_k\right\}+ i\left\{{\mathrm{E}}\left\{C_j D_k\right\}-{\mathrm{E}}\left\{C_k D_j\right\}
\right\}   = N^2 \sqrt{f\left(\lamj\right) f\left(\lamk\right)}
\frac{1}{8\pi^2}
K_{j,k},\\
{\mathrm{E}}\left\{U_j U_k\right\} &=&{\mathrm{E}}\left\{C_j C_k\right\}-{\mathrm{E}}\left\{D_j
D_k\right\}+ i\left\{{\mathrm{E}}\left\{C_j D_k\right\}+{\mathrm{E}}\left\{C_k D_j\right\}
\right\} = \o\left\{N^2 \sqrt{f\left(\lamj\right)
f\left(\lamk\right)}\right\}.
\end{eqnarray*}
Thus, for large $N$,
\begin{eqnarray*}
{\mathrm{E}}\left\{C_j C_k\right\} = {\mathrm{E}}\left\{D_j D_k\right\} &=&
\frac{1}{16\pi^2}N^2 \sqrt{f\left(\lamj\right)
f\left(\lamk\right)} \left(\Re\left(K_{j,k}\right)+o(1)\right)\\
&=&\frac{1}{16\pi^2}N^2 \sqrt{f\left(\lamj\right)
f\left(\lamk\right)} \left\{2\pi^2 V_{\lambda_j,\lambda_k,N}(\pole,\delta)+
W_{\lambda_j,\lambda_k,N}(\pole,\delta)+o(1)\right\}\\
{\mathrm{E}}\left\{C_j D_k\right\} = -{\mathrm{E}}\left\{C_k D_j\right\}
&=&\frac{1}{16\pi^2}N^2 \sqrt{f\left(\lamj\right)
f\left(\lamk\right)} \left\{\Im\left(K_{j,k}\right)+o(1)\right\}
\end{eqnarray*}
and ${\mathrm{E}}\left\{C_j D_j\right\}=\o(N^2)$. Using similar calculations,
we have that
\begin{eqnarray*}
{\mathrm{E}}\left\{-A_j E_j-B_j F_j\right\} =
\frac{1}{8\pi^2}\left\{-2\pi^2B_{\lambda_j,N}^{(N)}(\pole,\delta)+
\ddot{B}_{\lambda_j,N}^{(N)}(\pole,\delta)-\dot{C}_{\lambda_j,N}^{(N)}(\pole,\delta)\right\},
\end{eqnarray*}
where
\begin{eqnarray*}
B_{\lambda_j,N}^{(N)}(\pole,\delta) & = & \frac{(N-1)^2 }{N} \int_{-1/2}^{1/2} f(u)
\frac{\sin^2\left\{N\pi (\pole+j/N-u)\right\}}
{\sin^2\left\{\pi (\pole+j/N-u)\right\}}\;du\\
\ddot{B}_{\lambda_j,N}^{(N)}(\pole,\delta) & = & -\frac{2}{N} \int_{-1/2}^{1/2} \frac{\partial
f(u)}{\partial u} \frac{\sin\left\{N\pi
(\pole+j/N-u)\right\}}{ \sin\left\{\pi
(\pole+j/N-u)\right\}}\frac{\partial}{
\partial u}
\frac{\sin\left\{N\pi (\pole+j/N-u)\right\}}{
\sin\left\{\pi (\pole+j/N-u)\right\}}\;du\\
\dot{C}_{\lambda_j,N}^{(N)}(\pole,\delta)& = & \frac{2}{N} \int_{-1/2}^{1/2} f(u)
\left[\frac{\partial}{
\partial u}
\frac{\sin\left\{N\pi (\pole+j/N-u)\right\}}{ \sin\left\{\pi
(\pole+j/N-u)\right\}}\right]^2\;du.\\
\end{eqnarray*}

\begin{case}[$j\neq 0$] For large $N$ that looking at the
components of this expectation, and standardizing via
$f\left(\lamj\right)N^2$ that
\[
\frac{B_{\lambda_j,N}^{(N)}(\pole,\delta)}{f\left(\lamj\right)N^2} =
\frac{1}{f\left(\lamj\right)N^2} \frac{(N-1)^2}{ N
}\int_{-1/2}^{1/2} f(u) \frac{\sin^2\left\{N\pi
(\pole+j/N-u)\right\}}{
\sin^2\left\{\pi (\pole+j/N-u)\right\}}\;du= B_{\lambda_j,N}(\pole,\delta)
+\o(1)
\]
This follows directly from OMS. Note that (see \cite{Robinson1995a}):
\begin{equation}
B_{\lambda_j,N}(j,\delta)= 1+\O(\log(j)/j)+o(1).
\end{equation}
Consider the change of variables
from $\pole$ to $u$, where $\pole=\pole+u/N$. Then, after some
algebra
\begin{eqnarray*}
\label{part2}
\frac{1}{8\pi^2}\frac{\ddot{B}_{\lambda_j,N}^{(N)}(\pole,\delta)}{f\left(\lamj\right)N^2}&=&
\frac{1}{f\left(\lamj\right)N^2} \frac{1}{2^3 N
\pi^2}\int_{-1/2}^{1/2} \frac{\partial^2 f(u)}{\partial u^2}
\frac{\sin^2\left\{N\pi (\pole+j/N-u)\right\}}{ \sin^2\left\{\pi
(\pole+j/N-u)\right\}} \;du+o(1)
\end{eqnarray*}
and recalling that
\begin{equation*} \label{hu}
\frac{\partial^2f(u)}{\partial
u^2}=\left|\pole-u\right|^{-2\delta}
\left\{\frac{\partial^2 f^{\dagger}}{\partial u^2}-4\frac{\partial f^{\dagger}}{\partial
u}
\delta\left(\pole-u\right)^{-1}+2f^{\dagger}
\delta(2\delta+1)\left(\pole-u\right)^{-2}\right\}
\end{equation*}
we have
\begin{eqnarray}
\nonumber
{\mathrm{E}}\left\{\frac{1}{8\pi^2}\frac{\ddot{B}_{\lambda_j,N}^{(N)}(\pole,\delta)}
{f\left(\lamj\right)N^2}\right\}
&=& \frac{2\delta(2\delta+1)}{2^3 \pi^2}\int_{-\infty}^{\infty}
\frac{1}{u^2} \left|\frac{u}{j}\right|^{-2\delta} \frac{\sin^2\left\{\pi
(j-u)\right\}}{ \left\{\pi (j-u)\right\}^2} \;du+\o(1) =
\frac{\ddot{B}_{\lambda_j,N}(\pole,\delta)}{2^3\pi^2}+\o(1). \label{eqqyyoon}
\end{eqnarray}
Note that the latter integral converges. Note that for $j$ large
\begin{eqnarray}
\nonumber
\ddot{B}_{\lambda_j,N}(\pole,\delta)&=&2\delta(2\delta+1)\frac{1}{j^2}
\int_{-\infty}^{\infty}
\frac{j^2}{(s+j)^2}\left|\frac{s+j}{j}\right|^{-2\delta} \frac{\sin^2(\pi s)}{ (\pi s)^2} \;ds = \O\left(j^{-2}\right),
\end{eqnarray}
which decays \citep[\S 3.821(9)]{Gradshteyn} with increasing $j.$ The
derivation of this result resembles that of
$\dot{B}_{\lambda_j,N}(\pole,\delta)$ but the integration over $s=0$ needs
direct appeal to {\em mutatis mutandis} of the calculations in
\cite{Robinson1995a}, after the term $j^{-2}$ has been taken outside the
integration. Similarly
\begin{eqnarray*}
\label{part3}
\frac{1}{8\pi^2}\frac{\dot{C}_{\lambda_j,N}^{(N)}(\pole,\delta)}{f\left(\lamj\right)N^2} &=&
\int_{-\infty}^{\infty}
\left|\frac{u}{j}\right|^{-2\delta}\frac{\left[\sin\left\{\pi (j-u)\right\}
-\cos\left\{\pi (j-u)\right\}\pi (j-u)\right]^2} {4\left\{\pi
(j-u)\right\}^4} \;du+\o(1) \\
&=& \frac{1}{8\pi^2}\dot{C}_{\lambda_j,N}(\pole,\delta)+\o(1)
\end{eqnarray*}
We note that $\dot{C}_{\lambda_j,N}(\pole,\delta)\equiv
W_{\lambda_j,\lambda_j,N}(\pole,\delta).$ Note that for $j$ large, {\em mutatis
mutandis} results from \cite{Robinson1995a} bounding the Dirichlet kernel,  for
the expectation of the periodogram (up to terms $o(1)$):
\begin{equation}
\dot{C}_{\lambda_j,N}(\pole,\delta)=2\pi \int_{-\infty}^{\infty}
\left|\frac{s+j}{j}\right|^{-2\delta}
\frac{\left\{\sin\left(s\right)
-\cos\left(s\right)s\right\}^2} {s^4} \;ds = \frac{2 \pi^2}{3}+\O\left\{\frac{\log(j)}{j}\right\}
=\frac{2\pi^2}{3}+\O\left\{\frac{\log(j)}{j}\right\},
\label{largesampC}
\end{equation}
which tends to a constant for increasing $j$. We then have that
\begin{eqnarray*}
\frac{1}{f\left(\lamj\right)N^2}{\mathrm{E}}\left\{ \ddotPgramj \right\}
&=& \frac{8\pi^2}{f\left(\lamj\right)N^2}
{\mathrm{E}}\left\{D_j^2-A_j E_j-B_j F_j+C_j^2\right\} = \ddot{B}_{\lambda_j,N}(\pole,\delta)+\o(1).
\end{eqnarray*}
This gives us
\begin{equation}
\label{secondthing} {\mathrm{E}}\left\{ \ddotPgramj\right\}=
f\left(\lamj\right)N^2 \ddot{B}_{\lambda_j,N}(\pole,\delta)+\o(N^2).
\end{equation}
\end{case}

\begin{case}[$j=0$] For large $N$ considering the
components of this expectation, and standardizing via
$f^{\dagger}\left(\pole\right)N^{2\delta+2}$ it follows that
\[
\label{part10}
\frac{B_{\pole,N}^{(N)}(\pole,\delta)}{f^{\dagger}\left(\pole\right)N^{2+2\delta}}=
\frac{1}{f^{\dagger}\left(\pole\right)}\frac{(N-1)^2}{ N }\int_{-1/2}^{1/2} f(u)
\frac{\sin^2\left\{N\pi (\pole-u)\right\}}{\sin^2\left\{\pi (\pole-u)\right\}}\;du
= \Polebias+\o(1).
\]
This follows directly from Appendix A, including the definition of $\Polebias$.
Again, using the change of variables $\pole=\pole+u/N$, and a similar series of
calculations,
\begin{eqnarray*}
\label{part20}
\frac{1}{8\pi^2}\frac{\ddot{B}_{\pole,N}^{(N)}(\pole,\delta)}
{f^{\dagger}\left(\pole\right)N^{2+2\delta}}
&=&-2\delta \frac{1}{4 \pi}\int_{-\infty}^{\infty}
\left|u\right|^{-2\delta} \left(u\right)^{-1}
\left\{-\frac{\sin^2\left(\pi u\right)}{ \left(\pi u\right)^3}+
\frac{\sin\left(\pi u\right)\cos\left(\pi u\right)}{ \left(\pi
u\right)^2}
\right\}\;du+\o(1)\\
&=&\frac{1}{8\pi^2}\ddot{B}_{\pole,N}(\pole,\delta) B_{\pole,N}(\pole,\delta)+o(1).
\end{eqnarray*}
The $B_{\pole,N}(\pole,\delta)$ term has been added to simplify subsequent calculations.
The integral converges (to see this Taylor expansion of the
integrand at $u=0$).  Finally,
\begin{eqnarray*}
\label{part30}
\frac{1}{8\pi^2}\frac{\dot{C}_{\pole,N}^{(N)}(\pole,\delta)}
{f^{\dagger}\left(\pole\right)N^{2+2\delta}}
&=& \int_{-\infty}^{\infty}
\left|u\right|^{-2\delta}\frac{\left\{-\sin\left(\pi u\right)
+\cos\left(\pi u\right)\pi u\right\}^2} {4\left(\pi u\right)^4}
\;du+\o(1) = \frac{1}{8\pi^2}\dot{C}_{\pole,N}(\pole,\delta) + \o(1).
\end{eqnarray*}
The latter integral also clearly converges.  Thus
\begin{eqnarray*}
\frac{1}{f^{\dagger}\left(\pole\right)N^{2+2\delta}}{\mathrm{E}}\left\{\ddotPgramzero \right\} =
\frac{8\pi^2}{f^{\dagger}\left(\pole\right)N^{2+2\delta}}
{\mathrm{E}}\left\{D_0^2-A_0 E_0-B_0 F_0+C_0^2\right\} =
\ddot{B}_{\pole,N}(\pole,\delta) B_{\pole,N}(\pole,\delta)+o(1),
\end{eqnarray*}
which yields
\begin{equation} \label{secondthing0} {\mathrm{E}}\left\{
\ddotPgramzero \right\}=\ddot{B}_{\pole,N}(\pole,\delta) B_{\pole,N}(\pole,\delta)
f^{\dagger}\left(\pole\right)N^{2+2\delta}+o(N^{2+2\delta})
.
\end{equation}
\end{case}
\noindent These results enable us to determine the properties of the Fisher
Information.

\subsubsection{Asymptotic Properties of the Observed Information
Matrix}
We define
\begin{equation}
\nonumber
R_j^{(2)}=\frac{\partial^2}{\partial \delta^2}\left\{
\log(\eta_j)\right\}\quad
\tilde{R}_j^{(2)}=\frac{1}{\eta_j}\frac{\partial^2}{\partial \delta^2}\left(
\eta_j\right)
\end{equation}
so that
\begin{eqnarray*}
-\ell_{\delta,\delta } &=& -\frac{\partial^2 l }{\partial \delta^2 } =
\sum_j\left\{R_j^{(2)}-\tilde{R}_j^{(2)}\eta_j \Pgrm_j\right\}.
\end{eqnarray*}
Additionally with
\begin{eqnarray*}
T_j&=&\frac{\partial}{\partial \delta}\frac{f^{\dagger}_{j,\pole}}
{f^{\dagger}_j}\quad
\tilde{T}_0=\frac{1}{\eta_j}\frac{\partial}{\partial \delta} \left\{
\frac{1}{B_{\pole,N}(\pole,\delta) N^{2\delta}} \frac{f^{\dagger}_{0,\pole}}
{f^{\dagger 2}_0}\right\}
\\
\breve{T}_0&=&-\frac{N}{\eta_j}\frac{\partial}{\partial \delta} \left\{\frac{1}{B_{\pole,N}(\pole,\delta)
N^{2\delta}f^{\dagger}_0}\right\}\quad
\tilde{T}_j=\frac{1}{\eta_j}\frac{\partial}
{\partial \delta} \left(
\frac{1}{\left|N/j\right|^{2\delta}} \frac{f^{\dagger }_{j,\pole}}
{f^{\dagger 2}_j}\right)\\
\breve{T}_j&=&\frac{N}{\eta_j}\frac{\partial}
{\partial \delta} \left(
\frac{1}{\left|N/j\right|^{2\delta}} \frac{1}
{f^{\dagger }_j}\right)
\end{eqnarray*}
then we find that
\begin{eqnarray*}
-\ell_{\pole,\delta } & =& -\frac{\partial^2 l }{\partial \pole\partial
\delta } = \sum_j \left(T_j-\tilde{T}_j\eta_j \Pgrm_j+
\breve{T}_j\frac{\eta_j}{N} \Idot_j
\right).
\end{eqnarray*}
Finally with
\begin{eqnarray*}
S_j^{(2)}&=&\frac{\partial^2}{\partial \pole^2}\log\left(\eta_j\right)
\quad
\tilde{S}_j^{(2)}=\frac{1}{\eta_j}\frac{\partial^2}
{\partial \pole^2} \left(\eta_j\right),\quad
\breve{S}_j^{(2)}=-2\frac{1}{\eta_j}\frac{\partial^2}
{\partial \delta\pole} \left(\eta_j
\right),
\end{eqnarray*}
we have that:
\begin{eqnarray*}
-\ell_{\pole,\pole}& =&  -\frac{\partial^2 l }{\partial \pole^2 } =
\sum_j\left(S^{(2)}_j-\eta_j\tilde{S}_j^{(2)}\Pgrm_j
+N \breve{S}^{(2)}\frac{\eta_j}{N}\Idot_j
+N^2 \eta_j \frac{\ddotPgramj}{N^2}
\right).
\end{eqnarray*}
Then it transpires
\begin{eqnarray*}
\fisherNe_{\pole,\pole}
&=&\sum_j\left[S^{(2)}_j-\tilde{S}_j^{(2)}
+\breve{S}^{(2)}{\mathrm{E}}\{\Istddot_j\}
+N^2 {\mathrm{E}}\{ \Istdddot_j\}
\right].\\
&=& \frac{ f^{\dagger 2}_{0,\pole}}{f^{\dagger 2}_0} +\sum_{j\neq
0}\frac{ f_{j,\pole}^{\dagger  2}} {f^{\dagger 2
}_j}+N^{2}\ddot{B}_{\pole,N}(\pole,\delta)+\sum_{j\neq 0}\left\{-2 N
\dot{B}_{\lambda_j,N}(\pole,\delta)\frac{f^{\dagger}_{j,\pole}}{f^{\dagger}_{j}}
+N^{2}\ddot{B}_{\lambda_j,N}(\pole,\delta) \right\}+o(N).
\end{eqnarray*}
We can then note that for $N$ large this sum will be dominated by:
\begin{eqnarray*}
\fisherNe_{\pole,\pole} &=&N^2\sum_{j=J_1}^{J_2}\ddot{B}_{\lambda_j,N}(\pole,\delta)
+\o(N^2)=\fisherNlme_{\pole,\pole}N^2+\o(N^2).
\end{eqnarray*}
Note that $J_1/N=-\pole+\o(1)$ and $J_2/N=\frac{1}{2}-\pole+\o(1).$
\begin{eqnarray*}
\fisherNe_{\pole,\delta }&=&
\sum_j \left[T_j-\tilde{T}_j+
\breve{T}_j {\mathrm{E}}\{\dot{I}^{(f,N)}_j\}
\right]\\
&=&\frac{f^{\dagger}_{0,\pole}}
{f^{\dagger}_0}\left\{\frac{B_{\pole,N;\delta}(\pole,\delta)}{B_{\pole,N}(\pole,\delta)}+2\log(N)
+\frac{f_{0,\delta}^{\dagger }}{f^{\dagger}_0}\right\} +\sum_{j\neq
{0}} \left\{ \frac{f^{\dagger }_{j,\pole} }{f^{\dagger }_j}-N
\dot{B}_{\lambda_j,N}(\pole,\delta)\right\}
\left(2\log\left|\frac{N}{j}\right|+\frac{f_{j,\delta}^{\dagger}}{f^{\dagger}_j}
\right)+\o(N).
\end{eqnarray*}
Note that for large values of $j$, the two peaks of equation (\ref{valofone})
separate, and the peak around zero is actually a symmetric peak and trough, and
we have that $\dot{B}_{\lambda_j,N}(\pole,\delta)=\O(j^{-1}),$ whilst
$\dot{B}_{\lambda_j,N}(\pole,\delta) =-\dot{B}_{\lambda_{-j},N}(\pole,\delta).$
If $f^{\dagger}(\lambda)$ admits the representation
$f^{\dagger}(\lambda)=d_0(\chival)+d_1(\chival)(\lambda-\pole)+
d_2(\chival)(\lambda-\pole)^2+\O((\lambda-\pole)^3)$. After some algebra
\begin{eqnarray*}
\fisherNe_{\pole,\delta}
&=&2\log(N)\frac{d_{0,\pole}}{d_{0}}+\o\left\{\log(N)\right\} \\
&&+\sum_{j=J_1,j\neq {0}}^{J_2}
\left\{\frac{d_{0,\pole}}{d_{0}}+\left(
\frac{d_{1,\pole}}{d_{0}}-\frac{d_{0,\pole}d_{1}}{d_{0}^2}
\right)\left|\frac{j}{N}\right|+\o\left(\frac{1}{N}\right)\right\}\\
&& \qquad \qquad
\left\{2\log\left|\frac{N}{j}\right|+\frac{d_{0,\delta}}{d_{0}}+\left(
\frac{d_{1,\delta}}{d_{0}}-\frac{d_{0,\delta}d_{1}}{d_{0}^2}
\right)\left|\frac{j}{N}\right|+\o\left(\frac{1}{N}\right) \right\}\\
&=&2\log(N)q_0+\o\left\{\log(N)\right\}
-4N q_0\pole \log\left|\pole\right|
+N
\frac{q_0 d_{0,\delta}}{2 d_{0}} = \fisherNlme_{\pole,\delta }^{(1)} N+\o\{\log N\},
\end{eqnarray*}
where $q_0=d_{0,\pole}/d_{0}$ is a suitable constant. Finally tedious, but
trivial calculations, based on the Taylor expansion of the function
$f^{\dagger}(\pole)$ yield
\begin{eqnarray}
\nonumber
\fisherNe_{\delta,\delta} &=& \frac{
B_{\pole,N;\delta}^2(\pole,\delta)}{B_{\pole,N}^2(\pole,\delta)}+2\frac{
B_{\pole,N;\delta}(\pole,\delta)f^{\dagger }_{0,\delta}}{B_{\pole,N}(\pole,\delta)
f^{\dagger}_0}+\frac{f^{\dagger 2}_{0,\delta}}{f^{\dagger
2}_0}+4\log(N) \frac{B_{\pole,N;\delta}(\pole,\delta)f_{0,\delta}^{\dagger } }{B_{\pole,N}(\pole,\delta) f^{0,\dagger }}+4\log^2(N)\\ \nonumber
\\ \label{onny2}
&&+\sum_{j=J_1,j \neq 0}^{J_2} \left(
 \frac{f^{\dagger  2}_{j,\delta}}{f^{\dagger 2}_j}
+4\log\left|\frac{N}{j}\right| \frac{f^{\dagger
}_{j,\delta}}{f^{\dagger }_j}+
4\log^2\left|\frac{N}{j}\right|\right)+o(N)\\ \nonumber
\\ \nonumber
&=&\nonumber
\frac{d_{0,\delta}^2}{2d_{0}^2}N + 2N \left[
-4\left\{\left(\log\left|\pole\right|-1\right)\pole\right\}
\frac{d_{0,\delta}^2}{d_{0}^2}+
4\left\{\pole\left(\log^2\left|\pole\right|
- 2\log\left|\pole\right|+2\right)\right\}\right]+\o(N)
\\ \nonumber
&=& \fisherNlme_{\delta,\delta} N+\o(N).
\label{onny1}
\end{eqnarray}
This proves the large sample properties of the Fisher information matrix.
\end{proof}

%% file: appendixc.tex
\subsection{Asymptotic Distributions}
\subsubsection{Distributions of Standardized Scores}
\label{distscore}
\textbf{Score in $\delta$}: To determine the properties of
the MLEs we need to establish the joint distribution of $\ldot$ and
$\obsfisherstdN,$ defined in equations (\ref{scoreinlambda}) and
(\ref{defofBN}). We commence by discussing the first of these quantities. A
usual central limit theorem will apply for $ l_{\delta }\left( \chistar
\right),$ and we already noted that $E\left( l_{\delta }\left( \chistar
\right)\right)=0.$ Furthermore:
\[
{\mathrm{var}}\left\{l_{\delta }\left( \chistar  \right)\right\}
=\sum_{j=J_1}^{J_2}R_j^{(1)2} B_{\lambda_j,N}^2(\pole,\delta)+\o(N) \equiv \fisherNe_{\delta,\delta}+\o(N),
\]
and note that $\fisherNe_{\delta,\delta} = \fisherNlme_{\delta,\delta}N+\o(N).$
We may make the following note and definition:
\begin{eqnarray}
\nonumber
\ell_{\delta}\left( \chistar  \right)&=& \sqrt{N}
Z_1+\o( \sqrt{N}),\quad Z_1\sim
N\left(0,\fisherNlme_{\delta,\delta}
\right)\\
k_{N,2}\left( \chistar  \right)&=&\left( \fisherNlme_{\delta,\delta}N\right)^{-1/2}
\ell_{\delta}\left( \chistar  \right)\overset{{\cal L}}{\rightarrow} Z_2,\quad Z_2\sim \mathcal{N} \left(0,1\right)
\label{scoreindelta}.
\end{eqnarray}
Also we noted in the previous section that
\[{\mathrm{E}}\left\{-\ell_{\delta,\delta}\right\}=\fisherNe_{\delta,\delta} =\O(N),\quad
{\mathrm{var}}\left\{-\ell_{\delta,\delta}\right\}=\sum_j
\tilde{R}_j^{(2)2}B_{\lambda_j,N}^2(\pole,\delta)+o(N)=\O(N).\]
Thus
\[
W_{N,22}=-\frac{\ell_{\delta,\delta}}{\fisherNlme_{\delta,\delta}N}\overset{P}{\rightarrow}1,
\]
and so we may note that as $k_{N,2}$ and $W_{N,22}$ are asymptotically
uncorrelated and Gaussian, we find that using Slutsky's theorem
\begin{eqnarray}
\sqrt{\fisherNlme_{\delta,\delta}N}\left(\widehat{\delta}-\delta^{\ast}\right)
&=&k_{N,2}\left( \chistar  \right)\left[W_{N,22}\right]^{-1} \overset{\cal{L}}{\longrightarrow} \mathcal{N}(0,1),
\label{deltalim}
\end{eqnarray}
and from this result we can deduce Theorem 5. The value of
$\fisherNe_{\delta,\delta}$ and $\fisherNlme_{\delta,\delta}N$ are given by
equations (\ref{onny2}) and (\ref{onny1}), respectively.

\vspace{0.1 in}

\noindent\textbf{Score in $\pole$}: If the likelihood were sufficiently
regular, then the arguments that we used to derive the distribution of
$\sqrt{\fisherNlme_{\delta,\delta}N}\left(\widehat{\delta}-\delta^{\ast}\right)$
could be replicated for $\pole$ instead of $\delta,$ and the large sample
theory would be relative straightforward. However, this is not the case, and we
find that for the parameter $\pole,$ the situation is more complicated. The
first observation of interest is that we may note that the score is dominated
by the derivative of the demodulated periodogram, i.e. $\Istddot_j.$ In fact,
with an appropriate standardization of the score we determine that
\begin{eqnarray}
\nonumber \frac{1}{N^{3/2}}\ell_{\pole}\left(\chistar\right)
&=&\frac{1}{\sqrt{N}}\left[
\frac{1}{N}\sum_{j=J_1}^{J_2}
\left\{S_j^{(1)}\left(1-\eta_j \Istd_j\right)-\eta_j\Istddot_j\right\}
\right] = -\frac{1}{N^{1/2}}\sum_{j=J_1}^{J_2}
\Istddot_j+\o(1)\\
&&\label{normscore2}
\end{eqnarray}
where the sum random variable converges in distribution.  To be able to
determine the large sample properties of this object, we thus need to derive
the joint distribution of the random variables $\{\Istddot_j\}.$ $\Istddot_j$
is a quadratic form in correlated Gaussian random variables
$\standA_j,\;\standB_j,\;\standC_j$ and $\standD_j,$ that make up the
standardized derivative of the periodogram. Their joint distribution can be
determined from their covariance.
\begin{proposition}
\label{covarianceofder}
\begin{equation}
{\mathrm{cov}}\left\{\Istddot_j,
\Istddot_k\right\}=\frac{5\pi^2}{2}
V_{\lambda_j,\lambda_k,N}^2\left(\pole,\delta\right)+
\frac{1}{2} \dot{V}_{\lambda_j,\lambda_k,N}^2(\pole,\delta)+
\frac{1}{4}V_{\lambda_j,\lambda_k,N}\left(\pole,\delta\right)
W_{\lambda_j,\lambda_k,N}\left(\pole,\delta\right)
,\;j\neq k,
\label{covist}
\end{equation}
up to order $o(1)$, where $V_{\lambda_j,\lambda_k,N}\left(\pole,\delta\right)$ is defined in
section \ref{demodsec}, $W_{\lambda_j,\lambda_k,N}\left(\pole,\delta\right)$ is
defined by eqn (\ref{defofW}) and
$\dot{V}_{\lambda_j,\lambda_k,N}(\pole,\delta)$ is given by
\[\dot{V}_{\lambda_j,\lambda_k,N}(\pole,\delta)=-2\delta\int_{-\infty}^{\infty} s^{-1}\left|
\frac{s^2}{jk}\right|^{-\delta}\frac{\sin\{\pi(j-s)\}\sin\{\pi(k-s)\}}{\pi^2
(j-s)(k-s)}\;ds.\]
Note that if $\log(N)<k<j$ then
\[V_{\lambda_j,\lambda_k,N}(\pole,\delta)=\O\left\{\frac{\log(j)}{k}\right\},\quad
\dot{V}_{\lambda_j,\lambda_k,N}(\pole,\delta)=\O\left\{\frac{\log(j)}{k^2}\right\}.\]
\end{proposition}

\begin{proof}
Define $\bm{V}_j=\left(A_j,B_j,C_j,D_j\right)^\transpose$ and $
\bm{V}_j^{(f,N)}=\left\{\standA_j\;\standB_j\;
\standC_j\;\standD_j\right\}^\transpose$ and note that its components are
correlated normal random variables. Note that
${\mathrm{E}}\left\{\bm{V}_j\right\}=\bm{0}.$ We shall derive the final
calculations needed to complete the entries of the covariance matrix of this
object, namely ${\mathrm{cov}}\{A_j,C_j\},$ ${\mathrm{cov}}\{B_j,D_j\},$
$\textrm{cov}\{B_j,C_j\},$ and $\textrm{cov}\{A_j,D_j\}.$ As above, integrating
in the region $(\polestar \pm N^{-1/2}$ after change of variable to $u$ where
$\pole=\polestar+\frac{u}{N}$, we find that the suitably standardized random
variates have expectation:
\begin{eqnarray*}
{\mathrm{E}}\left\{\standA_j\standC_j\right\} =
{\mathrm{E}}\left\{\standB_j\standD_j\right\} =
\frac{1}{4}\int_{-\infty}^{\infty} \left|\frac{j}{u}\right|^{2\delta}
\frac{\sin^2\left[\pi (j-u)\right]}{ \left[\pi (j-u)\right]^2} du+\o(1)
=B_{\lambda_j,N}(\pole,\delta) /4+\o(1)
\end{eqnarray*}
and where the terms including the derivative of F\'{e}jer's kernel
cancel after a change of variable $u \rightarrow -u$.  Also we can
note from our calculations of the first differential that
\begin{eqnarray*}
{\mathrm{E}}\left\{A_j D_j\right\} &=&-1/2 Nf(\lamj)\dot{B}_{j,D,N}(\pole,\delta)/(4\pi)+\o(N)f(\lamj),
\end{eqnarray*}
as the cross-terms contribute terms of lesser order of magnitude for
large $N,$ and with a change of variable $\pole\rightarrow -\pole$
the terms multiplied by $N-1$ cancel. We also note that
\[
{\mathrm{E}}\left\{B_j C_j\right\}=-{\mathrm{E}}\left\{A_j D_j\right\}=1/2
Nf(\lamj)\dot{B}_{\lambda_j,N}(\pole,\delta)/(4\pi)+\o(N)f(\lamj),
\]
this result characterizing the second order structure of the derivative of the
periodogram. This, in combination with OMS and previously derived results
yields (up to terms $\o(1)f\left(\lamj\right)$): {\small \begin{eqnarray}
{\mathrm{var}}\left(\bm{V}_j\right)&=& f\left(\lamj\right)\begin{pmatrix}
\scriptstyle \frac{1}{2}B_{\lambda_j,N}(\pole,\delta) & 0 &
\scriptstyle\frac{1}{4} N B_{\lambda_j,N}(\pole,\delta) &
\scriptstyle  \frac{N}{8\pi} \dot{B}_{\lambda_j,N}(\pole,\delta)\\
0 & \scriptstyle \frac{1}{2}B_{\lambda_j,N}(\pole,\delta) &
\scriptstyle  \frac{N}{8\pi} \dot{B}_{\lambda_j,N}(\pole,\delta) &
\scriptstyle  \frac{1}{4} N B_{\lambda_j,N}(\pole,\delta)\\
\scriptstyle \frac{1}{4} N B_{\lambda_j,N}(\pole,\delta) &
\scriptstyle  \frac{N}{8\pi}\dot{B}_{\lambda_j,N}(\pole,\delta) &
\scriptstyle \frac{1}{16\pi^2}N^2 \Re\left(K_{j,j}\right)
& 0\\
\scriptstyle - \frac{N}{8\pi} \dot{B}_{j,D,N}(\pole,\delta) &
\scriptstyle \frac{1}{4} N B_{\lambda_j,N}(\pole,\delta) &0 &
\scriptstyle \frac{1}{16\pi^2}N^2 \Re\left(K_{j,j}\right)
\end{pmatrix} ,
\label{kata}
\end{eqnarray}}
which we shall denote $\bm{\Omega}_j.$ Note that $K_{jj}=2\pi^2
B_{\lambda_j,N}(\pole,\delta)+\dot{C}_{\lambda_j,N}(\pole,\delta).$ We are also
interested in the covariance between the terms $\Istddot_k,$ and thus need to
calculate
{\small {{\begin{eqnarray*} &&\frac{{\mathrm{cov}}\left\{
\Istddot_j,\Istddot_k\right\}}{\left(4\pi\right)^2} ={\mathrm{cov}}\left\{
\standB_j \standC_j-\standA_j \standD_j,\standB_k
\standC_k-\standA_k \standD_k\right\}\\
&=&{\mathrm{cov}}\left\{\standB_j
\standC_j,\standB_k
\standC_k \right\}-
{\mathrm{cov}}\left\{\standB_j
\standC_j,\standA_k\standD_k \right\}\nonumber\\
&&- {\mathrm{cov}}\left\{\standA_j\standD_j,\standB_k\standC_k \right\}
+{\mathrm{cov}}\left\{\standA_j\standD_j,\standA_k\standD_k \right\}
\end{eqnarray*}}
Using Isserlis's theorem (see \cite{isserlis}) for zero-mean Gaussian variates
we note that
\begin{eqnarray*}
{\mathrm{E}} \left\{ X_1 Y_1 X_2 Y_2 \right\} = {\mathrm{E}}\left\{X_1 X_2\right\} {\mathrm{E}}\left\{Y_1 Y_2\right\} + {\mathrm{E}}\left\{X_1 Y_2\right\}{\mathrm{E}}\left\{X_2 Y_1\right\}+
{\mathrm{E}}\left\{X_1 X_2\right\} {\mathrm{E}}\left\{Y_1 Y_2\right\}.
\end{eqnarray*}
Hence we find that $\mathrm{cov} \left\{ \Istddot_j,\Istddot_k
\right\}$ is equal to
{\small\begin{eqnarray}
\nonumber
&&(4\pi)^2
\left[{\mathrm{E}}\left\{\standB_j \standB_k\right\}
{\mathrm{E}}\left\{\standC_j \standC_k\right\}+
{\mathrm{E}}\left\{\standB_j \standC_k\right\}
{\mathrm{E}}\left\{\standB_k \standC_j\right\}\right.\\
&& \nonumber
-{\mathrm{E}}\left\{\standB_j \standA_k\right\}
{\mathrm{E}}\left\{\standC_j \standD_k\right\}-
{\mathrm{E}}\left\{\standB_j \standD_k\right\}
{\mathrm{E}}\left\{\standC_j \standA_k\right\}\\
\nonumber
&&-{\mathrm{E}}\left\{\standB_k \standA_j\right\}
{\mathrm{E}}\left\{\standD_j \standC_k\right\}
-
{\mathrm{E}}\left\{\standC_k \standA_j\right\}
{\mathrm{E}}\left\{\standB_k \standD_j\right\}\\
&&\nonumber
\left.
+{\mathrm{E}}\left\{\standA_j \standA_k\right\}
{\mathrm{E}}\left\{\standD_j \standD_k\right\}+
{\mathrm{E}}\left\{\standA_j \standD_k\right\}
{\mathrm{E}}\left\{\standD_j \standA_k\right\}
\right]
\\
\nonumber
&=&(4\pi)^2\left\{\frac{1}{2}V_{\lambda_j,\lambda_k,N}\left(\pole,\delta
\right)\right\}^2+(4\pi)^2
\left\{-\frac{1}{8\pi}
\dot{V}_{\lambda_j,\lambda_k,N}(\pole,\delta)\right\}\left\{-\frac{1}{8\pi}
\dot{V}_{\lambda_j,\lambda_k,N}(\pole,\delta)\right\}-\o(1)\\
\nonumber
&&-(4\pi)^2\left\{\frac{1}{4}V_{\lambda_j,\lambda_k,N}\left(\pole,\delta
\right)\right\}^2-\o(1)-(4\pi)^2\left\{\frac{1}{4}V_{\lambda_j,\lambda_k,N}\left(\pole,\delta
\right)\right\}^2
\\ \nonumber
&&
+\frac{(4\pi)^2}{16\pi^2}\left\{\frac{1}{4}V_{\lambda_j,\lambda_k,N}\left(\pole,\delta
\right)\right\}\left\{2\pi^2V_{\lambda_j,\lambda_k,N}(\pole,\delta)+W_{\lambda_j,\lambda_k,N}(\pole,\delta)
\right\}\\
&& \nonumber
+\left\{\frac{1}{2}
\dot{V}_{\lambda_j,\lambda_k,N}(\pole,\delta)\right\}\left\{\frac{1}{2}
\dot{V}_{\lambda_j,\lambda_k,N}(\pole,\delta)\right\}
\\
&=&\pi^2\frac{5}{2}V_{\lambda_j,\lambda_k,N}^2\left(\pole,\delta \right)+
\frac{1}{2}\dot{V}^2_{\lambda_j,\lambda_k,N}(\pole,\delta)+\frac{1}{4}V_{\lambda_j,\lambda_k,N}(\pole,\delta)
W_{\lambda_j,\lambda_k,N}(\pole,\delta)+\o(1)
\nonumber\\
&=&\O\left\{k^{-2} \log^2(j) \right\}+\O\left\{k^{-4} \log^2(j) \right\}+
\O\left\{k^{-2} \log^2(j) \right\}+\o(1) \label{orderofthing}.
\end{eqnarray}}}}
Note that the bound for $\dot{V}_{\lambda_j,\lambda_k,N}$ follows by arguments, {\em mutatis
mutandis}, \cite{Robinson1995a}.
\end{proof}

\subsubsection{Distribution of the Derivative of the Standardized Periodogram}
\label{delttty}
We now derive the distribution of $\Istddot_j$ to be able to determine the
distribution of $\sum_j \Istddot_j:$
\begin{proposition}
\label{covariancematrix}
\begin{equation}
\label{nyadistribution} \Istddot_j\sim \sum_{k=1}^{4}
\gamma_{k}^{(j)} R_{i,j}^2+o(1),
\end{equation}
where $\gamma_{k}^{(j)}$ are the roots of equation
\begin{eqnarray*}
&&\gamma^4-\dot{B}_{j,D,N}(\pole,\delta) \gamma^3+\left\{\frac{3}{8}
\dot{B}_{j,D,N}^2(\pole,\delta)-
\frac{1}{4}B_{\lambda_j,N}(\pole,\delta) \dot{C}_{\lambda_j,N}(\pole,\delta)\right\}\gamma^2
+ \frac{1}{4}\dot{B}_{\lambda_j,N}(\pole,\delta)\\
&&\times \left\{\frac{B_{\lambda_j,N}(\pole,\delta)
\dot{C}_{\lambda_j,N}(\pole,\delta)}{2}-2^{-2}\dot{B}_{\lambda_j,N}^2(\pole,\delta) \right\}\gamma^2\\
&&+ \frac{1}{2^6}B_{\lambda_j,N}(\pole,\delta)^2\dot{C}_{\lambda_j,N}^2(\pole,\delta)-\frac{1}{2^7}
\dot{B}_{\lambda_j,N}^2(\pole,\delta)B_{\lambda_j,N}(\pole,\delta) \dot{C}_{\lambda_j,N}(\pole,\delta)+
\frac{1}{2^6}\dot{B}_{\lambda_j,N}^4(\pole,\delta)
= 0,
\end{eqnarray*}
and $R_{i,j}$ are independent unit Gaussian variables across $i$ for each fixed
$j.$ This in turn implies
\begin{eqnarray}
\nonumber {\mathrm{E}}\{ \Istddot_j\} = \sum_{k=1}^{4}
\gamma_{k}^{(j)}+o(1) \qquad {\mathrm{{\mathrm{var}}}}\{\Istddot_j \} =  2
\sum_{k=1}^{4} \gamma_{k}^{2(j)}+o(1). \label{andraord}
\end{eqnarray}
\end{proposition}

\begin{proof}
Firstly note that for a fourth order polynomial with roots
$\left\{\gamma_k^{(j)}\right\}$ we find that
\begin{eqnarray*}
\prod_{k=1}^{4} \left\{\gamma-\gamma_k^{(j)}\right\} &=&
\gamma^4-\gamma^3 \sum_{k=1}^{4} \gamma_{k}^{(j)}
+\gamma^2\sum_{l<k} \gamma_k^{(j)}\gamma_l^{(j)}
-\gamma \sum_{u<l<k} \gamma_k^{(j)}\gamma_l^{(j)}\gamma_u^{(j)}
+\gamma_1^{(j)}\gamma_2^{(j)}\gamma_3^{(j)}\gamma_4^{(j)}\\
&=& \gamma^4+b_j \gamma^3+c_j \gamma^2+d_j\gamma+e_j.
\end{eqnarray*}
Also note that
\begin{equation}
\nonumber \sum_{k=1}^{4} \gamma_{k}^{2(j)}= \left\{\sum_{k=1}^{4}
\gamma_{k}^{(j)} \right\}^2-2\left\{\sum_k \sum_{l<k}
\gamma_k^{(j)}\gamma_l^{(j)}\right\}.
\end{equation}
Thus we find that
\begin{eqnarray}
E\{ \Istddot_j\}&=&\sum_{k=1}^{4}
\gamma_{k}^{(j)} = -b_j = \dot{B}_{\lambda_j,N}(\pole,\delta)+o(1) \label{espected}\\
\nonumber
{\mathrm{var}}\{ \Istddot_j\}&=&2\sum_{k=1}^{4}
\gamma_{k}^{2(j)} = 2b_j^2-4c_j = 2\dot{B}_{j,D,N}^2(\pole,\delta)-4 \left\{\frac{3}{8}\dot{B}_{\lambda_j,N}^2(\pole,\delta)-\frac{1}{4}B_{\lambda_j,N}(\pole,\delta)
\right.\\
&&\left.\dot{C}_{\lambda_j,N}(\pole,\delta)\right\}+\o(1) =\frac{1}{2}
\dot{B}_{\lambda_j,N}^2(\pole,\delta)+B_{\lambda_j,N}(\pole,\delta)
\dot{C}_{\lambda_j,N}(\pole,\delta)+\o(1), \label{espected2}
\end{eqnarray}
this giving the full first and second order structure of the standardized
derivative, from the quadratic form. Equation (\ref{espected}) matches the
previously developed results for the expectation of $\Istddot_j$.
(\ref{espected2}) gives a compact expression for the variance. Of some interest
is now the difference in magnitude between this quantity and the $j$th
contribution of $\fisherNe_{\pole,\pole}/N^2,$ but this is not sufficient to
establish the large sample properties of the distribution, as
$-\ell_{\pole,\pole},$ does not converge in probability to a constant if
suitably standardized. Note that $B_{\lambda_j,N}(\pole,\delta)$ nearly takes
the value unity for most $j,$ and for $j$ small due to the integrand of
$\dot{B}_{\lambda_j,N}(\pole,\delta)$ being odd near the origin, clearly
$\left|\dot{C}_{\lambda_j,N}(\pole,\delta)\right|>>
0.5\dot{B}_{\lambda_j,N}^2(\pole,\delta).$ We therefore to derive a compact
expression for the properties of $\Istddot$  to compare the magnitude of
$\dot{C}_{\lambda_j,N}(\pole,\delta)$ and
$\ddot{B}_{\lambda_j,N}(\pole,\delta)$ to justify this argument.

To derive eqn (\ref{nyadistribution}) we use results given in
\cite[p.~149--188]{JohnsonandKotz} on quadratic forms. Note that with
\begin{eqnarray*}
\bm{T}=\begin{pmatrix}
0 & 0 & 0 & -1/2\\
0 & 0 & 1/2 & 0\\
0 & 1/2 & 0 & 0\\
-1/2 & 0 & 0 & 0
\end{pmatrix}
\end{eqnarray*}
we have
\begin{equation}
\label{derri}
\Istddot_j=4\pi\bm{V}_j^{(f,N)T}
\bm{T} \bm{V}_j^{(f,N)}.
\end{equation}
Firstly define
\begin{equation}
\label{sq}
\mathrm{var}\{\bm{V}_j^{(f,N)}\}=\bm{\Omega}_j^{(f,N)}+o(1)=
{\mathbf{\mathcal{L}}}_j{\mathbf{\mathcal{L}}}_j^\transpose+o(1),
\end{equation}
where $\bm{\Omega}_j^{(f,N)}$ is the normalized version of eqn (\ref{kata})
and
where ${\bm{\mathcal{L}}}_j$ is the lower triangular matrix given by, where,
for notational purposes we take $\B_{\lambda_j,N}=\B_{\lambda_j,N}(\pole,\delta)$
and $\dot{C}_{\lambda_j,N}=\dot{C}_{\lambda_j,N}(\pole,\delta)$ :
\begin{equation}
\nonumber {\bm{\mathcal{L}}}_j=\begin{pmatrix}
\sqrt{\frac{B_{\lambda_j,N}}{2}} & 0 & 0 & 0\\
0 & \sqrt{\frac{B_{\lambda_j,N}}{2}} & 0&0\\
\sqrt{\frac{B_{\lambda_j,N}}{2^3}} &  \frac{
\dot{B}_{\lambda_j,N}}{4\pi \sqrt{2B_{\lambda_j,N}}}&
\frac{1}{4\pi}\sqrt{\frac{\dot{C}_{\lambda_j,N}B_{\lambda_j,N}
-2\dot{B}_{\lambda_j,N}^2
}{B_{\lambda_j,N}}}&0\\
-\frac{\dot{B}_{\lambda_j,N}}{4\pi\sqrt{2B_{\lambda_j,N}}}&
\sqrt{\frac{B_{\lambda_j,N}}{2^3}}
 & 0 &
\frac{1}{4\pi}\sqrt{\frac{\dot{C}_{\lambda_j,N}B_{\lambda_j,N}
-2\dot{B}_{\lambda_j,N}^2
}{B_{\lambda_j,N}}}
\end{pmatrix}.
\end{equation}
Note that
\begin{equation}
\label{dista} \bm{V}_j^{(f,N)}\sim
{\mathcal{N}}\left(\bm{0},\bm{\Omega}_j^{(f,N)} \right)+o(1),
\end{equation}
and thus $\bm{Z}_j={\bm{\mathcal{L}}}_j^{-1}\bm{V}_j\sim
{\mathcal{N}}\left(\bm{0},\bm{I}_4 \right).$ The quadratic form is
then given by (ignoring terms $o(1)$):
\begin{equation}
\nonumber (4\pi)^{-1}\Istddot =
\bm{V}_j^{(f,N)T} \bm{T} \bm{V}_j^{(f,N)} = \bm{Z}_j^\transpose
\bm{\mathcal{L}}_j^\transpose \bm{T} \bm{\mathcal{L}}_j\bm{Z}_j= \bm{Z}_j^\transpose
\bm{\mathcal{M}}_j\bm{Z}_j,
\end{equation}
and thus the distribution of this object depends wholly on the
eigenvalue of $\bm{\mathcal{M}}_j .$ Note that
\begin{eqnarray*}
{\mathcal{M}}_j&=&\bm{\mathcal{L}}_j^\transpose\bm{T} \bm{\mathcal{L}}_j=
\begin{pmatrix}
\frac{\dot{B}_{\lambda_j,N}(\pole,\delta)}{8\pi}& 0&0 & \Gamma_j \\
0 & \frac{\dot{B}_{\lambda_j,N}(\pole,\delta)}{8\pi} & \Gamma_j & 0\\
0 & \Gamma_j& 0 &0\\
\Gamma_j & 0 & 0 & 0
\end{pmatrix}
\end{eqnarray*}
where
\[
\Gamma_j =
-\frac{1}{8\pi}\sqrt{\frac{B_{\lambda_j,N}(\pole,\delta)}{2}\dot{C}_{\lambda_j,N}
(\pole,\delta)-\dot{B}^2_{\lambda_j,N}(\pole,\delta)}.
\]
We are interested in $4\pi  {\mathcal{M}}_j$ which has eigenvalues
$\gamma_k^{(j)}$ given as the solution of
\begin{eqnarray*}
&&\gamma^4-\dot{B}_{\lambda_j,N}(\pole,\delta) \gamma^3+\left\{\frac{3}{8}
\dot{B}^2_{\lambda_j,N}(\pole,\delta)-
\frac{1}{4}B_{\lambda_j,N}(\pole,\delta)\dot{C}_{\lambda_j,N}(\pole,\delta)\right\}\gamma^2
+\frac{\dot{B}_{\lambda_j,N}(\pole,\delta)}{4}\\
&&\times \left\{\frac{B_{\lambda_j,N}(\pole,\delta)
\dot{C}_{\lambda_j,N}(\pole,\delta)}{2}-\dot{B}^2_{\lambda_j,N}(\pole,\delta)/4 \right\}\gamma^2\\
&&+ \frac{1}{2^6}\frac{B_{\lambda_j,N}(\pole,\delta)^2\dot{C}^2_{\lambda_j,N}(\pole,\delta)}{4}-
\frac{1}{2^7}
\dot{B}^2_{\lambda_j,N}(\pole,\delta) B_{\lambda_j,N}(\pole,\delta) \dot{C}_{\lambda_j,N}(\pole,\delta)+\frac{1}{2^6}
\dot{B}^4_{\lambda_j,N}(\pole,\delta)
= 0.
\end{eqnarray*}
We then note from \citet[p.~151]{JohnsonandKotz} that if we define new
variables $\bm{R}_j$ in terms of the orthogonal matrix of eigenvectors of
$\bm{\mathcal{M}}_j$ and $\bm{Z}_j,$ they will be $\bm{R}_j\sim
{\mathcal{N}}\left(\bm{0},\bm{I}_4 \right),$ and
\begin{equation}
\label{distrub} \Istddot_j=4\pi \bm{Z}_j^\transpose\bm{\mathcal{M}}_j\bm{Z}_j\sim \sum_{k=1}^4
\gamma_k^{(j)} R_{j,k}^2+o(1),
\end{equation}
thus completing the proof of the proposition, and establishing the marginal
distribution
of $\Istddot_j.$
\end{proof}

\begin{proposition}
\label{dist8} The standardized score function satisfies constraint:
\begin{eqnarray}
\nonumber
k_{N,1}(\chistar)&=&\frac{1}{N^{3/2}}\ell_{\pole}=K_N+\o(1)\\
K_N&\sim&\mathcal{N}\left(0,\breve{\sigma}^2_{N}\right)
\label{modNscore}
\\
K_N&\overset{{\cal L}}{\Longrightarrow}&Z_4\sim N\left(0,\frac{\pi^2}{3}\right),
\label{largeNscore}
\end{eqnarray}
where
\begin{eqnarray*}
\nonumber
\breve{\sigma}^2_N&=&\frac{1}{N}
\sum_{j=J_1}^{J_2} \left\{\frac{1}{2}\dot{B}^2_{\lambda_j,N}(\pole,\delta)+B_{\lambda_j,N}(\pole,\delta)
\dot{C}_{\lambda_j,N}(\pole,\delta)\right\}\\
\nonumber
&&+\frac{1}{N}\sum_{j\neq k}\left\{
\frac{5\pi^2}{2}
V_{\lambda_j,\lambda_k,N}^2\left(\pole,\delta\right)+
\frac{1}{2} \dot{V}_{\lambda_j,\lambda_k,N}^2(\pole,\delta)+
\frac{1}{4}V_{\lambda_j,\lambda_k,N}\left(\pole,\delta\right)
W_{\lambda_j,\lambda_k,N}\left(\pole,\delta\right)\right\}\\
&=& \frac{1}{N}\sum_{j=J_1}^{J_2} \nonumber
\dot{C}_{\lambda_j,N}(\pole,\delta)+\o(1) \longrightarrow \frac{\pi^2}{3}.
\end{eqnarray*}
\end{proposition}
\begin{proof}
PART I (\textit{Determining the first and second order properties of $k_{N,1}(\chistar)$}): We note that
\[k_{N,1}(\chistar)=
\frac{1}{N^{3/2}}\ell_{\pole}=- \frac{1}{\sqrt{N}}\sum_j
\Istddot(\lamj)+\o(1)=Y_{1,N}(\chistar)+\o(1),\]
from equation (\ref{normscore2}), the equation defining the random variable
$Y_{1,N}(\chistar).$ We then
note from equations (\ref{valofone}), (\ref{espected2}) and (\ref{covist}) that:
\begin{eqnarray*}
\nonumber
E\left\{Y_{1,N}(\chistar)\right\}&=&\o(1)\\
\nonumber{\mathrm{var}}\left\{\frac{1}{\sqrt{N}}\Istddot_j\right\}&=&
\frac{1}{N} \left\{\frac{1}{2}\dot{B}^2_{\lambda_j,N}(\pole,\delta)
+B_{\lambda_j,N}(\pole,\delta) \dot{C}_{\lambda_j,N}(\pole,\delta)\right\}
+\o\{N^{-1}\}\\
\nonumber{\mathrm{cov}}\left\{\frac{1}{\sqrt{N}}\Istddot_j,
\frac{1}{\sqrt{N}}\Istddot_k\right\}&=&
\nonumber
\frac{1}{N}
\left\{
\frac{5\pi^2}{2}
V_{\lambda_j,\lambda_k,N}^2\left(\pole,\delta\right)+
\frac{1}{2} \dot{V}_{\lambda_j,\lambda_k,N}^2(\pole,\delta)+\right.\\
&&+\left.
\frac{1}{4}V_{\lambda_j,\lambda_k,N}\left(\pole,\delta\right)
W_{\lambda_j,\lambda_k,N}\left(\pole,\delta\right)\right\}
+\o\{N^{-1}\} \; \;j\neq
k.\nonumber\end{eqnarray*}
Thus it follows that:
\begin{eqnarray*}
\nonumber
{\mathrm{var}}\left\{Y_{1,N}(\chistar)\right\}&=& \frac{1}{N}\sum_j
\left\{\frac{1}{2}\dot{B}^2_{\lambda_j,N}(\pole,\delta)+B_{\lambda_j,N}(\pole,\delta) \dot{C}_{\lambda_j,N}(\pole,\delta)\right\}\nonumber \\
&&
+\frac{1}{2N}\sum_{k\neq j} \left[\dot{V}_{\lambda_j,\lambda_k,N}^2(\pole,\delta)+
V_{\lambda_j,\lambda_k,N}(\pole,\delta)\left\{5\pi^2
V_{\lambda_j,\lambda_k,N}(\pole,\delta)+\frac{1}{2}W_{\lambda_j,\lambda_k,N}(\pole,\delta) \right\}
\right]+\o(1).
\end{eqnarray*}
Note that
\begin{eqnarray*}
&&\frac{1}{N} \sum_j
\left\{\frac{1}{2}\dot{B}_{\lambda_j,N}^2(\pole,\delta) +B_{\lambda_j,N}
(\pole,\delta)\dot{C}_{\lambda_j,N}(\pole,\delta)\right\}\\
&+&\frac{1}{2N}\sum_{k\neq j} \left[\dot{V}_{\lambda_j,\lambda_k,N}^2(\pole,\delta)
+V_{\lambda_j,\lambda_k,N}(\pole,\delta) \left\{5\pi^2V_{\lambda_j,\lambda_k,N}(\pole,\delta)+
\frac{1}{2} W_{\lambda_j,\lambda_k,N}(\pole,\delta) \right\} \right]
\end{eqnarray*}
equates to
\[
\frac{1}{N}\sum_{j}  \dot{C}_{\lambda_j,N}(\pole,\delta)+\o(1)
\]
We arrived at this result using the order of
$V_{\lambda_j,\lambda_k,N}\left(\pole,\delta\right)$,
$\dot{V}_{\lambda_j,\lambda_k,N}^2(\pole,\delta)$ and
$W_{\lambda_j,\lambda_k,N}\left(\pole,\delta\right)$ noted in eqn
(\ref{orderofthing}) and that: $ B_{\lambda_j,N}(\pole,\delta)=
1+\O\left\{\frac{\log(j)}{j}\right\},$
$V_{\lambda_j,\lambda_k,N}(\pole,\delta)=\O\left\{\frac{\log(j)}{k}\right\},$
if $\log(N)<k<j,$ $\dot{B}_{\lambda_j,N}(\pole,\delta) =
\O\left(\frac{1}{j}\right),$
$\dot{V}_{\lambda_j,\lambda_k,N}(\pole,\delta)=\O\left\{\frac{\log(j)}{k^2}\right\},$
and $\dot{C}_{\lambda_j,N}(\pole,\delta)=
\frac{2\pi^2}{3}+\O\left\{\frac{\log(j)}{j}\right\}.$

We thus have that
\begin{eqnarray}
\nonumber{\mathrm{var}}\left\{Y_{1,N}(\chistar)\right\}&=&
\frac{1}{N}{\mathrm{var}}\left[-\sum_j\Istddot_j\right]=
\frac{1}{N}\left[\sum_j \dot{C}_{\lambda_j,N}(\pole,\delta)+\o(N)\right]\\
&=& \frac{2\pi^2}{3N}\times \frac{N}{2}+\o(1)=\frac{\pi^2 }{3}+\o(1).
\label{variance}
\end{eqnarray}
Thus to obtain an $\O(1)$ random variate we must consider a standardization of
$N^{-3/2}\ell_{\pole}.$
\\ \\
PART II (\textit{Determining the asymptotic law}): In outline, we note:
\begin{eqnarray*}
E\left\{\Istddot_j \right\}&=&\O\left(\frac{1}{j} \right)+o(1)\quad
{\mathrm{var}}\left\{\Istddot_j \right\}=\frac{2\pi^2}{3}+\O\left\{\frac{\log(j)}{j}\right\}+o(1)\\
cov\left\{\Istddot_j, \Istddot_k\right\}&=& \O\left\{\frac{\log^2(j)}{k^2} \right\}+o(1),
\quad \log(N)<k<j.\end{eqnarray*}
Now we wish to derive conditional expectations, to be able to derive the
stated distributional result for $Z_4$.
Define for $\log(N)<k<j<N/2:$
\begin{eqnarray*}
{\bm{\Omega}}_j^{(f,N)}&=&\begin{pmatrix}
\frac{1}{2}+\O\left\{j^{-1} \log(j)\right\} & 0 & \frac{1}{4}+\O\left\{
j^{-1} \log(j)\right\}
& \O\left(j^{-1}\right)\\[4pt]
0 &\frac{1}{2}+\O\left\{j^{-1} \log(j)\right\} & \O\left(j^{-1}\right)
&\frac{1}{4}+\O\left\{j^{-1} \log(j)\right\}\\[4pt]
\frac{1}{4}+\O\left\{j^{-1} \log(j)\right\}&
\O\left(j^{-1}\right) &
\frac{1}{6}+\O\left\{j^{-1} \log(j)\right\} & 0\\[4pt]
\O\left(j^{-1}\right) & \frac{1}{4}+\O\left\{j^{-1} \log(j)\right\}
& 0 &\frac{1}{6}+\O\left\{j^{-1} \log(j)\right\}
\end{pmatrix}\\[6pt]
{\bm{\Omega}}_{jk}^{(f,N)}&=&\begin{pmatrix}
\O\left\{k^{-1} \log(j)\right\} & 0 & \O\left\{k^{-1} \log(j)\right\}
& \O\left\{k^{-2} \log(j)\right\}\\[4pt]
0 &\O\left\{k^{-1} \log(j)\right\} & \O\left\{k^{-2} \log(j)\right\}
&\O\left\{k^{-1} \log(j)\right\}\\[4pt]
\O\left\{k^{-1} \log(j)\right\}& \O\left\{k^{-2} \log(j)\right\} &
\O\left\{k^{-1} \log(j)\right\} & 0\\[4pt]
\O\left\{k^{-2} \log(j)\right\} & \O\left\{k^{-1} \log(j)\right\}
& 0 &\O\left\{k^{-1} \log(j)\right\}
\end{pmatrix}
\end{eqnarray*}
Then the full covariance matrix of $\left\{ \bm{V}_j^{(f,N)}\quad \bm{V}_k^{(f,N)}\right\}$ is given by
$\bm{\Sigma}_{jk}=\begin{pmatrix} {\bm{\Omega}}_{j}^{(f,N)} &
{\bm{\Omega}}_{jk}^{(f,N)}\\
{\bm{\Omega}}_{kj}^{(f,N)} & {\bm{\Omega}}_{k}^{(f,N)}\end{pmatrix}+o(1),$ and if we
define
\[
\bm{\Upsilon}_{jk}=\left\{ {\bm{\Omega}}_{k}^{(f,N)}-
 {\bm{\Omega}}_{kj}^{(f,N)}({\bm{\Omega}}_{j}^{(f,N)})^{-1} {\bm{\Omega}}_{jk}^{(f,N)}
\right\}^{-1}=\O(1)+\O\left\{k^{-1}\log(k)\right\}+\O\left(k^{-2}\log^2(j) \right),\]
then
\begin{eqnarray*}
\bm{\Sigma}^{-1}_{jk}=\begin{pmatrix}\bm{\Xi}_{j} & \bm{\Xi}_{jk}\\
\bm{\Xi}_{kj} & \bm{\Xi}_{k}\end{pmatrix},\quad
\bm{\Xi}_{j}&=& ({\bm{\Omega}}_{j}^{(f,N)})^{-1} +({\bm{\Omega}}_{j}^{(f,N)})^{-1} {\bm{\Omega}}_{jk}^{(f,N)}
\bm{\Upsilon}_{jk}{\bm{\Omega}}_{kj}^{(f,N)}({\bm{\Omega}}_{j}^{(f,N)})^{-1}.
\end{eqnarray*}
We may thus deduce that for $\log(N)<k<j<N/2$
\begin{eqnarray}
E\left\{\Istddot_j|\Istddot_k \right\}&=&\O\left(j^{-1} \right)+
\O\left\{ j^{-1} k^{-2} \log^2(j) \right\}+\o(1)
\\
{\mathrm{var}}\left\{\Istddot_j|\Istddot_k \right\}&=&
\frac{2\pi^2}{3}+\O\left\{j^{-1} \log(j) \right\}+\O\left\{k^{-2} \log^2(j) \right\}+\o(1).
\end{eqnarray}
These results are reminiscent of results obtained for the periodogram itself,
thus using arguments in the vein of \cite{Hurvich1998}; we argue that for $j$
sufficiently small the sum of the terms over $j$ are of negligible magnitude so
that when they are standardized by $N^{-1/2},$ they decay.

In fact if we define $U_j=\Istddot_j$ and calculate the characteristic
function of $\sum_{\left|j\right|=l}^{J} U_j,$ denoted $\phi(t),$ with $l=O\{\log(N)\}$ then
\begin{eqnarray}
\nonumber
\log(\phi(t))&=&\log\left\{E\left(e^{i\frac{t}{\sqrt{N}}\sum_j U_j}\right)\right\}
= \log\left[E\left\{e^{i\frac{t}{\sqrt{N}}\sum_j^{J-1} U_j}E\left(\left.e^{i \frac{t}{\sqrt{N}}
U_J}\right|U_{J-1}\dots\right)\right\}\right]\\
\label{replacecov}
&=&\log\left[E\left\{e^{i\frac{t}{\sqrt{N}}\sum_j^{J-1} U_j}
\left(1+i\frac{t}{\sqrt{N}}\left[
\O\left(\frac{1}{J}\right)+\sum_{k<J} \O\left\{\frac{\log^2(J)}{J k^2}\right\}\right]\right.\right.\right.\\
\nonumber
&&\left.\left.\left.-\frac{1}{2}\frac{t^2}{N}
\left[\frac{2\pi^2}{3}+\sum_{k<J} \O\left\{\frac{\log^2(J)}{k^2} \right\}\right]
 +\O\left(N^{-3/2}\right)\right)\right\}\right]
\\
\nonumber
&=&\sum_j \log\left(1+i\frac{t}{\sqrt{N}}\left[
\O\left(\frac{1}{j}\right)+\sum_{k<j} \O\left\{\frac{\log^2(j)}{j k^2}\right\}\right]-\frac{1}{2}\frac{t^2}{N}
\left[\frac{2\pi^2}{3}+\sum_{k<j} \O\left\{\frac{\log^2(j)}{k^2}\right\} \right]
\right.\\
\nonumber
&&\left. +\O\left(N^{-3/2}\right)\right)\\
\nonumber
&=&\sum_j\left(\frac{i t}{\sqrt{N}}\left[
\O\left(\frac{1}{j}\right)+\sum_{k<j} \O\left\{\frac{\log^2(j)}{j k^2}\right\}\right]-\frac{1}{2}\frac{t^2}{N}
\left[\frac{2\pi^2}{3}+\sum_{k<j} \O\left\{\frac{\log^2(j)}{k^2}\right\} \right]
+\O\left(N^{-3/2}\right)\right)\\
\nonumber
&=&\frac{i t}{\sqrt{N}}\left[
\O\{\log(N)\}+\o(\sqrt{N})\right]-\frac{1}{2}\frac{t^2}{N}
\left\{\left(J_2-J_1+1-2l\right)\frac{2\pi^2}{3}
+\o(N)\right\}+\O\left(N^{-1/2}\right)\\
\nonumber
&\rightarrow &-\frac{1}{2}\frac{2\pi^2 t^2}{3}(J_2-J_1+1-2l)=-\frac{1}{2}\frac{\pi^2 t^2}{3}.
\label{normalscore}
\end{eqnarray}
We want the characteristic function of $N^{-1/2}\sum_{j=J_1}^{J_2} U_j.$ We
split this into two parts $N^{-1/2}\sum_{\left|j\right|=l}^{J_1,J_2} U_j$ and
$N^{-1/2}\sum_{\left|j\right|<l} U_j,$ and note that the latter sum converges
to the point zero. Thus the sum of the $\Istddot_j$ will converge to a Gaussian
random variable with a zero mean and a variance of $\frac{\pi^2}{3}: $ or:
\begin{eqnarray}
\label{x1} k_{N,1}(\chistar)&=&K_N+\o(1) \overset{\cal L}{\rightarrow}
\mathcal{N}\left(0,\frac{1}{2}\frac{2\pi^2}{3}\right).\end{eqnarray} In fact,
stopping the argument at equation (\ref{replacecov}) and replacing $2\pi^2/3$
by $\dot{C}_{\lambda_j,N}(\pole,\delta)$, we may deduce that
\begin{eqnarray}
k_{N,1}(\chistar)&=&K_N+\o(1) \sim \mathcal{AN}\left(0,\sum_{J_1}^{J_2} \dot{C}_{\lambda_j,N}(\pole,\delta)\right).
\label{largefinitek}
\end{eqnarray}
The approximation of eqn. (\ref{largefinitek}) may serve as a better
approximation to the distribution of $k_{N,1}(\chistar)$, rather than the
distribution given in eqn (\ref{x1}), at moderate values of $N.$
\end{proof}

\subsubsection{Limit behaviour of the Fisher Information}
\label{FisherLim} Having established the large sample properties of $k_{N,1}$
to be able to relate them back to a suitably standardized version of
$\widehat{\pole}$ we must also establish the large sample behaviour of
$\left[\bm{W}_N\right]_{11}$ near the true value of the pole. We shall use the
same normalizations and local regions as defined by \cite{Sweeting2}, when
treating asymptotic ancillarity.  Recall that ${\bm{B}}_N$ was defined in
equation (\ref{defofBN}), and refer to the notation specified in this section.
To be able to do so define the ${\bm{B}}_N^{-1/2}$ neighbourhoods of $\chival$
by ${\mathcal{N}}_N\left(\chistar,c\right)=\left\{\chival\in\bm{\Omega}:
\;\left|{\bm{B}}_N\left(\bm{\chival}-\chistar\right) \right|<c \right\}$.
\begin{proposition}
\label{Fisherinfoconv} Define
$\phi_N^s=\left\{\chival:\;\chival=\chistar+\bm{B}_N^{-1/2}s,\;s\in{\mathbb{R}}^2
\right\}$.  For $\chival\in \phi_N^s$,
\item $\bm{W}_N(\chival) \overset{\mathcal{L}}{\rightarrow}\bm{W},$
where
\begin{equation}
\bm{W}\sim \begin{pmatrix} W_{11} & 0\\
0 & 1.
\end{pmatrix},
\end{equation}
and $W_{11} \sim N\left(0,\frac{8\pi^4}{15}\right).$
Furthermore note the finite large sample approximation that for $\chival=\chistar$ we have
{\small
\begin{eqnarray}
\label{modNfish}
\left[\bm{W}_N(\chistar)\right]_{11}&=&\tilde{W}_{N,11}+\o(1),\quad
\tilde{W}_{N,11}\sim \mathcal{N}
\left(\frac{\sum_j \ddot{B}_{\lambda_j,N}(\polestar,\delta^{\ast})}{\sqrt{N}},
\frac{\sum_j\tilde{\sigma}_j(\chistar)}{N}\right)\\
\tilde{W}_{N,11}&\overset{\cal{L}}{\Longrightarrow}&
Z_5,\quad Z_5\sim \mathcal{N}\left(0,\frac{8\pi^4}{15}\right),
\label{largeNfish}
\end{eqnarray}}
where
\begin{alignat}{1}
\nonumber
\tilde{\sigma}_j(\chistar)&={\mathrm{var}}\left(\Istdddot_j \right),\quad
\lim_{j\rightarrow N}\tilde{\sigma}_j(\chistar)=\frac{16\pi^4}{15}.
\end{alignat}
\end{proposition}

\begin{proof}
{\em Distribution of $\left[\bm{W}_N\left(\chistar \right)\right]_{11}.$} \\
We seek to establish the distribution of $\bm{W}_N\left(\chival \right),$ but
intend to start by determining the distribution of $\bm{W}_N\left(\chistar
\right).$  Most of the entries in the matrix are easily established: we have
already specified the distribution of $\left[\bm{W}_N\left(\chival
\right)\right]_{22}$ and we may note that $-\ell_{\pole,\delta}$ when
standardized by $N^{7/4},$ converges to zero (see Proposition 8). This implies
that three of the entries of $\bm{W}_N\left(\chival \right)$ appropriately
converge, and the fourth element needs to be determined, as well as note of the
correlation of the four elements need to be considered before the limit is
taken. We consider $\left[\bm{W}_N(\chistar)\right]_{11}$ for large sample
sizes. As
\[\left[\bm{W}_N(\chistar)\right]_{11}=-\frac{1}{ N^{5/2}}\ell_{\pole,\pole} =\frac{1}{N^{5/2}}
\sum_j\left\{S^{(2)}_j-\eta_j\tilde{S}_j^{(2)}\Pgrm_j +N
\breve{S}^{(2)}\frac{\eta_j}{N}\Idot_j +N^2 \eta_j \frac{\ddotPgramj}{N^2}.
\right\},\] and we note that $\Pgrm_j ,$ $\Idot_j $ and $\ddotPgramj $ are
quadratic forms in variables
$\tilde{\bm{V}}_j=\left[A_j,B_j,C_j,D_j,E_j,F_j\right]^\transpose,$  that are
more reasonably treated in terms of the standardized forms, we can note that:
\begin{eqnarray}
\nonumber
\left[\bm{W}_N(\chistar)\right]_{11}=-\frac{1}{N^{5/2}}\ell_{\pole,\pole}\left(\chistar\right)
&=&\frac{1}{N^{1/2}}\sum_j\left\{\frac{S^{(2)}_j}{N^2}-\frac{\tilde{S}_j^{(2)}\Istd_j}{N^2}
+\frac{\breve{S}^{(2)}\Istddot_j}{N}
+\Istdddot_j.
\right\}\\
&=&-\frac{1}{\sqrt{N}}\sum_{j=J_1}^{J_2}
\Istdddot_j+\o(1)=Y_{2,N}\left(\chistar\right)+\o(1).
\label{normscore6}
\end{eqnarray}
As for large $j,$ we note that
$\ddot{B}_{\lambda_j,N}(\pole,\delta)=\O(j^{-2}),$ and so we find that:
\[\lim_{N\rightarrow \infty} \sum_{j}  \ddot{B}_{\lambda_j,N}(\pole,\delta)\rightarrow C_{10}=\O(1),\]
and thus,
\begin{equation}
E\left\{ Y_{2,N}\left(\chistar\right)\right\}=-N^{-1/2}E\left\{\sum_{j=J_1}^{J_2}
\Istdddot_j\right\}=\O(N^{-1/2}).
\end{equation}
We then consider the variance of $Y_{2,N}\left(\chistar\right),$ to determine
the properties of this random variable. To find the full properties of
$Y_{2,N}\left(\chistar\right)$ we note that it is a quadratic form in the full
set $\left\{\tilde{\bm{V}}_j\right\},$ and replicate our previous treatment of
$\left\{\bm{V}_j\right\}.$ It transpires, that the important properties to
establish, for a heuristic argument, is the mean and variance of the random
variates $\Istdddot_j.$ The variates are correlated across $j,$ but given the
weak correlation, this need not be accounted for, just like in the previous
arguments, the combined correlation once suitably renormalized converges to a
negligible contribution. After some very lengthy calculations that are not
replicated here, we obtain that the variance of $\Istdddot_j$ is given by:
\begin{eqnarray}
\label{varofIdd}
\tilde{\sigma}^2_j&=&{\mathrm{var}}\left\{\Istdddot_j \right\}\\
\nonumber
&=&2^6 \pi^4 {\mathrm{var}}\left\{\standDs_j-\standA_j
\standE_j -\standB_j \standF_j + \standCs_j
\right\}\\
\nonumber
&=&2^6 \pi^4\left[{\mathrm{var}}\left\{\standDs_j\right\}+{\mathrm{var}}\left\{\standA_j
\standE_j\right\}+
{\mathrm{var}}\left\{\standB_j \standF_j\right\}+
{\mathrm{var}}\left\{\standCs_j\right\}\right.\\
\nonumber
&&-2{\mathrm{cov}}\left\{\standDs_j,\standA_j
\standE_j\right\}-2{\mathrm{cov}}\left\{\standDs_j,
\standB_j \standF_j\right\}+2{\mathrm{cov}}\left\{\standDs_j,
\standCs_j\right\}
\\
\nonumber
&&+2{\mathrm{cov}}\left\{\standA_j
\standE_j,\standB_j \standF_j\right\}-2{\mathrm{cov}}\left\{\standA_j
\standE_j,\standCs_j\right\}
-2{\mathrm{cov}}\left\{\standB_j \standF_j,\standCs_j\right\}.
\nonumber
\end{eqnarray}
Each of these terms is given by
\begin{eqnarray*}
{\mathrm{var}}\left\{\standDs_j\right\}
&=&{\mathrm{var}}\left\{\standCs_j\right\}=\frac{1}{2^7\pi^4}
\left\{2\pi^2B_{\lambda_j,N}(\pole,\delta)+\dot{C}_{\lambda_j,N}(\pole,\delta)\right\}^2+\o(1)\\
&=& \frac{1}{2^7\pi^4}
\left(2\pi^2+2\pi^2/3\right)^2+\o(1)
=  \frac{1}{18}+\o(1).
\end{eqnarray*}
Also
\begin{eqnarray*}
\nonumber
{\mathrm{var}}\left\{\standA_j \standE_j\right\}&=&
{\mathrm{var}}\left\{\standB_j \standF_j\right\}\\
\nonumber
&=&\frac{1}{4}B_{\lambda_j,N}(\pole,\delta)\left\{
\frac{1}{16}B_{\lambda_j,N}(\pole,\delta)-
\frac{1}{16\pi^2}\ddot{B}_{\lambda_j,N}(\pole,\delta)+\frac{3}{16\pi^2}\dot{C}_{\lambda_j,N}(\pole,\delta)
\right.\\
\nonumber
&&\left. +\ddot{C}_{\lambda_j,N}(\pole,\delta)
\right\}
+\frac{1}{2^8\pi^4}
\left\{-2\pi^2B_{\lambda_j,N}-\dot{C}_{\lambda_j,N}
(\pole,\delta)+\ddot{B}_{\lambda_j,N}(\pole,\delta)\right\}^2+\o(1)\\
&= & \frac{1}{20}+\frac{1}{36}+\o(1),
\end{eqnarray*}
where $\ddot{C}_{\lambda_j,N}(\pole,\delta)$ is given by
\begin{eqnarray}
\nonumber
\ddot{C}_{\lambda_j,N}(\pole,\delta)&=&\frac{1}{16}\left\{\begin{array}{lcr}
B_{\lambda_j,N}(\pole,\delta)-2\int_{-\infty}^{\infty}
\left|\frac{j}{u}\right|^{2\delta}
\frac{\sin\{\pi(u-j)\}}{\pi(u-j)}\psi_{2}(j,u)\;du+
\int_{-\infty}^{\infty}\left|\frac{j}{u}\right|^{2\delta}
\psi_{2}^2(j,u)\;du &{\mathrm{if}} &j \neq 0\\
B_{0,D,N}(\pole,\delta)-2\int_{-\infty}^{\infty}
\left|u\right|^{-2\delta}
\frac{\sin(\pi u)}{\pi u}\psi_{2}(0,u)\;du+
\int_{-\infty}^{\infty}\left|u\right|^{-2\delta}
\psi_{2}^2(0,u)\;du &{\mathrm{if}} &j = 0
\end{array}\right.,\\
\label{laterC}
\end{eqnarray}
and $\psi_{2}(j,u) = 2\left[-\cos\{\pi(u-j)\}/\{\pi(u-j)\}^2
+\sin\{\pi(u-j)\}/\{\pi(u-j)\}^3 \right).$ Finally we note that
\begin{eqnarray*}
{\mathrm{cov}}\left\{\standDs_j,\standA_j
\standE_j\right\}&=&
{\mathrm{cov}}\left\{\standCs_j,\standB_j
\standF_j\right\}= -\frac{\dot{B}_{\lambda_j,N}}{2^5\pi^2}\left\{ \dot{B}_{\lambda_j,N}/2
+\ddot{C}_{\lambda_j,N}(\pole,\delta)\right\},
\end{eqnarray*}
plus $\o(1)$ terms where
\[\ddot{C}_{\lambda_j,N}(\pole,\delta)=
\int_{-\infty}^{\infty} \left| \frac{s}{j}\right|^{-2\delta} s^{-1}
\frac{\left[\sin\left\{\pi
(j-s)\right\}
-\cos\left\{\pi (j-s)\right\}\pi (j-s)\right]^2}
{2\left\{\pi (j-s)\right\}^4}\;ds.
\]
Also
\begin{eqnarray*}
{\mathrm{cov}}\left\{\standCs_j,\standA_j
\standE_j\right\}&=&
{\mathrm{cov}}\left\{\standDs_j,\standB_j
\standF_j\right\}\\
&=&\frac{B_{\lambda_j,N}(\pole,\delta)}{2^6\pi^2}
\left\{2\pi^2B_{\lambda_j,N}-\ddot{B}_{\lambda_j,N}(\pole,\delta)+3\dot{C}_{\lambda_j,N}(\pole,\delta)
\right\}\\
&&+\o(1)=  \frac{1}{16}+\o(1)\\
{\mathrm{cov}}\left\{\standDs_j,
\standDs_j\right\}&=&\o(1)\\
{\mathrm{cov}}\left\{\standA_j
\standE_j,\standB_j \standF_j\right\}&=&
\delta \dot{B}_{\lambda_j,N}(\pole,\delta)+\o(1) = \O\left(j^{-1}\right)+\o(1).
\end{eqnarray*}
Combining these results we find that as $j\rightarrow N$, and $N\rightarrow
\infty$, $\tilde{\sigma}_j^2\rightarrow \frac{16\pi^4}{15}\approx 104.$ Thus
for increasing $j$ the variance of $\ddot{I}^{(f,N)}(\chistar)$ tends to a
constant, again the covariance terms will behave like the covariance terms in
the score, and the mean is of negligible magnitude. We are thus adding many
identically distributed variates with order one variance, and the same weak
dependence as before. We can yet again adapt the arguments of
\cite{Hurvich1998}. The argument will necessarily become very complicated, as
we now need to consider a quadratic form in twelve Gaussian correlated
variates, and there is no real point in giving the exact details of the
argument.

The distribution may for non-negligible values of $\delta$ be slow to attain,
and so for large but more moderate $N$ we propose to use:
\begin{equation}
Y_{2,N}(\chistar)=\tilde{W}_{N,11}+\o(1),\quad \tilde{W}_{N,11}
\sim \mathcal{N}\left(\frac{1}{\sqrt{N}}\sum_j \ddot{B}_{\lambda_j,N}(\pole,\delta),
\frac{1}{N}\sum_j \tilde{\sigma}^2_j\right).
\end{equation}
For large $N$ we find $\sqrt{N}\sum_j \ddot{B}_{\lambda_j,N}(\pole,\delta)=
\o(1)$ whilst
\[
\frac{1}{N}\sum_j
\tilde{\sigma}^2_j=\frac{1}{N}\frac{16\pi^4}{15}\frac{N}{2}+\o(1)
=\frac{8\pi^4}{15}+\o(1),
\]
and so we may note that $\tilde{W}_{N,11}\overset{\cal
L}{\Longrightarrow}Z_5,\quad Z_5\sim \mathcal{N}(0,8\pi^4/15)$. We furthermore
note that as the variance increases linearly with $\left|J_2-J_1\right|$ the
distribution of the second derivative at values of $j$ near the pole eventually
becomes negligible in influence in the random variable
$\left[\bm{W}_N\left(\chival\right)\right]_{11},$ and thus the distributional
results will also hold for $\left[\bm{W}_N\left(\chival\right)\right]_{11}$
when $\chival\in \phi_N^{s},$ or
\begin{eqnarray}
\label{x2}
\left[W_N(\chival)\right]_{11}&\overset{\cal L}{=}
&\left[W_N(\chistar)\right]_{11}+\o(1)
\overset{\cal L}{=}Z_5+\o(1).
\end{eqnarray}
This establishes the distribution of the standardized observed Fisher information
of the likelihood.
\end{proof}

However, before we may combine these results to note the distribution of
$N\widehat{\pole}$ we must consider the dependence between
$N^{-1}\ell_{\pole}(\chistar)$ and $N^{-5/2}\ell_{\pole,\pole}(\chival),$
which, based on the argument of the distributional equivalence of
$\ell_{\pole,\pole}(\chival),$ and $\ell_{\pole,\pole}(\chistar),$ and the
asymptotic Gaussianity of the variables corresponds to bounding the covariance
of $N^{-1}\ell_{\pole}(\chistar)$ and $N^{-1}\ell_{\pole,\pole}(\chistar).$
\begin{proposition}
\label{scorefishercov} The restandardized score in $\pole$ and the
restandardized observed Fisher information in $\pole$ evaluated at $\chistar$
satisfy ${\mathrm{cov}}\left\{k_{N,1}(\chistar),
\left[W_N(\chistar)\right]_{11}\right\} = \o(1)$. We can thus deduce that as
$\left[W_N(\chival)\right]_{11}\overset{\cal L}{=}
\left[W_N(\chistar)\right]_{11}$ and asymptotic Gaussianity is valid,
asymptotic independence follows.
\end{proposition}
\begin{proof}
Due to previous arguments of large sample distributional equivalence, and due
to the asymptotic Gaussianity, we need only consider the covariance of
$Y_{1,N}(\chistar)$ and $Y_{2,N}(\chistar),$ and thus start by considering the
covariance of the elements that make up these objects. We note that
\begin{eqnarray*}
\tilde{c}_{k,j}&=&\mathrm{cov}\left\{\dot{I}^{{\tiny (f,N)}}_k,\ddot{I}^{(f,N)}_j\right\}\\
&=&(4\pi)(8\pi^2){\mathrm{cov}}\left\{\standB_k \standC_k-
\standA_k \standD_k,\standDs_j+ \standCs_j-\standA_j \standE_j -\standB_j \standF_j
\right\}.
\end{eqnarray*}
We consider the $j=k$ terms and show that their contribution decays suitably in
$j$: the cross terms will be bounded like in previous arguments, relying of the
decay for $\log(N)<k<j$. Then combining the results of the previous section
with Isserlis's theorem we find that (up to $\o(1)$):
\begin{eqnarray*}
\tilde{c}_{j,j}&=&\frac{1}{2}\dot{B}_{\lambda_j,N}(\pole,\delta)\Re(K_{j,j})+
4\pi^3\dot{B}_{\lambda_j,N}(\pole,\delta)B_{\lambda_j,N}(\pole,\delta)+\pi^2
\left\{\dot{B}_{\lambda_j,N}(\pole,\delta)-4\delta\ddot{C}^{(2)}_{\lambda_j,N}(\pole,\delta)
/\pi\right\}\\
&&+\frac{1}{2}\left\{2\pi^2B_{\lambda_j,N}(\pole,\delta)-\ddot{B}_{\lambda_j,N}(\pole,\delta)
-\frac{1}{2}\dot{C}_{\lambda_j,N}(\pole,\delta)\right\}\dot{B}_{\lambda_j,N}(\pole,\delta)\\
&&+\frac{1}{2}\dot{B}_{\lambda_j,N}(\pole,\delta)\Re(K_{j,j})
+\pi^2 B_{\lambda_j,N}(\pole,)
\left\{-\frac{1}{2}\dot{B}_{\lambda_j,N}(\pole,\delta)+2\delta\ddot{C}^{(2)}_{\lambda_j,N}
(\pole,\delta)/\pi\right\}
\\&&+ \frac{1}{2}\left\{2\pi^2B_{\lambda_j,N}(\pole,\delta)-\ddot{B}_{\lambda_j,N}(\pole,\delta)
\dot{C}_{\lambda_j,N}(\pole,\delta)\right\}-\frac{1}{2}\dot{B}_{\lambda_j,N}(\pole,\delta)
+4\pi^3\dot{B}_{\lambda_j,N}(\pole,\delta)B_{\lambda_j,N}(\pole,\delta).
\end{eqnarray*}
Thus we may deduce $\tilde{c}_{j,j}=\O(j^{-1})+\o(1)$.  The cross terms, {\em
i.e.} $\tilde{c}_{k,j},$ may be bounded in a standard fashion using the same
argument, so that
\begin{equation}
{\mathrm{cov}}\left\{k_{N,1}(\chistar), \left[W_{N}(\chistar)\right]_{11}\right\}=\o(1).
\end{equation}
We can thus deduce the asymptotic independence of variables $k_{N,1}(\chistar)$
and $\left[W_{N}(\chival)\right]_{11}.$
\end{proof}

\begin{proposition}
\label{mledistt}
The large sample distribution of the MLE of $\pole$ tends to:
\begin{equation}
N(\polehat-\polestar)=\frac{N^{5/2}}{-\ell_{\pole,\pole}(\chiprime)}N^{-3/2}\ell_{\pole}(\chistar)
\rightarrow \frac{\sqrt{5}}{2\pi\sqrt{2}}C,
\end{equation}
where $C\sim Cauchy.$
\end{proposition}
\begin{proof}
To show this result we can simply use Propositions \ref{dist8}, \ref{Fisherinfoconv}
and \ref{scorefishercov}.

\vspace{0.1 in}

\noindent \emph{Note on Usage of Asymptotic Form}: We have
\begin{equation}
\label{polehat2}
C_N=N(\polehat-\polestar)=k_{N,1}(\chistar)/\left[W_N(\chival)\right]_{11}.
\end{equation}
We note that from equations (\ref{largeNscore}) and (\ref{largeNfish}), using proposition \ref{scorefishercov} that
\begin{eqnarray}
\frac{2^{3/2}\pi}{\sqrt{5}}C_N&=&
\frac{2^{3/2}\pi}{\sqrt{5}}
\frac{Z_4}{Z_5}+\o(1)=
\frac{Z_4/\sqrt{\pi^2/3}}{Z_5/\sqrt{8\pi^4/15}}\sim Cauchy.
\end{eqnarray}
Define $c_1=\tan(\pi(-\frac{1}{2}+0.025))$ and $c_2=\tan(\pi(-\frac{1}{2}+0.975)),$
then
\begin{eqnarray*}
P\left(c_1\le \frac{2^{3/2}\pi}{\sqrt{5}}C_N \le
c_2\right) & = &
P\left(\frac{\sqrt{5}}{2\pi\sqrt{2}}c_1\le C_N\le
\frac{\sqrt{5}}{2\pi\sqrt{2}}c_2\right) = 0.95 \\
\therefore P\left(\widehat{\pole}+\frac{\sqrt{5}}{2N\pi\sqrt{2}}c_1\le \pole\le
\widehat{\pole}+\frac{\sqrt{5}}{2N\pi\sqrt{2}}c_2\right)&=&0.95.
\end{eqnarray*}
Thus a $95\%$ CI is given for $\pole$ by
$(\widehat{\pole}-3.20/N,\widehat{\pole}+ 3.20/N).$ This establishes the large
sample theory for $\widehat{\pole}.$ However, the effect on MLE of low $j$
contributions decays slowly, and so we provide an additional approximation to
the distribution, based on equations (\ref{modNscore}) as well as
(\ref{modNfish}).

\vspace{0.1 in}

\noindent\emph{Note on Usage of Large Sample Approximation Form} For finite
$N,$ as already discussed, it may be more appropriate to approximate the
distribution of the two random variables using $K_N \sim
\mathcal{N}(\mu_1,\sigma^2_1)$ and $\tilde{W}_{N,11} \sim
\mathcal{N}(\mu_2,\sigma^2_2)$ where
\begin{eqnarray*}
\mu_1=\frac{1}{\sqrt{N}}\sum_{j=J_1}^{J_2} \dot{B}_{\lambda_j,N}(\pole,\delta)= \o(1) && \sigma^2_1=
\frac{1}{N}\sum_{j=J_1}^{J_2}\left\{\frac{1}{2}\delta^2\dot{B}_{\lambda_j,N}^2(\pole,\delta)+
B_{\lambda_j,N}(\pole,\delta)\dot{C}_{\lambda_j,N}(\pole,\delta) \right\} + \o(1)\\
\mu_2=\frac{1}{\sqrt{N}}\sum_{j=J_1}^{J_2}
\ddot{B}_{\lambda_j,N}(\pole,\delta)+\o(1) && \sigma^2_2=
\frac{1}{N}\sum_{j=J_1}^{J_2}\tilde{\sigma}_j^2+\o(1),
\end{eqnarray*}
where $\tilde{\sigma}_j^2$ is given by equation (\ref{varofIdd}). With these
quantities, we have $C_1=K_N/\tilde{W}_{N,11}$, $C_2=\tilde{W}_{N,11}$,
$\tilde{W}_{N,11}=C_2$ and $K_N=C_1 C_2$, and find a confidence interval for
$C_1,$ $\mathrm{Pr}\left(c_{11}<C_1<c_{12}\right)=1-\alpha$, assuming that
asymptotic independence of $K_N$ and $\tilde{W}_{N,11}$ is approximately
attained, we have by transformation techniques
\begin{eqnarray*}
f_{C_1,C_2}(c_1,c_2)=
\frac{1}{2\pi \sigma_1 \sigma_2}e^{-\frac{1}{2}
\left\{\frac{c_1^2 c_2^2}{\sigma^2_1}+\frac{(c_2-\mu_2)^2}{\sigma^2_2}\right\}}
\left|c_2\right| &\therefore&
f_{C_1}(c_1) = \int_{-\infty}^{\infty}
\frac{1}{2\pi \sigma_1 \sigma_2}e^{-\frac{1}{2}
\left\{\frac{c_1^2 c_2^2}{\sigma^2_1}+\frac{(c_2-\mu_2)^2}{\sigma^2_2}\right\}}\left|c_2\right|\;dc_2\\
\int_{-\infty}^{c_{11}}f_{C_1}(c_1)\;dc = \alpha/2&&
\int_{-\infty}^{c_{12}}f_{C_1}(c_1)\;dc=1-\alpha/2.
\end{eqnarray*}
Thus, once $\mu_2,$ $\sigma_1^2$ and $\sigma_2^2$ have been determined by
calculating the integrals we can derive the approximation to the distribution
of the estimator of $\pole.$ In fact, with $u(c_1)=\sigma_1^2+c_1^2\sigma^2_2$,
\begin{eqnarray}
\nonumber
f_{C_1}(c_1)&=&\int_{-\infty}^{\infty}
\frac{1}{2\pi \sigma_1 \sigma_2}e^{-\frac{1}{2}
\left\{\frac{c_1^2 c_2^2}{\sigma^2_1}+\frac{(c_2-\mu_2)^2}{\sigma^2_2}\right\}}\left|c_2\right|\;dc_2\\
\nonumber
&=&\frac{1}{\sqrt{2\pi}}u(c_1)^{-3/2}
\left[\frac{\sqrt{2u(c_1)} \sigma_1 \sigma_2}{\sqrt{\pi }}e^{-
\frac{\mu_2^2}{2\sigma_2^2}}+
\sigma_1^2\mu_2 e^{-\frac{\mu_2^2c_1^2}{2u(c_1)}}
{\mathrm{erf}}\left\{
\mu_2\frac{\sigma_1}{\sqrt{2}\sigma_2\sqrt{u(c_1)}}\right\}\right],\;c_1\in{\mathbb{R}}
\label{c1pdf}\\
&\longrightarrow&
\frac{\sigma_1 \sigma_2}{\pi}u(c_1)^{-1},\;c_1\in{\mathbb{R}}
\end{eqnarray}
as $\mu_2\rightarrow 0$, and the distribution becomes a scaled Cauchy
distribution.
\begin{table}
\caption{The quantities necessary to approximate the distribution
of $N\polehat$ using $C_1.$ \label{table1}}
\begin{tabular}{|c|r|r|r|r|r|} \hline
N & $\delta$ & $95 \%$ interval & $\mu_2$ & $\sigma_1^2$ & $\sigma_2^2$\\
\hline
1024 & 0.30 &$\widehat{\xi}\pm 3.17 N^{-1}$ &  0.7778 & 3.3413 & 52.9845
\\
\hline
1024 & 0.40 & $\widehat{\xi}\pm 2.90 N^{-1}$&  3.2203 & 3.4503 & 54.9996\\
\hline
1024 & 0.45 &$\widehat{\xi}\pm  1.43 N^{-1}$ & 9.6724 & 3.6872 &  56.5742\\
\hline
2048 & 0.30 &$\widehat{\xi}\pm  3.18 N^{-1}$ &  0.5606 & 3.3180 & 52.5227
\\  \hline
2048 & 0.40 &$\widehat{\xi}\pm  3.00 N^{-1}$ &  2.3937 & 3.3779 &  53.6337
\\ \hline
2048 & 0.45 &$\widehat{\xi}\pm 1.96 N^{-1}$ & 7.3754 & 3.5085&   54.6138\\
\hline
4096 & 0.30 &$\widehat{\xi}\pm  3.19 N^{-1}$ &  0.4021 & 3.3051 &  52.2645
\\ \hline
4096 & 0.40 &$\widehat{\xi}\pm 3.10  N^{-1}$ &  1.7646 & 3.3377 &  52.8721  \\  \hline
4096 & 0.45 &$\widehat{\xi}\pm 2.41 N^{-1}$ &  5.5698 & 3.4091 &   53.4590\\
\hline
8192 & 0.30 &$\widehat{\xi}\pm 3.20  N^{-1}$ &  0.2874 & 3.2981 &  52.1217
\\ \hline
8192 & 0.40 &$\widehat{\xi}\pm 3.14 N^{-1}$ &  1.2921 & 3.3157 &  52.4517  \\  \hline
8192 & 0.45 &$\widehat{\xi}\pm 2.41 N^{-1}$ &  4.1725 & 3.3544 &   52.7938\\ \hline
$\infty$& $\delta>0$ & $\widehat{\xi}\pm 3.20 N^{-1}$& 0 &  3.2899 & 51.9515\\ \hline
\end{tabular}
\end{table}
Using this approximation, we may derive CIs for $\pole$ for a given value
$\delta$ by determining $c_{11}$ and $c_{12}$ for that value of $\delta$ from
\begin{equation}
\nonumber
P\left(c_{11}<N\left(\widehat{\pole}-\pole\right)<c_{12}\right)=
P\left(\widehat{\pole}- c_{12}/N<\pole<\widehat{\pole}- c_{11}/N\right) = 1-\alpha.
\label{poleci}
\end{equation}

\vspace{0.1 in}

\noindent\emph{Long Memory Parameter dependence of the CI's} The $\delta$
dependence is implicit in the distribution of $C_1$ in equation (\ref{c1pdf}),
as $\mu_2,$ $\sigma^2_1$ and $\sigma^2_2$ depend on $\delta.$ Thus a
$(1-\alpha)$ CI is simply given by $\widehat{\pole} \pm c_{12}/N$. For a real
data set, we do not know the true value of $\delta,$ but note that
$\widehat{\delta} = \delta^{\ast} + Z_2/\sqrt{N \fisherNlme_{\delta \delta}}$
where $Z_2\sim \mathcal{N}(0,1),$ from equation (\ref{deltalim}), as the same
central limit argument will be valid for the score evaluated at $\delta$ lying
between $\delta^{\ast}$ and $\widehat{\delta}.$ We note that $c_{11}$ and
$c_{12}$ are smooth functions of $\delta,$. Making the dependence on $\delta$
explicit we find
\begin{eqnarray}
\left|c_{1k}(\delta^{\ast})-c_{1k}(\widehat{\delta)}\right|=N^{-1/2}
\left|c^{\prime}_{1k}(\delta^{\ast})\right|\left|Z_2\right|,
\end{eqnarray}
and so as $N^{-1/2}
\left|c^{\prime}_{1k}(\delta^{\ast})\right|\left|Z_2\right|\overset{P}{\rightarrow}0$
we can use equation (\ref{poleci}) with $c_{11}$ and $c_{12}$ calculated at
$\delta=\widehat{\delta}.$ For our simulation study, to reduce the numerical
burden of the procedure, we have calculated the CIs at $\delta^{\ast}.$ This
would not be the approach in a real problem, but given the reduced
computational cost of a single calculation of $c_{11}$ and $c_{12}$ for real
examples, this is not an issue.
\end{proof}

Finally, we establish that the score in $\delta$ and $\pole$ are uncorrelated,
as the off-diagonal terms of the standardized observed Fisher information
converge to zero.
\begin{proposition}
\label{crosscovariance} We have that ${\mathrm{cov}}\left\{k_{N,1}(\chistar),
k_{N,2}(\chistar)\right\}=\o(1)$, and thus we can note that the distributional
results follow.
\end{proposition}
\begin{proof} Note that $k_{\pole,N}(\chistar) =  N^{-3/2}
{l}_{\pole}(\chistar)$ and $k_{\delta,N}(\chistar) =
\fisherNlme_{\delta,\delta}^{-1/2} N^{-1/2} {l}_{\delta}(\chistar)$. We thus
consider
$N^{-2}\mathrm{cov}\left\{{l}_{\pole}(\chistar),{l}_{\delta}(\chistar)\right\}$.
We have
\[
{\mathrm{cov}}
\left\{k_{\pole,N}(\chistar),k_{\delta,N}(\chistar) \right\}
=\frac{1}{\sqrt{\fisherNlme_{\delta,\delta}}N^{2}}{\mathrm{cov}}\left\{\sum_j \left[S_j^{(1)}\left\{1-\Istd_j\right\}-
\Istddot_j\right],\sum_j R_j^{(1)}\left\{1-\Istd_j\right\}
\right\},
\]
plus $\o(1)$ terms.
Thus it follows
\begin{eqnarray*}
{\mathrm{cov}}
\left\{k_{\pole,N}(\chistar),k_{\delta,N}(\chistar) \right\}
&=&\frac{1}{\sqrt{\fisherNlme_{\delta,\delta}} N^{2}}{\mathrm{cov}}\left\{-\sum_j S_j^{(1)}\Istd_j-\sum_j\Istddot_k,-\sum_j R_k^{(1)}\Istd_k
\right\}+\o(1)\\
&=&\frac{1}{\sqrt{\fisherNlme_{\delta,\delta}} N^{2}}\left[
\sum_j\sum_k {\mathrm{cov}}\left\{S_j^{(1)}\Istd_j, R_k^{(1)}\Istd_k
\right\}\right.  \\
&&\left.+\sum_j\sum_k {\mathrm{cov}}\left\{\Istddot_j, R_k^{(1)}\Istd_k
\right\}
\right]+\o(1)\\
&=&\frac{1}{\sqrt{\fisherNlme_{\delta,\delta}} N^{2}}\left[
2\sum_j\sum_{k\le j} S_j^{(1)} R_k^{(1)}\O\left\{k^{-2}\log^2(j)\right\}
\right.  \\
&&\left.+\sum_j\sum_k R_k^{(1)}  {\mathrm{cov}}\left\{\Istddot_j, \Istd_k
\right\}
\right]+\o(1)\\
&=&\frac{1}{\sqrt{\fisherNlme_{\delta,\delta}} N^{2}}\sum_j\sum_k R_k^{(1)}  {\mathrm{cov}}\left\{\Istddot_j, \Istd_k
\right\}+\o(1)
\end{eqnarray*}
Note that using Isserlis's theorem \citep{isserlis} we have:
\begin{eqnarray*}
{\mathrm{cov}}\left\{\Istd_j,\Istddot_k\right\}&=&4\pi
{\mathrm{cov}}\left\{\standAs_{j}+\standBs_{j},\standB_{k}
\standC_{k}-\standA_{k}\standD_{k} \right\}.
\end{eqnarray*}
For $j=k$ we have
\begin{eqnarray*}
{\mathrm{cov}}\left\{\Istd_j,\Istddot_j\right\}&=&4\pi
\frac{\dot{B}_{\lambda_j,N}(\pole,\delta)}{4\pi} \frac{1}{2}B_{\lambda_j,N}(\pole,\delta)
+\o(1)=\O(j^{-1})+\o(1)\rightarrow 0,
\end{eqnarray*}
for increasing $j$ and $N$. The cross-terms ${\mathrm{cov}}\left\{I_j^{(f,N)},\dot{I}_k^{(f,N)}\right\}$
may be bounded in the usual fashion.
\[
{\mathrm{var}}\{I_j^{(f,N)}\}=\O(1) \qquad \qquad {\mathrm{var}}\{\dot{I}_j^{(f,N)}\}=\O(1).
\]
Combining these results we find that
\[\lim_{N\rightarrow \infty}
\left[(\fisherNlme_{\delta,\delta})^{-1/2} N^{-2}{\mathrm{cov}}\left\{\ell_{\pole}(\chistar),
\ell_{\delta}(\chistar) \right\}\right]=0.\]
\end{proof}

%% file: asymptotics-arxiv.bbl
\begin{thebibliography}{}

\bibitem[\protect\citeauthoryear{Abramovich, Benjamini, Donoho, and
  Johnstone}{Abramovich et~al.}{2006}]{Abramovich+2006}
Abramovich, F., Benjamini, Y., Donoho, D.~L., and Johnstone, I.~M. (2006).
\newblock Adapting to unknown sparsity by controlling the false discovery rate.
\newblock {\em Ann. Statist.\/},~{\bf 34}, 584--653.

\bibitem[\protect\citeauthoryear{Andel}{Andel}{1986}]{Andel1986}
Andel, J. (1986).
\newblock Long memory time series models.
\newblock {\em Kybernetika\/},~{\bf 22}, 105--23.

\bibitem[\protect\citeauthoryear{Beran}{Beran}{1994}]{Beran1994}
Beran, J. (1994).
\newblock {\em Statistics for Long-Memory Processes}.
\newblock Chapman and Hall, London.

\bibitem[\protect\citeauthoryear{Beran and Gosh}{Beran and Gosh}{2000}]{Beran}
Beran, J. and Gosh, S. (2000).
\newblock Estimation of the dominating frequency for stationary and
  nonstationary fractional autoregressive models.
\newblock {\em J. Time Ser. Anal.\/},~{\bf 21}, 517--533.

\bibitem[\protect\citeauthoryear{Berger, Liseo, and Wolpert}{Berger
  et~al.}{1999}]{Berger}
Berger, J.~O., Liseo, B., and Wolpert, R.~L. (1999).
\newblock Integrated likelihood methods for eliminating nuisance parameters.
\newblock {\em Statistical Science\/},~{\bf 14}, 1--28.

\bibitem[\protect\citeauthoryear{Brillinger}{Brillinger}{1975}]{brillinger}
Brillinger, D. (1975).
\newblock {\em Time Series, Data Analysis and Theory}.
\newblock New York, USA: Holt, Rhinehart and Winston.

\bibitem[\protect\citeauthoryear{Cand{\`e}s and Tao}{Cand{\`e}s and
  Tao}{2004}]{CandesTao}
Cand{\`e}s, E. and Tao, T. (2004).
\newblock Near optimal signal recovery from random projections: Universal
  encoding strategies.
\newblock Technical report, Caltech.

\bibitem[\protect\citeauthoryear{Chen, Wu, and Dahlhaus}{Chen
  et~al.}{2000}]{ChenWuDahlhaus}
Chen, Z.~G., Wu, K.~H., and Dahlhaus, R. (2000).
\newblock Hidden frequency estimation with data tapers.
\newblock {\em J. Time Ser. Anal.\/},~{\bf 21}, 113--142.

\bibitem[\protect\citeauthoryear{Cheng and Taylor}{Cheng and
  Taylor}{1995}]{cheng}
Cheng, R. C.~H. and Taylor, L. (1995).
\newblock Non-regular maximum likelihood problems.
\newblock {\em J. Roy. Statist. Soc. B\/},~{\bf 57}, 3--44.

\bibitem[\protect\citeauthoryear{Coifman and Donoho}{Coifman and
  Donoho}{1995}]{coif}
Coifman, R.~R. and Donoho, D.~L. (1995).
\newblock Translation-invariant denoising.
\newblock In A.~Antoniadis and G.~Oppenheim (Eds.), {\em Wavelets and
  Statistics (Lecture Notes in Statistics, Volume 103)}, pp.\  125--150. New
  York: USA: Springer-Verlag.

\bibitem[\protect\citeauthoryear{Contreras-Cristan, Gutierrez-Pena, and
  Walker}{Contreras-Cristan et~al.}{2006}]{Contreras}
Contreras-Cristan, A., Gutierrez-Pena, E., and Walker, S.~G. (2006).
\newblock A note on {W}hittle's likelihood.
\newblock {\em Comm. Stats. -- Sim. Comp.\/},~{\bf 35}, 857--875.

\bibitem[\protect\citeauthoryear{Coursol and Dacunha-Castelle}{Coursol and
  Dacunha-Castelle}{1982}]{coursol}
Coursol, J. and Dacunha-Castelle, R. (1982).
\newblock Remarks on the approximation on the likelihood function of a
  stationary {G}aussian process.
\newblock {\em Theory Prob. Appl\/},~{\bf 27}, 162--67.

\bibitem[\protect\citeauthoryear{Donoho}{Donoho}{2006}]{Donoho2006}
Donoho, D.~L. (2006).
\newblock Compressed sensing.
\newblock {\em IEEE Transactions on Information Theory\/},~{\bf 52},
  1289--1306.

\bibitem[\protect\citeauthoryear{Dzhamparidze and Yaglom}{Dzhamparidze and
  Yaglom}{1983}]{Dzham1983}
Dzhamparidze, K.~O. and Yaglom, A.~M. (1983).
\newblock Spectrum parameter estimation in time series analysis.
\newblock In P.~Krishnaiah (Ed.), {\em Developments in Statistics}, Volume~4,
  pp.\  1--181. New York: Academic Press.

\bibitem[\protect\citeauthoryear{Geweke and Porter-Hudak}{Geweke and
  Porter-Hudak}{1983}]{GPH1983}
Geweke, J. and Porter-Hudak, S. (1983).
\newblock The estimation and application of long memory time series models.
\newblock {\em J. Time Ser. Anal.\/}~{\em 4\/}(4), 221--238.

\bibitem[\protect\citeauthoryear{Gil-Alana}{Gil-Alana}{2002}]{Gil-Alana2002}
Gil-Alana, L.~A. (2002).
\newblock Seasonal long memory in the aggregate output.
\newblock {\em Econ. Lett.\/}~{\em 74\/}(3), 333--7.

\bibitem[\protect\citeauthoryear{Giraitis, Hidalgo, and Robinson}{Giraitis
  et~al.}{2001}]{Giraitis2001}
Giraitis, L., Hidalgo, J., and Robinson, P.~M. (2001).
\newblock {G}aussian estimation of parametric spectral density with unknown
  pole.
\newblock {\em Ann. Statist.\/},~{\bf 29}, 987--1023.

\bibitem[\protect\citeauthoryear{Gradshteyn, Ryzhik, and Jeffrey}{Gradshteyn
  et~al.}{1994}]{Gradshteyn}
Gradshteyn, I.~S., Ryzhik, I.~M., and Jeffrey, A. (1994).
\newblock {\em Table of Integrals, Series, and Products}.
\newblock New York: Academic Press.

\bibitem[\protect\citeauthoryear{Gray, Zhang, and Woodward}{Gray
  et~al.}{1989}]{Gray1989}
Gray, H.~L., Zhang, N.~F., and Woodward, W. (1989).
\newblock On generalized fractional processes.
\newblock {\em J. Time Ser. Anal.\/},~{\bf 10}, 233--57.

\bibitem[\protect\citeauthoryear{Grenander and Szeg\"{o}}{Grenander and
  Szeg\"{o}}{1984}]{Grenander1984}
Grenander, U. and Szeg\"{o}, G. (1984).
\newblock {\em Toeplitz Forms and their Applications\/} (2 ed.).
\newblock Chelsea Publishing Company, New York.

\bibitem[\protect\citeauthoryear{Hannan}{Hannan}{1973}]{Hannan1973}
Hannan, E.~J. (1973).
\newblock Estimation of frequency.
\newblock {\em J. Appl. Prob.\/},~{\bf 10}, 510--519.

\bibitem[\protect\citeauthoryear{Hannan}{Hannan}{1986}]{Hannan1986}
Hannan, E.~J. (1986).
\newblock A law of the iterated logarithm for an estimate of frequency.
\newblock {\em Stochastic Processes And Their Applications\/},~{\bf 22},
  103--109.

\bibitem[\protect\citeauthoryear{Hidalgo}{Hidalgo}{2005}]{Hidalgo2005}
Hidalgo, J. (2005).
\newblock Semiparametric estimation for stationary processes whose spectra have
  an unknown pole.
\newblock {\em Ann. Statist.\/}~{\em 33\/}(4), 1843--1889.

\bibitem[\protect\citeauthoryear{Hidalgo and Soulier}{Hidalgo and
  Soulier}{2004}]{HandS2004}
Hidalgo, J. and Soulier, P. (2004).
\newblock Estimation of the location and exponent of the spectral singularity
  of a long memory process.
\newblock {\em J. Time Ser. Anal.\/}~{\em 25\/}(1), 55--81.

\bibitem[\protect\citeauthoryear{Hosoya}{Hosoya}{1974}]{Hosoya1974}
Hosoya, Y. (1974).
\newblock {\em Estimation Problems on Stationary Time-Series Models}.
\newblock Ph.\ D. thesis, Yale University.

\bibitem[\protect\citeauthoryear{Huerta and West}{Huerta and
  West}{1999}]{HuertaWest1999}
Huerta, G. and West, M. (1999).
\newblock Priors and component structure in autoregressive time series models.
\newblock {\em J. Roy. Statist. Soc. B\/}~{\em 61\/}(4), 881--99.

\bibitem[\protect\citeauthoryear{Hurvich and Beltrao}{Hurvich and
  Beltrao}{1993}]{HB1993}
Hurvich, C.~M. and Beltrao, K. (1993).
\newblock Asymptotics for the low frequency ordinates of the periodogram of a
  long memory time series.
\newblock {\em J. Time Ser. Anal.\/},~{\bf 14}, 455--472.

\bibitem[\protect\citeauthoryear{Hurvich, Deo, and Brodsky}{Hurvich
  et~al.}{1998}]{Hurvich1998}
Hurvich, C.~M., Deo, R., and Brodsky, J. (1998).
\newblock The mean square error of {G}eweke and {P}orter-{H}udak's estimator of
  the memory parameter of a long-memory time series.
\newblock {\em J. Time Ser. Anal.\/},~{\bf 19}, 19--46.

\bibitem[\protect\citeauthoryear{Isserlis}{Isserlis}{1918}]{isserlis}
Isserlis, L. (1918).
\newblock On a formula for the product-moment coefficient of any order of a
  normal frequency distribution in any number of variables.
\newblock {\em Biometrika\/},~{\bf 12}, 134--139.

\bibitem[\protect\citeauthoryear{Johnson and Kotz}{Johnson and
  Kotz}{1970}]{JohnsonandKotz}
Johnson, N.~I. and Kotz, S. (1970).
\newblock {\em Continuous Univariate Distributions, Vol.~2.}
\newblock New York, USA: Wiley.

\bibitem[\protect\citeauthoryear{Lapsa}{Lapsa}{1997}]{Lapsa1997}
Lapsa, P. (1997).
\newblock Determination of {G}egenbauer-type random process models.
\newblock {\em Sig. Proc.\/},~{\bf 63}, 73--90.

\bibitem[\protect\citeauthoryear{Olhede, McCoy, and Stephens}{Olhede
  et~al.}{2004}]{OMS2004}
Olhede, S.~C., McCoy, E.~J., and Stephens, D.~A. (2004).
\newblock Large sample properties of the periodogram estimator of seasonally
  persistent processes.
\newblock {\em Biometrika\/}~{\em 91\/}(3), 613--628.

\bibitem[\protect\citeauthoryear{Ooms}{Ooms}{2001}]{Ooms2001}
Ooms, M. (2001).
\newblock A seasonal periodic long memory model for monthly river flows.
\newblock {\em Env. Modell. Soft.\/},~{\bf 16}, 559--69.

\bibitem[\protect\citeauthoryear{Pei and Ding}{Pei and
  Ding}{2004}]{PeiDing2004}
Pei, S.~C. and Ding, J.~J. (2004).
\newblock Generalized eigenvectors and fractionalization of offset dfts and
  dcts.
\newblock {\em IEEE Trans. Sig. Proc.\/},~{\bf 52}, 2032--2046.

\bibitem[\protect\citeauthoryear{Porter-Hudak}{Porter-Hudak}{1990}]{Porter-Hud%
ak1990}
Porter-Hudak, S. (1990).
\newblock An application of the seasonal fractionally differenced model to the
  monetary aggregates.
\newblock {\em J. Amer. Statist. Assoc.\/},~{\bf 85}, 338--344.

\bibitem[\protect\citeauthoryear{Rathouz, Satten, and Carroll}{Rathouz
  et~al.}{2002}]{Rathouz2002}
Rathouz, P.~J., Satten, G.~A., and Carroll, R.~J. (2002).
\newblock Semiparametric inference in matched case-control studies with missing
  covariate data.
\newblock {\em Biometrika\/}~{\em 89\/}(4), 905--916.

\bibitem[\protect\citeauthoryear{Robins, Rotnitsky, and Zhao}{Robins
  et~al.}{1994}]{Robins1994}
Robins, J.~M., Rotnitsky, A., and Zhao, L.~P. (1994).
\newblock Estimation of regression coefficients when some regressors are not
  always observed.
\newblock {\em J. Amer. Statist. Assoc.\/},~{\bf 89}, 846--866.

\bibitem[\protect\citeauthoryear{Robinson}{Robinson}{1995}]{Robinson1995a}
Robinson, P.~M. (1995).
\newblock Log-periodogram regression of time-series with long-range dependence.
\newblock {\em Ann. Statist.\/},~{\bf 23}, 1048--72.

\bibitem[\protect\citeauthoryear{Sweeting}{Sweeting}{1980}]{Sweeting}
Sweeting, T.~J. (1980).
\newblock Uniform asymptotic normality of the maximum likelihood estimator.
\newblock {\em Ann. Statist.\/},~{\bf 8}, 1375--81.

\bibitem[\protect\citeauthoryear{Sweeting}{Sweeting}{1992}]{Sweeting2}
Sweeting, T.~J. (1992).
\newblock Asymptotic ancillarity and conditional inference for stochastic
  processes.
\newblock {\em Ann. Statist.\/},~{\bf 20}, 580--589.

\bibitem[\protect\citeauthoryear{Taniguchi and Kakizawa}{Taniguchi and
  Kakizawa}{2000}]{Taniguchi2000}
Taniguchi, M. and Kakizawa, Y. (2000).
\newblock {\em Asymptotic Theory of Statistical Inference for Time Series}.
\newblock New York: Springer.

\bibitem[\protect\citeauthoryear{Thomson}{Thomson}{1990}]{Thomson1990}
Thomson, D.~J. (1990).
\newblock Time-series analysis of holocene climate data.
\newblock {\em Phil. Trans. Roy. Statist. Soc. Lond. A\/},~{\bf 330}, 601--616.

\bibitem[\protect\citeauthoryear{v.~Sachs}{v.~Sachs}{1993}]{VonSachs1993}
v.~Sachs, R. (1993).
\newblock Detecting periodic components in stationary time series by an
  improved non-parametric procedure.
\newblock In {\em Proceedings of the International Conference on Applications
  of Time Series in Astronomy and Meteorology}, pp.\  115--118. University of
  Padova.

\bibitem[\protect\citeauthoryear{Walker}{Walker}{1964}]{Walker}
Walker, A.~M. (1964).
\newblock Asymptotic properties of least squares estimates of the spectrum of a
  stationary non-deterministic time series.
\newblock {\em J. Austr. Math. Soc.\/},~{\bf 4}, 363--384.

\bibitem[\protect\citeauthoryear{Walker}{Walker}{1965}]{Walker2}
Walker, A.~M. (1965).
\newblock Some asymptotic results for the periodogram of a stationary time
  series.
\newblock {\em J. Austr. Math. Soc.\/},~{\bf 5}, 107--128.

\bibitem[\protect\citeauthoryear{Whitcher}{Whitcher}{2004}]{Whitcher2004}
Whitcher, B. (2004).
\newblock Wavelet-based estimation for seasonal long-memory processes.
\newblock {\em Technometrics\/}~{\em 46\/}(2), 225--238.

\bibitem[\protect\citeauthoryear{Whittle}{Whittle}{1951}]{Whittle1951}
Whittle, P. (1951).
\newblock {\em Prediction and regulation by linear least-square methods}.
\newblock Almquist \& Wicksell, Uppsala, Sweden.

\end{thebibliography}
